\documentclass[12pt,english,american]{article}

\usepackage{hyperref}

\usepackage{amsmath,amsfonts,amssymb,amsthm,mathrsfs}
 
\usepackage{enumerate}
\usepackage{bm}
\usepackage{bbm}
\usepackage{verbatim}
\usepackage[usenames,dvipsnames]{color}
\usepackage{graphicx}
\usepackage{textcomp}

\theoremstyle{plain}\newtheorem{theorem}{Theorem}[section]
\theoremstyle{plain}\newtheorem{lemma}[theorem]{Lemma}
\theoremstyle{plain}\newtheorem{corollary}[theorem]{Corollary}
\theoremstyle{plain}
\theoremstyle{plain}
\theoremstyle{definition}
\theoremstyle{remark}\newtheorem{remark}[theorem]{Remark}

\newcounter{remarks}
\newcommand{\re}{\operatorname{Re}}

\newcommand{\lno}{\Big\vert \!\!\, \Big \vert}
\newcommand{\rno}{\Big\vert \!\!\, \Big \vert}

\newcommand{\floor}[1]{\lfloor #1 \rfloor}

\newcommand{\lsp}{\big<}
\newcommand{\rsp}{\big>}
\newcommand{\no}{\vert \!\!\,  \vert}

\newcommand{\bno}{\Big \vert \!\!\, \Big\vert}
\newcommand{\I}{\ensuremath{\mathfrak{S}}}
\newcommand{\Vol}[1]{\frac{1}{L^{#1}}} 
\newcommand{\V}[1]{\ensuremath{\mathcal{V}_{#1}(N,\rho)}}
\newcommand{\T}[1]{\ensuremath{\text{\tiny{#1}}}}

\usepackage{enumerate}

\title{Free time evolution of a tracer particle coupled to a Fermi gas in the high-density limit}

\author{Maximilian Jeblick\footnote{Ludwig-Maximilians-Universit\"at, Mathematisches Institut, Theresienstr.\ 39, {80333} M\"unchen, Germany. Email: {\tt jeblick@math.lmu.de}},
David Mitrouskas\footnote{Ludwig-Maximilians-Universit\"at, Mathematisches Institut, Theresienstr.\ 39, {80333} M\"unchen, Germany. Email: {\tt dmitrous@math.lmu.de}},
S\"oren Petrat\footnote{Institute for Advanced Study, 1 Einstein Drive, Princeton, NJ 08540, USA. Email: {\tt spetrat@ias.edu}},
Peter Pickl\footnote{Ludwig-Maximilians-Universit\"at, Mathematisches Institut, Theresienstr.\ 39, {80333} M\"unchen, Germany. Email: {\tt pickl@math.lmu.de}}
}
\begin{document}

\maketitle

\begin{abstract}
The dynamics of a particle coupled to a dense and homogeneous ideal Fermi gas in two spatial dimensions is studied.\ We analyze the model for coupling parameter $g=1$ (i.e., not in the weak coupling regime), and prove closeness of the time evolution to an effective dynamics for large densities of the gas and for long time scales of the order of some power of the density.\ The effective dynamics is generated by the free Hamiltonian with a large but constant energy shift which is given at leading order by the spatially homogeneous mean field potential of the gas particles.\ Here, the mean field approximation turns out to be accurate although the fluctuations of the potential around its mean value can be arbitrarily large.\ Our result is in contrast to a dense bosonic gas in which the free motion of a tracer particle would be disturbed already on a very short time scale.\ The proof is based on the use of strong phase cancellations in the deviations of the microscopic dynamics from the mean field time evolution.
\end{abstract}

\section{Introduction and Main Result}

{In this article we consider the dynamics of a tracer particle interacting with a dense and homogeneous two-dimensional fermionic gas. In order to keep the analysis simple we neglect the interaction between the gas particles and focus only on the interaction between tracer particle $y$ and gas particles $x_1,\ldots,x_N$. The general model we wish to study is defined by the Hamiltonian
\begin{equation}\label{H_micro}
H = -\frac{1}{2m_y} \Delta_y - \sum_{i=1}^N \frac{1}{2m_{x}} \Delta_{x_i} + g\sum_{i=1}^N v(x_i-y),
\end{equation}
where $v \in C_0^{\infty}$ (the space of smooth functions with compact support), and $g>0$ is a coupling constant. The time evolution of the (N+1)-body wave function $\Psi_t \in \mathcal H_y\otimes \mathcal H_{N}  = L^2(\mathbb T^{d}) \otimes L^2(\mathbb T^{dN})$, where $d$ is the dimension, $\mathbb T$ a one-dimensional torus of length $L\in \mathbb R$, and $L^2$ denotes the space of complex square integrable functions (for simplicity, we neglect spin), is given by the Schr\"odinger equation
\begin{equation}\label{microscopic_model_2}
i\partial_t \Psi_t = H \Psi_t.
\end{equation}
As initial condition we choose a factorized state $\Psi_0 = \varphi_0 \otimes \Omega_0$, where $\varphi_0 \in \mathcal H_y$ is the initial wave function of the tracer particle and $\Omega_0 \in \mathcal H_N$ is the free fermionic ground state with periodic boundary conditions in the $d$-dimensional box of side length $L$. For analyzing $\Psi_t$ we first take the limit $N,L\to\infty$ with $\rho= N/L^d = const.$ in order to remove finite size effects and then consider large gas densities $\rho$. Note that in this situation the average potential energy of the tracer particle is proportional to $g\rho$. We later choose $g=1$, such that our analysis is beyond any weak-coupling limit.

We expect that the above model exhibits some interesting phenomena which depend in particular on the time scale that one considers.\ Here, we consider time scales for which the tracer particle moves in the mean field of the gas particles.\ Since the mean field potential is spatially homogeneous for the ideal Fermi gas, the effective dynamics is equivalent to the free time evolution.\ For longer times, we expect that the tracer particle will create electron-hole pairs and eventually lose its energy.\ The situation may differ depending also on the spatial dimension. For reasons which we explain in Section~\ref{sec:dimension}, we focus on the two dimensional case. Let us also remark that the described model is relevant, e.g., for understanding the motion of ions in a degenerate and dense electron plasma. In this situation it is known that the ability of the plasma to stop ions decreases in the high-density limit; cf.\ Section \ref{sec:applications}.

 We prove in this article that in the limit $\rho\to\infty$, the time evolution of the tracer particle is close to the free dynamics on a particularly large time scale.\ Our main result is readily stated:

\begin{theorem}\label{thm:main_thm} 
Let $d=2$, the masses $m_x=m_y=1/2$ and the coupling constant $g=1$. Let further $\varphi_0 \in  \mathcal H_y$ with $\no \nabla^4 \varphi_0\no \le C$ uniformly in $\rho$.\ Then, for any small enough $\varepsilon>0$, there exists a positive constant $C_\varepsilon$ such that
\end{theorem}
\vspace{-0.2cm}
\begin{equation}\label{main_estimate}
\lim_{\substack{N,L\to\infty \\ \rho=N/L^2=const.}} \left\| e^{-iHt} \Psi_0 - e^{-i H^\text{mf} t }  \Psi_0 \right\|_{\mathcal H_y\otimes \mathcal H_{N}} \leq C_\varepsilon (1+t)^{\frac{3}{2}} \rho^{ -\frac{1}{8}+\varepsilon}
\end{equation}

\vspace{0.3cm}
\noindent\textit{holds for all $t>0$, where}
\vspace{-0.2cm}
\begin{equation}
\label{Mean-field Hamiltonian} H^{\text{mf}} = - \Delta_y -  \sum_{i=1}^N \Delta_{x_i} +  \rho \mathcal F[v](0) - E_{re}(\rho)
\end{equation}

\vspace{-0.07cm}
\noindent \textit{is the free Hamiltonian with constant mean field $ \rho \mathcal F[v](0) = \langle \Omega_0, \sum_{i=1}^N v(x_i-y) \Omega_0 \rangle_{\mathcal H_N}$ minus a positive $\rho$-dependent next-to-leading order energy correction $E_{re}(\rho)$ which is defined in \eqref{def:F}.
}
\\

Let us remark that in Theorem \ref{thm:main_thm}, we have fixed all scales except for the density $\rho$ and the time $t$.\ The statement is meaningful for all pairs of $\rho$ and $t$ for which the r.h.s.\ of \eqref{main_estimate} becomes small compared to one. Note that a more detailed expression for the error term can be inferred from \eqref{NORM:DIFFERENCE:BOUND} in combination with Lemma \ref{lem:main_lemma}. The proof of Theorem \ref{thm:main_thm} is given in Section~\ref{Proof of the Main Result}. Before we discuss the model, the theorem and its application in physics in more detail, let us stress that Theorem~\ref{thm:main_thm} is nontrivial and might be surprising at first sight:

\begin{itemize}
\item Contrast the situation with a tracer particle in a classical or bosonic gas. Since the velocity of the tracer particle is of order one and the interaction proportional to the density $\rho$, then after times of $\mathcal O(1)$, the tracer particle has scattered with $\mathcal O(\rho)$ particles in the gas. The expected mean free path of the tracer particle is accordingly small, namely $\propto \rho^{-\delta}$ for some $\delta>0$.
\item 
In a fermionic gas, the kinetic energy of the tracer particle can dissipate into its environment by means of particle-hole excitations. One might expect that this kind of friction mechanism would become stronger the larger $\rho$. This is the case for a tracer particle in a Bose gas which was shown in the mean field regime on a rigorous level in \cite{froehlich:2012,froehlich:2014,deRoeck:2014}. For fermions one finds a different behavior: the larger the density, the less the particle is disturbed and, vice versa, disturbs the gas less. As a consequence, the free motion holds on a much larger time scale $t=O(\rho^{\delta})$ for some $\delta>0$; cf. the r.h.s.\ in \eqref{main_estimate}.
\end{itemize}

Our result follows from a careful analysis of the fluctuations in the gas and their propagation, and relies heavily on the Fermi pressure, i.e., the antisymmetry of the wave function of the fermionic particles. We give a sketch of the proof in Section~\ref{Sketch of the Proof} and provide a physically more intuitive explanation in Section~\ref{sec:sketch_proof_physical_picture}.

\subsection{The Model in More Detail\label{sec:model_more_detail}}
Let us discuss the considered model and its properties in more detail. First, note that we do not take any internal degrees of freedom such as spin into account. On the level of our main result, we do not expect a qualitative different behavior by doing so. Note also that our focus lies on the analysis of the interaction between the tracer particle and the gas, whereas the mutual interaction of the gas particles is neglected. While this is generally expected to be a reasonable approximation for many situations, its rigorous justification is a very interesting question on its own.\ Physical units are chosen such that the constant $\hbar$ and the masses of tracer particle and gas particles are dimensionless and $\hbar = 2m_{y} = 2m_x = 1$.

We model the potential between the tracer particle and each of the gas particles by an infinitely differentiable function with compact support (uniformly in $L$), i.e., $v \in C_0^{\infty}(\mathbb T^2)\cap C_0^{\infty}(\mathbb R^2)$. Theorem~\ref{thm:main_thm} holds as well for less regular potentials with fast enough decay at infinity. In order to simplify the proof as much as possible, however, the chosen class of potentials is very convenient. We often abbreviate the total interaction term in $H$ by $V = \sum_{i=1}^N v(x_i-y)$. Since $V$ is bounded, $H$ defines a self-adjoint operator on the second Sobolev space $ H^2(\mathbb T^{2(N+1)}) \subset \mathcal H_y\otimes \mathcal H_{N}$. For the corresponding time evolution, we write $U(t)=e^{-iHt}$, $t\ge 0$.

The initial wave function $\varphi_0$ of the tracer particle is restricted to be an element of $H^4(\mathbb T^{2}) \subset \mathcal H_y$ with $\no \varphi_0\no_{H^4} <C$ for all values of $\rho$.\ The initial state of the gas is assumed to be given by the ground state of the ideal Fermi gas which is described by the antisymmetric product of $N$ one-particle plane waves,
\begin{align}\label{FERMI:SEA}
\Omega_0(x_1,...,x_N) = \frac{1}{\sqrt{N!}} \sum_{\tau\in S_N} {(-1)^{\tau}} \prod_{i=1}^N \phi_{p_{\tau(i)}}(x_i),
\end{align}
with $\phi_{p}(x) = L^{-1} e^{i p \cdot x }\in L^2(\mathbb T^2)$, and $( p_j )_{j=1}^N$ the $N$ pairwise different elements of $\mathbb (2\pi / L ) \mathbb Z^2 $ with smallest absolute value. $S_N$ denotes the group of permutations of integers $\{1,...,N\}$ and $(-1)^{\tau}$ is the sign of the permutation $\tau$. Since the system is defined on a torus of side length $L$ (with periodic boundary conditions), the set of possible momenta in the gas is given by the lattice $\mathbb ( 2\pi / L ) \mathbb Z^2$. We label the momenta such that for $j_1, j_2\ge 1$ we have $j_1<j_2$ $\Leftrightarrow$ $\vert p_{j_1} \vert \le \vert p_{j_2} \vert$. The wave function $\Omega_0$ corresponds thus to the lowest possible kinetic energy given by $\sum_{k=1}^N p_k^2$.

It is later very convenient to use fermionic creation and annihilation operators.\ For wave functions $\Psi \in \mathcal H_y\otimes \mathcal H_N$ which are antisymmetric in the gas-coordinates, we have
\begin{align}
\label{creation_annihilation_op}a^*(p_l) a(p_k) \Psi(y,x_1,...,x_N)	= \sum_{i=1}^N \phi_{p_l}(x_i) \int_{\mathbb T} {dx_i}\, \phi_{p_k}^*(x_i) \Psi(y,x_1,...,x_N),
\end{align}
i.e., a particle with momentum $p_k \in (2\pi / L) \mathbb Z^2$ is replaced by a particle with momentum $p_l \in (2\pi / L) \mathbb Z^2$.

An important quantity that characterizes the state $\Omega_0$ is the Fermi momentum $k_F$. It is defined as the radius of the Fermi sphere which corresponds to the lattice points $(p_j)_{j=1}^N$ lying within a distance $\vert p_N \vert = k_F$ around the origin. The Fermi momentum is related to the average density $\rho$ via
\begin{align}
\rho = \frac{1}{L^2}\sum_{k=1}^N = \int_{\vert p \vert \le k_F} \frac{d^2p}{(2\pi)^2} = \frac{k_F^2}{4\pi} \hspace{1cm} \Leftrightarrow \hspace{1cm} k_F = \sqrt { 4 \pi\rho} .
\end{align}
We study the model in the thermodynamic limit $N,L \to \infty$ and $\rho=const.$ which we denote by $\lim_{\T{TD}}$. This simplifies the analysis because it allows us to ignore additional effects which are due to the chosen boundary conditions. For very large systems, i.e., in particular for $L/ \text{supp}(v) \gg 1$, such boundary effects are not expected to be physically relevant which justifies the analysis in the thermodynamic limit. We emphasize that for the result we are interested in in this work, it is really the parameter $\rho \gg 1$ which is the physically interesting one. We expect a very similar result to hold if one repeats all estimates for fixed but large values of $N$ and $L$, and then considers the regime in which $N\gg L$.

Let us next discuss the effective model.\ The effective dynamics is described by the Schr\"odinger equation with mean field Hamiltonian $H^{\text{mf}}$. Note that $H^{\text{mf}}$ is also self-adjoint on $H^2(\mathbb T^{2(N+1)} )$ and the corresponding mean field time evolution is denoted as $U^{\text{mf}}(t)$, $t\ge 0$. The average potential w.r.t.\ $\Omega_0$ that acts at position $y\in \mathbb T^2$,
\begin{align}
  E(y) =  \langle \Omega_0, V \Omega_0 \rangle_{\mathcal H_N} (y) = \rho \mathcal F[v](0),
\end{align}
where $\mathcal F$ denotes the Fourier transform, is spatially constant. The homogeneity of $E(y) = E$ is furthermore conserved under the mean field time evolution $U^{\text{mf}}(t)$, i.e.,
\begin{align}
\langle \Omega_t^f, V \Omega_t^f \rangle_{\mathcal H_N} = \langle  \Omega_0, V \Omega_0 \rangle_{\mathcal H_N}, \hspace{1cm} \Omega_t^f = e^{-iH^f_N t}\Omega_0,
\end{align}
where $H_N^f = - \sum_{i=1}^N \Delta_{x_i}$ denotes the free Hamiltonian of the gas. The Schr\"odinger equation with Hamiltonian (\ref{Mean-field Hamiltonian}) defines therefore a self-consistent approximation. The reason why we call $H^{\text{mf}}$ a mean field Hamiltonian is that to leading order, it is obtained from $H$ by replacing the potential $V$ by its average value $E$. The constant $E_{re}(\rho)$ is due to immediate recollisions between the tracer particle and gas particles. It is given by
\begin{align}
E_{re}(\rho) & = \lim_{\T{TD}} \widetilde E_{re}(N,\rho),
\label{def:F} \\
\widetilde E_{re}(N,\rho) & =  \frac{1}{L^4} \sum_{k=1}^N\sum_{l=N+1}^{\infty} \frac{\vert \mathcal F[v](
p_k - p_l ) \vert^2}{p_l^2-p_k^2} \theta \Big( \vert p_l \vert - \vert p_k \vert  - \rho^{-\frac{1}{2}} \Big), \nonumber
\end{align}
where $\theta(x)$ denotes the usual Heavyside step function, i.e., $\theta(x) = 1$ for $x\ge 0$ and zero otherwise. Eq.~\eqref{lemma:v_hat_estimates_Ka} of Lemma \ref{lem:v_hat_estimates} shows that for any $\varepsilon>0$ there are positive constants $C,C_{\varepsilon}$ such that
\begin{align} 
C \le E_{re}(\rho) \le C\rho^{2\varepsilon} + C_\varepsilon \rho^{-1/\epsilon}.
\end{align}
Since $\rho \mathcal F[v](0) - E_{re}(\rho)$ is constant as a function of the coordinates $y,x_1,...,x_N$, the time evolution $U^{\text{mf}}(t)$ is physically equivalent to the free dynamics generated by $H_y^f + H_N^f$ (where $H_y^f=-\Delta_y$).

Note that in the rest of this article we omit the subscripts $\mathcal H_y$, $\mathcal H_N$ or $\mathcal H_y\otimes \mathcal H_{N}$ on all scalar products and norms, since it is always clear from the argument on which space the scalar product or norm is meant.

\subsection{Discussion of the Main Result}

We give a list of nonrigorous remarks and assertions about the described model and Theorem \ref{thm:main_thm}.

\subsubsection{Spectral Properties}\label{sec:spectrum}

$H$ and $H^{\text{mf}}$ describe translation invariant systems and therefore the total momentum is conserved by both dynamics, $U(t)$ as well as $U^{\text{mf}}(t)$.\ In the microscopic model, however, the initial momentum of the tracer particle is not necessarily conserved due to the presence of the interaction potential.\ The joint energy-momentum spectrum of ($H$, $\hat P_{tot}$), $\hat P_{tot} = -i\nabla_y -\sum_{i=1}^N i\nabla_{x_i}$ being the total momentum operator, is thus expected to consist of degenerate values $(E_{tot}, P_{tot})$ where the degeneracy results from the different possibilities of splitting the total momentum $P_{tot}$ between the tracer particle and the gas. For $(E_{tot}, P_{tot}) = (\langle \Psi_0, H \Psi_0 \rangle , \langle \Psi_0, \hat P_{tot} \Psi_0 \rangle )$, the kinetic energy of the tracer particle may assume values between $E_y^{kin}=0$ and $E_y^{kin}=P_{tot}^2$. Note here that the smallest excitation energy of the gas is equal to $4\pi^2/ L^{2} \ll 1$. It is not difficult to verify that for every value $b\in [0,P_{tot}^2]$, there exists a wave function $\Psi^b \in \mathcal H_y\otimes \mathcal H_N $, such that
\begin{align}
\langle \Psi^b , -\Delta_y \Psi^b \rsp  =  b + \mathcal O( L^{-2} ), ~~~ \langle \Psi^b, \Psi^{b'} \rangle = 0 ~ \text{for} ~ \vert b-b' \vert > 4\pi^2/ L^2,
\end{align}
while the $\Psi^b$ are dynamically accessible in the sense that
\begin{align}
 \lsp \Psi^b, H \Psi^b\rsp = \lsp \Psi_0, H \Psi_0 \rsp, ~~~  \lsp \Psi^b, \hat P_{tot} \Psi^b \rsp = \lsp \Psi_0, \hat P_{tot} \Psi_0 \rsp.
\end{align}
In Figure \ref{Figure_1} (l.h.s.), we depict an example of such a wave function $\Psi^b$ (in this case a simple particle-hole excitation) which lowers the kinetic energy of the tracer particle for given values $( \langle \Psi_0, H \Psi_0 \rangle , \langle \Psi_0, \hat P_{tot} \Psi_0 \rangle)$. The reduced energy spectrum of the tracer particle is shown in Figure \ref{Figure_1} (r.h.s.): for $P_{tot}\neq 0$ and in the limit $L \to \infty$, our initial wave function $\Psi_0$ seems to lie on top of a continuous fiber $[0, P_{tot}^2 ]$ of dynamically accessible states.\footnote{To illustrate the above argument, let us give a very simple example of a transition that conserves total energy and total momentum, but lowers the kinetic energy of the tracer particle. Suppose, the initial momentum of the tracer particle is given by $\lsp \varphi_0, (-i\nabla) \varphi_0 \rsp = P_0 \neq 0$. Let us now consider a particle-hole excitation that absorbs all the momentum of the tracer particle, i.e., the momentum transfer is $\delta p = P_0$, such that the total momentum is conserved. Then the difference in the total energy before and after the collision is
\begin{equation}\label{E_heuristics}
\delta E = P_0^2 + p_k^2 - \Big(p_k + P_0 \Big)^2 = -2 p_k \cdot P_0.
\end{equation}
Therefore, energy is conserved in this case if $p_k \cdot P_0 \approx 0$. One can easily convince oneself similarly that there are particle-hole excitations from all states within two small regions (depicted in grey in Figure~\ref{Figure_1}, l.h.s.) which change the momentum of the tracer particle by $\vert \delta p\vert \in [0,P_0]$ and conserve the total energy.
}

Although the rigorous analysis of spectral properties in the thermodynamic limit is very subtle, we expect the above considerations to be true for $d\ge 2$, and for all $\rho \ge 1$.  
\begin{figure}[ttb] \centering \def\svgwidth{540pt} 
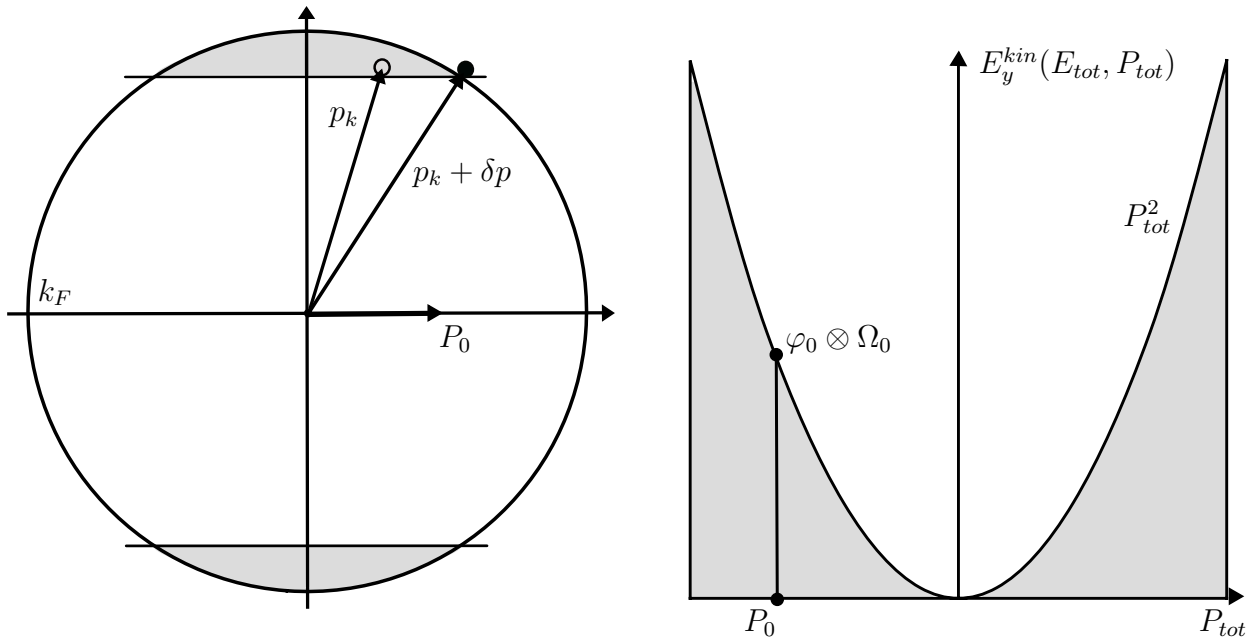
\caption{\label{Figure_1} L.h.s.:\ A possible particle-hole excitation in the 2d Fermi sphere which lowers the kinetic energy of the tracer particle from $P_0^2$ to $(P_0-\delta p)^2$ while leaving total momentum $P_0 = \lsp \Psi_0, \hat P_{tot} \Psi_0 \rsp$ and total energy $E_0 = \lsp \Psi_0, H \Psi_0 \rsp$ unchanged.\ For given $P_0$ and large $k_F$, only particles from the grey region can be excited. It follows from energy and momentum conservation that the area of this region is of order $1/k_F$, see Section \ref{sec:Momentum_space}.\ R.h.s.:\ Reduced kinetic energy spectrum of the tracer particle $E_y^{kin}(E_{tot},P_{tot})$ for given value $E_{tot}$ and as function of $P_{tot}$. The part below $P_{tot}^2$ corresponds to dynamically accessible states for which $E_y^{kin}(E_{tot},P_{tot})$ ranges between zero and $P_{tot}^2$.}
\end{figure} 

\subsubsection{Dimension\label{sec:dimension}}
The spectral properties are very different in one spatial dimension. For $d=1$, there are no dynamically accessible states, i.e., wave functions with total energy and momentum equal to that of $\Psi_{0}$, for which the average kinetic energy of the tracer particle is smaller than its initial value $P_{tot}^2$. Any significant transfer of momentum from $\varphi_0$ to $\Omega_0$ would cause an increase in the energy of the gas proportional to $k_F$ which is due to the quadratic energy dispersion relation.\footnote{This can be seen from \eqref{E_heuristics}: In one dimension the energy difference is $-2p_k \cdot P_0 \propto -k_F$ for any $p_k$ near the Fermi surface. For $d\ge 2$, this is different as explained above and indicated in Figure \ref{Figure_1}.} The reduced kinetic energy spectrum of the tracer particle in one dimension and for large $k_F$ is therefore the same as in the free model: $E_{y}^{kin} = \langle \Psi_0, \hat P_{tot} \Psi_0 \rangle^2$, with no other values allowed. This makes a result similar to Theorem~\ref{thm:main_thm} less surprising. A rigorous analysis of the one dimensional model was carried out in \cite{jeblick:2013}. In Appendix \ref{app:one_dimension}, we write down the theorem for $d=1$ and give a short sketch of how the the argument we employ to prove Theorem~\ref{thm:main_thm} is adapted to the one dimensional case.\\
\\
For $d=3$, the spectral properties are similar to the case $d=2$.\ However, for $d=3$, it is unclear if a similar result about the dynamics holds; see also Remark~\ref{rem:dimension} and Appendix \ref{app:three_dimensions}. This is why we chose to study $d=2$ in this article.

\subsubsection{Asymptotic Energy Loss} \label{sec:dissipation_time}

From the spectral picture that was explained in Section \ref{sec:spectrum}, we expect that eventually, the kinetic energy of the tracer particle dissipates into the gas by means of particle-hole excitations. Recall Figure \ref{Figure_1} (r.h.s): for $L \to \infty$, the initial wave function $\Psi_0$ lies above the continuous fiber over $P_0=\lsp \Psi_0, \hat P_{tot} \Psi_0 \rsp$. If the initial momentum is nonzero, the tracer particle occupies an excited state which is coupled to a dispersive medium with a large number of degrees of freedom. In such a situation, one may expect that the excited subsystem approaches asymptotically its lowest energy state. For the Fermi gas, this friction mechanism is suppressed for large values of $\rho$.\ Theorem~\ref{thm:main_thm} states that $\Psi_0$, or equivalently, the initial momentum distribution of the tracer particle, is stable on a large time scale, namely of the order of some power of $\rho$.\ On some larger time scale the tracer particle is expected to slow down until it reaches its ground state $E_{y}^{kin} = 0 $.

The rigorous understanding of existence and properties such as lifetime and decay rate of long-lived resonances is, however, very difficult. It needs more refined techniques and perhaps a more general formulation of the model (e.g., defining it directly on $\mathbb R^2$) in order to describe the physically correct behavior for $t \to\infty$. In \cite{Sabin:2014}, e.g., a similar question was studied on the level of the (fermionic) Hartree equation for which it was shown that a small defect added initially to the translation-invariant homogeneous state disappears for large times due to dispersive effects of the gas.

\subsubsection{Momentum Space\label{sec:Momentum_space}}
Let us take a closer look at the number of states in momentum space that can in principle interact with the tracer particle.\ In the limit of large $\rho$ only few particles in $\Omega_0$ can be lifted above $k_F$. The number of transitions is restricted because of the Pauli principle and conservation of energy and momentum. In Figure \ref{Figure_1} (l.h.s.) we also indicate the momentum space region of particles that can exchange energy with the tracer particle for given $P_0=\langle \Psi_0,P_{tot} \Psi_0 \rangle$. The probability of finding a gas particle in this phase space area at a given instant of time within the range of the interaction of the tracer particle is small, namely $\propto 1/k_F$ (note that the grey area in Figure \ref{Figure_1} (l.h.s.) is chosen to have a width $\propto \vert P_0\vert = \mathcal O(1)$ which implies that its height shrinks like $1/k_F$). During a unit time interval the number of such particles is nevertheless again at least of order one, i.e., in particular independent of the value of $\rho$.\ This is true because the corresponding Boltzmann cylinder (the distance from where the gas particles can reach the tracer particle during a unit time interval) has length $\propto k_F$. Therefore, already for times of order one, the tracer particle could interact with order one gas particles.\ However, in Theorem \ref{thm:main_thm} we show that this does not happen and the tracer particle moves freely even for longer time scales than order one.

\subsubsection{Norm Distance}\label{sec:norm}
Since we consider a regime of strong coupling, even a single collision can be enough to disturb the free motion of the system. It is thus necessary to prove that all particles behave according to the mean field equation. This is the reason why the difference in norm between $U(t)\Psi_0$ and $U^{\text{mf}}(t)\Psi_0$ is the right quantity to consider. Note that the situation is different compared to the weak coupling regime where the aim is usually to prove that the relative number of particles which evolve according to the mean field potential is close to one; see, e.g., \cite{bardos:2003,erdoes:2004,froehlich:2011,benedikter:2013,
petrat:2016c,BachPetrat:2016,Petrat:2016b} for the fermionic case.

\subsubsection{Fluctuations and Mean Field Regime} \label{sec:fluc}
The substitution of the potential $V$ in $U(t)\Psi_0$ by its average value $E$ would be easy to justify if fluctuations around $E$ were negligibly small, i.e., if $\label{18}	V \approx E$ would hold with probability close to one (w.r.t.\ the probability density defined by $\vert\Omega_t^f\vert^2$). We show in Lemma~\ref{lem:v_hat_estimates}, Eq.\eqref{lemma:v_hat_estimates_fluc}, that this is not the case: while $ \lim_\T{TD}\no (V-E)\Omega_t^f\no^2$ is suppressed in the sense that it grows only with $\sqrt \rho$ instead of $\rho$ (as naively expected from the square root of N law), it still diverges in the limit $\rho\to\infty$. The reason for the large fluctuations is the strong coupling $g=1$. If we had assumed a weak coupling, say $g=\rho^{-1}$, the fluctuations would vanish when $\rho$ tends to $\infty$ and an estimate like in Theorem \ref{thm:main_thm} would follow almost trivially. We emphasize this because it exemplifies an interesting fact: the mean field regime for fermions does not necessarily coincide with a weak coupling limit $g\to0$ $(\rho\to\infty)$. For bosons, on the other hand, the mean field regime coincides with the weak coupling limit. In other words, Theorem \ref{thm:main_thm} provides an explicit example of a setting where the mean field regime for fermions is much larger compared to bosons or classical particles. In Section \ref{sec:sketch_proof_physical_picture} we give a short explanation of why the mean field description is valid even though the fluctuations can be very large.
\\
\\
Let us also remark that more generally one finds for $d=1,2,3$ spatial dimensions,
\begin{align}
\label{flucutations_d_dimension}\lim_{\T{TD}} \bno (V-E) \Omega_t^f \bno^2 = C_d \rho^{\frac{d-1}{d}},
\end{align}
with $d$-dependent constants $C_d$. A similar result about the suppression of fluctuations in a Fermi gas has been mentioned in \cite[Eqs.~(48)-(50)]{castin:2006}. Compared to \eqref{flucutations_d_dimension}, there appears an additional factor $\ln \rho$ on the r.h.s.\ which is due to the fact that $v$ was chosen less regular than in our case.

\subsubsection{Subleading Energy Correction $E_{re}(\rho)$}

In Lemma \ref{lem:v_hat_estimates} we show for $d=2$ that $C \leq E_{re}(\rho) \le C \rho^{2\varepsilon} + C_\varepsilon \rho^{-1/\varepsilon}$ for any $\varepsilon>0$ and positive constants $C,C_\varepsilon$.\ This means in particular that the claimed estimate in Theorem~\ref{thm:main_thm} would not be correct without including $E_{re}(\rho)$ in the definition of $H^{\text{mf}}$.\ Nevertheless, $E_{re}(\rho)$ is only a subleading correction to the mean field energy $\rho \mathcal F[v]$.\ It arises from so-called immediate recollisions, i.e.,\ collisions of the type where the tracer particle excites a particle-hole pair in the gas and then immediately recollides with the excited particle which recombines with the hole.\ Such processes appear in the expansion of $\Psi_t$ into the different collision histories that have to be controlled, see the end of Section~\ref{sec:main_lemma}.\footnote{Let us remark that if we iterate the Duhamel expansion \eqref{integral_Se} infinitely often, one can identify in each order terms that contain only immediate recollisions. Then one can indeed show that the phase factor $e^{iE_{re}(\rho)t}$ cancels exactly the leading order contribution of those immediate recollision terms summed up in all orders. However, we refrain from showing this here since for the proof of Theorem \ref{thm:main_thm} it is sufficient to stop the Duhamel expansion at third order.} From the definition of $E_{re}(\rho)$ in \eqref{def:F} one can see that only gas particles near the Fermi surface contribute to $E_{re}(\rho)$.

\subsubsection{Application to Physics} \label{sec:applications}

The presented model is very close to the physically interesting situation of ions moving through a degenerate and dense electron plasma. An understanding of what is often referred to as slowing down of ions in a degenerate plasma has been of interest in the physics literature at least since a work by Fermi and Teller in 1947\ \cite{fermi:1947} (see also \cite{mott:1949}). They have pointed out that the efficiency of the gas for slowing down ions with velocities far below the Fermi edge is very low. The same question has later been analyzed explicitly for the high-density case for which the energy loss of the ions was found to be caused mainly by (rare) collisions with other ions instead of interactions with the electrons from the plasma; see, e.g., \cite{gryzinski:1957,ritchie:1959,dar:1974,williams:1975,yakovlev:1983}. These results raised considerable interest in the field of nuclear physics in which it was known that the existence of long-lived ions in the plasma is essential for the occurrence of fusion reactions; e.g., \cite{brysk:1974,peres:1975}.

Let us stress that to our knowledge, the analysis has remained so far on a purely formal level. The rigorous bound we present here, starting from the microscopic dynamics and taking into account the full strong interaction, seems to be novel.

\subsection{Sketch of the Proof}\label{Sketch of the Proof}

For deriving Theorem \ref{thm:main_thm} we use Duhamel's expansion in order to decompose $\Psi_t$ into different wave functions that correspond to different collision histories of the tracer particle. The main difficulty is to control the interaction with particles occupying momenta close to the Fermi edge. Our main ingredient here is the large shift in the energy and the thereby caused phase cancellation during the scattering with such particles. It turns out to be necessary but also sufficient to use a third order expansion in the difference $H-H^{\text{mf}}$. Let us stress again that $g=1$. This prevents us from using a straightforward order by order expansion of the time evolution. Thus, after expanding to third order, we have to estimate an error term involving the whole time evolution $U(t)$. In order to convey the main ideas and techniques behind the proof, let us start by expanding
\begin{align}
\label{integral_Se} U(t)\Psi_0 - U^{\text{mf}} (t)\Psi_0 = -i\int_0^td\tau_1 & U^{\text{mf}}(t-\tau_1)(H-H^{\text{mf}}) U^{\text{mf}}(\tau_1)\Psi_0   \\
& \hspace{-0.5cm} 
- i \int_0^td \tau_1 \Big( U(t-\tau_1) - U^{\text{mf}}(t-\tau_1) \Big) (H-H^{\text{mf}}) U^{\text{mf}}(\tau_1)\Psi_0  , \nonumber
\end{align}
which follows from expanding $U$ around $U^{\text{mf}}$ in terms of Duhamel's formula and then splitting $U = U^{\text{mf}} + ( U - U^{\text{mf}})$.  The first term on the r.h.s.\ contains deviations from the effective dynamics due to single particle-hole excitations. In order to present the main argument, let us ignore the next-to-leading order energy correction $E_{re}(\rho)$ in the following. Using some elementary algebra (only momenta inside the Fermi sphere can be annihilated and momenta outside the Fermi sphere created), one readily rewrites
\begin{align}\label{DECOMPOSITION:V-E}
	(V-E) \Psi_0 = \Vol{2} \sum_{k=1}^N \sum_{l=N+1}^{\infty} \mathcal F[v](p_{l}-p_{k}) \Big( e^{i ( p_{l} - p_{k} ) y} \varphi_0 \Big) \otimes a^*(p_{l}) a(p_{k}) \Omega_0.
\end{align}
Abbreviating $k_{kl}(\tau_1) = e^{i H_y^f \tau_1} e^{i ( p_{l} - p_{k} ) y } e^{-iH_y^f \tau_1}$
(which evolves the tracer particle freely to time $\tau_1$ at which its momentum is changed by $p_{l} - p_{k} $ and then evolves it back to the initial time), it is also straightforward to arrive at
\begin{align}
&  \bno \int_0^t d\tau_1 U^{\text{mf}}(-\tau_1)(V-E)U^{\text{mf}} (\tau_1)\Psi_0  \bno^2 \nonumber \\
&\hspace{3cm}= \underbrace{\Vol{4} \sum_{k=1}^N \sum_{l=N+1}^{\infty} \Big\vert \mathcal F[v](p_{k}-p_{l}) \Big\vert^2}_{=\vert\vert (V-E) \Omega_t^f \vert\vert^2} \bno \int_0^t d\tau_1 e^{i (p_l^2-p_k^2 )\tau_1} k_{kl}(\tau_1) \varphi_0 \bno^2.  \label{FIRST:ORDER:DEVIATIONS}
\end{align}
Due to the regularity of the potential $v$ it is unlikely that a single collision causes a large momentum transfer between $\varphi_0$ and $\Omega_0$. This is reflected in the fact that the Fourier transform of a smooth and compactly supported function decays faster than any polynomial: for all $p \in \mathbb N$ there exists a constant $D_p$ such that
\begin{align}
\label{Paley_Wiener_potential} \Big\vert \mathcal F[v](p_{k}-p_{l}) \Big\vert \le \frac{D_p}{( 1+ \vert p_{k}-p_{l}\vert )^p},
\end{align}
which follows directly from the Paley-Wiener Theorem; e.g., \cite[Theorem XI.11]{reed1975methods}. At this point it is convenient to introduce the following notation. For $\varepsilon>0$ we define $v^{l,\varepsilon}$ and $v^{s,\varepsilon}$ such that
\begin{align} \label{potential_large}
		\mathcal F[v^{l,\varepsilon}](p_{k}-p_{l}) & = \theta\big(\vert p_{k}-p_{l} \vert - \rho^\varepsilon\big) \mathcal F[v] (p_{k}-p_{l})\\
	\label{potential_small}	\mathcal F[v^{s,\varepsilon}] (p_{k}-p_{l}) & = \theta\big(\rho^\varepsilon - \vert p_{k}-p_{l} \vert \big) \mathcal F[v](p_{k}-p_{l} ).
\end{align}
The transition amplitude $\vert \mathcal F [v^{l,\varepsilon}] (p_{k}-p_{l}) \vert^2$ is negligible for $\rho\gg 1$ which can be inferred from \eqref{Paley_Wiener_potential}. What remains to be bounded is the transitions in \eqref{FIRST:ORDER:DEVIATIONS} with momentum transfer of order one, i.e.,
\begin{align}
\Vol{4} \sum_{k=1}^N \sum_{l=N+1}^{\infty} \Big\vert \mathcal F[v^{s,\varepsilon}](p_{k}-p_{l}) \Big\vert^2 \bno \int_0^t d\tau_1 e^{i (p_l^2-p_k^2) \tau_1} k_{kl}(\tau_1) \varphi_0 \bno^2. \label{FIRST:ORDER}
\end{align}
The reason that this term vanishes as well is the oscillation of the integrand $e^{i (p_l^2-p_k^2) \tau_1} k_{kl}(\tau_1) \varphi_0$. Outside a set of critical points of the phase for which $\vert p_{l}\vert - \vert p_{k} \vert \le \kappa (\rho)$, for some appropriately small $\kappa (\rho)\ll 1$, the energy shift grows rapidly: $p_l^2-p_k^2 = \big(|p_l| + |p_k|\big) \big(|p_l| - |p_k|\big) \gtrsim \sqrt \rho \kappa(\rho) \gg 1$. By partial integration, one thus finds that
\begin{align}
	\eqref{FIRST:ORDER} \lesssim \frac{t^2}{L^4} \underset{ \{ \text{ stationary points} \ \} }{ \bigg[\sum_{k=1}^N \sum_{l=N+1}^{\infty}}+ \frac{1}{\rho \kappa (\rho)^2 } \sum_{k=1}^N \sum_{l=N+1}^{\infty}	\bigg] ~ \Big\vert \mathcal F[v^{s,\varepsilon}] (p_{k}-p_{l}) \Big\vert^2 ,\label{sketch_of_estimate} 
\end{align}
which will be shown to vanish in the limit $\rho \to \infty$, see Remarks~\ref{rem:sketch_proof_stat_point} and \ref{rem:dimension} below. This result is the key ingredient to understand the proof of Theorem~\ref{thm:main_thm}. Since the interaction is modeled by a two-body potential, it is reasonable to expect that an appropriate estimate for the higher order terms in \eqref{integral_Se} follows from a bound of the r.h.s.\ of \eqref{sketch_of_estimate}. Technically, however, it is more tedious to obtain good control of the higher-order contributions. The difficulty is the appearance of the full time evolution $U$. Using the Duhamel expansion for $U-U^{\text{mf}}$, one finds an estimate similar to
\begin{align}
\bno & \int_{0}^t d\tau_1   \int_{0}^{\tau_1}    d\tau_2   U(\tau_2)   (V-E) U^{\text{mf}} (\tau_2    -\tau_1)(V-E) U^{\text{mf}} (\tau_1)\Psi_0 \bno^2  \lesssim t^2 \cdot \eqref{FIRST:ORDER}   \cdot \lno (V-E) \Omega_t^f \rno^2 , \nonumber
\end{align}
for which the r.h.s., however, is still divergent for $\rho\to \infty$ (recall that $\lim_{\T{TD}}\no (V-E) \Omega_t^f \no \to \infty$ when $\rho$ tends to $\infty$). Expanding $U$ another time, the main contribution that has to be controlled is given by
\begin{align} 
\label{R_moral_estimate_2} \bno \int_{0}^t d\tau_1  \int_{0}^{\tau_1} d\tau_2  {\int_{0}^{\tau_2} d\tau_3} U(\tau_3) (V-E) U^{\text{mf}} (\tau_3-\tau_2)(V-E) U^{\text{mf}}(\tau_2-\tau_1)(V-E) U^{\text{mf}}(\tau_1)\Psi_0 \bno.\nonumber
\end{align}
Now one can use the oscillation of the integrand also in the second time-variable. It will be shown in detail that this can be bounded in terms of
\begin{align}
 t^2  \cdot \eqref{FIRST:ORDER}^2  \cdot \lno (V-E) \Omega_t^f \rno^2    \to 0 \ (\rho\to \infty).
\end{align}
This explains why we expand the dynamics up to third order for proving Theorem~\ref{thm:main_thm}. Let us conclude with some remarks.

\begin{remark}\label{rem:sketch_proof_stat_point}
In Lemma~\ref{lem:v_hat_estimates} we show that $\lim_{\T{TD}} \no (V-E) \Omega_0 \no^2   \propto \sqrt \rho$. The term containing the nonstationary terms in (\ref{sketch_of_estimate}) will therefore be proportional to $\rho^{-\frac{1}{2}} \kappa(\rho)^{-2}$. The term containing the stationary points in (\ref{sketch_of_estimate}) turns out, in $d=2$, to be proportional to $\sqrt \rho \kappa(\rho)^2$. Thus, one actually needs a finer separation around the stationary points in order to obtain the desired bound in $\rho$. The details of this separation are explained in Section~\ref{sec:preliminaries}.
\end{remark} 
\begin{remark}\label{rem:dimension}The second summand on the r.h.s.\ of \eqref{sketch_of_estimate} behaves at best like $\no (V-E)\Omega_0\no^2 / k_F^2$.
Recalling \eqref{flucutations_d_dimension} as well as $k_F \propto \rho^{\frac{1}{d}}$ in $d$ dimensions, it is clear that a similar statement as Theorem~\ref{thm:main_thm} also holds for $d=1$. For $d=3$, on the contrary, the last term is not small, even if one optimizes the separation of the nonstationary points. We provide more details about $d=1$ and $d=3$ in Appendices \ref{app:one_dimension} and \ref{app:three_dimensions}.
\end{remark} 
\begin{remark}
Many of the techniques we use in the proof of our main result also appear in the proof of quantum diffusion  of a particle in an external random potential \cite{Erdoes:2008,Erdoes:2007}, which is in several respects much more involved. The main difficulty of our problem is that we do not have a small coupling constant, but we need to show that the interaction is effectively small.
\end{remark} 

\subsubsection{Physical Picture Behind the Proof}\label{sec:sketch_proof_physical_picture} 

On the one hand, it is obvious that for $\rho\gg 1$ the tracer particle can interact only with particles that occupy a momentum close to the Fermi edge. This is due to the exclusion principle and because all momenta smaller than $k_F$ are occupied in $\Omega_0$. Particles with small momentum can simply not be lifted above $k_F$. In other words, for $\rho \gg 1$, the Fermi pressure becomes so strong that the particles far inside the Fermi sphere behave very rigidly and are thus hardly disturbed by the presence of the tracer particle (and vice versa). On the other hand, the reason why collisions with particles with momentum close to $k_F$ do not disturb the free motion is that such particles have very large momenta when $\rho\gg 1$ (i.e., for $m_{x}=1/2$, large velocities) and thus interact only on a very short time scale with the tracer particle. Hence, the momentum transfer is effectively small. Let us note that in the limit of very short wave lengths, the particle behavior is dominant, which makes this explanation plausible. In the proof, the high momenta of the gas particles (or, the short time scale of interaction) appear as the factor $\rho^{-1}\propto k_F^{-2}$ in (\ref{sketch_of_estimate}).

\section{\label{Proof of the Main Result}Proof of the Main Result}
\subsection{Notations and Definitions}
Let for any $t\ge 0$
\begin{align}
 k_{kl}(t): \mathcal H_y \to \mathcal H_y,& \hspace{1cm} \varphi\to k_{kl}(t)\varphi = e^{iH_y^ft}e^{-i(p_k-p_l)y}e^{-iH_y^ft}\varphi \\
g_{kl}(t): \mathcal H_y \mapsto \mathcal H_y,& \hspace{1cm} \varphi \mapsto g_{kl}(t)\varphi = e^{-i( p_k^2-p_l^2)t} k_{kl}(t) \varphi,
\end{align}
and
\begin{align}D(t) : \mathcal H_y\otimes\mathcal H_{N} \to \mathcal H_y \otimes \mathcal H_{N}, \hspace{1cm} \Psi \mapsto D(t) \Psi := U(-t) U^{\text{mf}}(t) \Psi.
\end{align}
We denote the Fourier transform of the potential $v$ by $\mathcal F[v]$, where $\mathcal F[v]$ is defined such that
\begin{align}
v(x) = \Vol{2} \sum_{k=1}^\infty \mathcal F[v] ( p_k ) e^{i p_k x }.
\end{align}
We use the following abbreviations:
\begin{itemize}
\item $\hat v_{kl} = \mathcal F[v] ( p_k -p_l)$,
\item $\hat v_{kl}^{\ell,\varepsilon} = \theta\big(-\rho^\varepsilon + \vert p_k-p_l \vert \big) \hat v_{kl} , ~~~ \hat v_{kl}^{s,\varepsilon} = \theta\big(\rho^\varepsilon - \vert p_k-p_l \vert\big) \hat v_{kl}$; cf.\ \eqref{potential_large},\eqref{potential_small},
\item $ E_k = p_k^2$,
\item $\no \cdot \no_{\T{TD}} = \lim_{\T{TD}} \no \cdot \no$;\footnote{Note that despite the chosen notation, $\no \cdot \no_{\T{TD}}$ does not define a proper norm since $\no f \no_{\T{TD}}$ may be zero for nonzero $f$.} recall that $\lim_{\T{TD}}\no \cdot \no = \lim_{N,L\to\infty, \rho=const.} \no \cdot \no$,
\item $a^*(p_l)a(p_k) \Omega_0 = \Omega_0^{[l^*k]}$ and all kind of variations thereof.
\end{itemize}
\noindent Next, we introduce
\begin{align}
\Psi_{1} (\tau_2,\tau_1)& = \Vol{4}  \sum_{k_1=1}^N \sum_{l_1=N+1}^{\infty}  \vert \hat v_{k_1l_1}\vert^2 \Big( g_{l_1k_1}(\tau_2) g_{k_1l_1}(\tau_1) \varphi_0 \Big) \otimes \Omega_0,\\
\Psi_{2}(\tau_2,\tau_1) & =   \Vol{4}  \sum_{k_1,k_2=1}^N \sum_{l_1=N+1}^{\infty}  \hat v_{k_2k_1} \hat v_{k_1l_1} \Big( g_{k_2k_1}(\tau_2)g_{k_1l_1}(\tau_1) \varphi_0 \Big) \otimes \Omega_0^{[k_2 l_1^*]}, \\ 
\Psi_{3} (\tau_2,\tau_1)  &=  \Vol{4}  \sum_{k_1=1}^N \sum_{l_1,l_2=N+1}^{\infty} \hat v_{l_1l_2} \hat v_{k_1l_1} \Big( g_{l_1l_2}(\tau_2) g_{k_1l_1}(\tau_1) \varphi_0 \Big) \otimes \Omega_0^{[l_2^* k_1 ]}, \\
\Psi_{4}(\tau_2,\tau_1) &=  \Vol{4}  \sum_{k_1,k_2=1}^N \sum_{l_1, l_2=N+1}^{\infty} \hat v_{k_2l_2} \hat v_{k_1l_1} \Big( g_{k_2l_2}(\tau_2) g_{k_1l_1}(\tau_1) \varphi_0 \Big) \otimes \Omega_0^{[l_2^*k_2l_1^*k_1]}
\end{align}
which arise from the equality
\begin{align} \label{DEF:PSI_1234}
	U^{\text{mf}} (-\tau_2)(V-E) U^{\text{mf}} (\tau_2-\tau_1)(V-E) U^{\text{mf}}(\tau_1) \Psi_0 = \Psi_{1} + \Psi_{2} + \Psi_{3} + \Psi_{4}.
\end{align}
We further define\allowdisplaybreaks
\begin{align}
\Psi_{\T{A}}(t) = & E_{re}(\rho) \int_0^t d\mu_2(\tau) D(\tau_2) U^{\text{mf}} (-\tau_1) (V-E) U^{\text{mf}}(\tau_1)\Psi_0,\\
\Psi_{\T{B}}(t) = &  E_{re}(\rho) \int_0^{t}d\tau_1  D(\tau_1)\Psi_0 + i \int_0^t d\mu_2(\tau)  D (\tau_2) \Psi_{1}(\tau_2,\tau_1) ,\\
\Psi_{\T{C}}(t) = & \int_0^t d\mu_2(\tau) D(\tau_2)   \Psi_{2}(\tau_2,\tau_1) ,\\
\Psi_{\T{D}}(t) = & \int_0^t d\mu_2(\tau) D(\tau_2)  \Psi_{3}(\tau_2,\tau_1) ,\\
\Psi_{\T{E}}(t) = &  E_{re}(\rho)  \int_0^t d\mu_3(\tau)  D(\tau_3)\Psi_{4}(\tau_2,\tau_1),\\
\Psi_{\T{F}}(t)= & \int_0^t d\mu_3(\tau)  U(-\tau_3) (V-E) U^{\text{mf}} (\tau_3)\Psi_{4}(\tau_2,\tau_1),
\end{align}
where we have introduced the shorthand notation
\begin{align}
\int_0^t d\mu_n(\tau) = \int_{0}^t d\tau_1 \int_0^{\tau_{1}} d\tau_{2} ~ ... \int_{0}^{\tau_{n-1}} d\tau_n.
\end{align}

\subsection{Main Lemma and Proof of Theorem~\ref{thm:main_thm} \label{sec:main_lemma}}

The main lemma we have to prove is

\begin{lemma}\label{lem:main_lemma} Let $0<\varepsilon< 1/8$. Under the same assumptions as in Theorem~\ref{thm:main_thm}, there exist positive constants $C$, $C_\varepsilon$ such that
\end{lemma}
\vspace{-0.7cm}
\allowdisplaybreaks
\begin{align}
 \no \Psi_{\T{\tiny A}}(t) \no_{\T{TD}} &\le C(1+t)^2 \Big( \rho^{-\frac{1}{4} + 5\varepsilon} + C_\varepsilon \rho^{-\frac{3}{2\varepsilon}} \Big)   \label{lemma:main-norm-of-E} ,\\ 
\no  \Psi_{\T{\tiny B}}(t) \no_{\T{TD}} & \le C(1+t)^2 \Big( \rho^{-\frac{1}{4} + 2\varepsilon}  + C_\varepsilon \rho^{- \frac{1}{\varepsilon} } \Big) \label{lemma:main-norm-of-A1}, \\
\no \Psi_{\T{\tiny C}}(t) \no_{\T{TD}}  & \le C(1+t)^2 \Big( \rho^{-\frac{1}{4} + 2\varepsilon  } +  C_\varepsilon \rho^{\varepsilon -\frac{1}{2\varepsilon}} \Big), \label{lemma:main-norm-of-C}\\
 \no \Psi_{\T{\tiny D}}(t) \no_{\T{TD}} & \le C(1+t)^2 \Big( \rho^{-\frac{1}{4} + 2\varepsilon } +  C_\varepsilon \rho^{\varepsilon-\frac{1}{2\varepsilon}} \Big) ,  \label{lemma:main-norm-of-B1}\\ 
  \no \Psi_{\T{\tiny E}}(t) \no_{\T{TD}} & \le C(1+t)^3 \Big( \rho^{-\frac{1}{2}+8\varepsilon }  + C_\varepsilon \rho^{\frac{1}{4}+6\varepsilon-\frac{1}{2\varepsilon}} \Big), \label{lemma:main-norm-D1}   \\
\label{lemma:main_lemma_2}  \no \Psi_{\T{\tiny F}}(t) \no_{\T{TD}} & \le C(1+t)^3 \Big(\rho^{-\frac{1}{4} + 6\varepsilon} + C_\varepsilon \rho^{\frac{1}{2} -\frac{1}{2\varepsilon}}\Big),
\end{align}
\textit{hold for all $t>0$.} \\

Theorem \ref{thm:main_thm} follows from the above bounds.

\begin{proof}[Proof of Theorem~\ref{thm:main_thm}] Let $ \varepsilon > 0$. We begin with Duhamel's formula,
\begin{align}
U(t)-U^{\text{mf}}(t) = -i\int_0^t d\tau_1 U^{\text{mf}} (t-\tau_1) \Big(V-E + E_{re}(\rho) \Big)  U(\tau_1),
\end{align}
then use that $\langle U^{\text{mf}}(\tau_1) \Psi_0,( V-E ) U^{\text{mf}} (\tau_1)\Psi_0 \rangle = \langle \Psi_0,( V-E ) \Psi_0 \rangle =0,$ apply Duhamel's formula again, and eventually use the identity in \eqref{DEF:PSI_1234}:
\allowdisplaybreaks
\begin{align}
\frac{1}{2} & \bno U(t) \Psi_0 - U^{\text{mf}}(t) \Psi_0 \bno^2 =  - \re\lsp \Big( U(t) - U^{\text{mf}}(t) \Big) \Psi_0,U^{\text{mf}}(t)\Psi_0 \rsp   \\
	= &  - \re \lsp \Psi_0, i E_{re}(\rho) \int_0^{t}d\tau_1 U(-\tau_1) U^{\text{mf}} (\tau_1)\Psi_0 \rsp \nonumber\\
	  & +  \re \lsp i \int_0^{t}d\tau_1  \Big( U(\tau_1)- U^{\text{mf}} (\tau_1) \Big)\Psi_0, \big( V-E \big) U^{\text{mf}} (\tau_1)\Psi_0 \rsp   \nonumber\\
 = & - \re \lsp \Psi_0, i E_{re}(\rho) \int_0^{t}d\tau_1 U(-\tau_1) U^{\text{mf}}  (\tau_1)\Psi_0 \rsp \nonumber\\
 &    + \re\lsp \Psi_0, E_{re} (\rho) \int_0^t d\mu_2(\tau_2)  U(-\tau_2) U^{\text{mf}} (\tau_2-\tau_1)\big( V-E \big) U^{\text{mf}} (\tau_1)\Psi_0 \rsp \nonumber\\
 & + \re\lsp \int_0^{t} d\mu_2(\tau) U^{\text{mf}} (-\tau_2) U(\tau_2) \Psi_0, \underbrace{U^{\text{mf}}(-\tau_2) (V-E )U^{\text{mf}} (\tau_2-\tau_1)(V-E)U^{\text{mf}}(\tau_1)\Psi_0}_{ = \Psi_{\text 1} + \Psi_{\text 2} + \Psi_{\text 3} + \Psi_{\text 4}} \rsp  \nonumber.
\end{align}
We proceed with the term that contains $\Psi_4$.\ Using $ \lsp \Psi_0, \Psi_{\text 4}(\tau_2,\tau_1) \rsp   =0 $ (note that $\Psi_{\text 4}$ always contains a particle outside the Fermi sphere), and applying one more time Duhamel's formula, we find
\begin{align}
  & \re\lsp \int_0^{t} d\mu_2(\tau) U^{\text{mf}}(-\tau_2) \Big( U(\tau_2) - U^{\text{mf}}(\tau_2) \Big)  \Psi_0,  \Psi_{\text 4} (\tau_2,\tau_1) \rsp \\
    & =  \re\lsp \Psi_0, i \int_0^{t} d\mu_3(\tau) U(-\tau_3) \Big( V-E  + E_{re}(\rho)   \Big) U^{\text{mf}}(\tau_3) \Psi_{\text{4}}(\tau_2,\tau_1)  \rsp = \re   \lsp \Psi_0, i \Big(  \Psi_{\T{\tiny E}}(t) +  \Psi_{\T{\tiny F}}(t) \Big)\rsp .\nonumber
\end{align}
By means of the triangle inequality and Cauchy Schwarz, it follows that
\begin{align}
\bno U(t) \Psi_0 - U^{\text{mf}}(t) \Psi_0 \bno \le & 2 \Big( \sqrt{\no \Psi_{\T{\tiny A}}(t) \no } + \sqrt{\no \Psi_{\T{\tiny B}}(t) \no } +\sqrt{\no \Psi_{\T{\tiny C}}(t) \no } \nonumber \\
& \hspace{2cm} + \sqrt{\no \Psi_{\T{\tiny D}}(t) \no } +\sqrt{\no \Psi_{\T{\tiny E}}(t) \no } +\sqrt{\no \Psi_{\T{\tiny F}}(t) \no } \Big) . \label{NORM:DIFFERENCE:BOUND}
\end{align}
Choosing $\varepsilon$ in \eqref{lemma:main-norm-of-E}-\eqref{lemma:main_lemma_2} small enough then proves the Theorem.
\end{proof}
The different wave functions $\Psi_{\T{\tiny X}}(t)$, $\text{\small{X}}\in \{\text{\small{A,B,C,D,E,F}}\}$ can be identified with the following collision histories of the tracer particle:
\begin{itemize}
\item [$\text{\small A:}$] single collisions which cause particle-hole excitations in the Fermi gas.
\item [$\text{\small B:}$] two collisions with the same particle, removing the particle-hole excitation which was caused in the first collision; the constant $E_{re}(\rho)$ cancels the contribution in which the second collision follows immediately after the first one.
\item [$\text{\small C:}$] two collisions with the same particle; the second collision scatters the lifted particle into another momentum above the Fermi edge.
\item [$\text{\small D:}$] two collisions; the second collision scatters a particle from below the Fermi edge into the hole that was created in the first collision.
\item [$\text{\small E:}$] two collisions with two different particles; causing two particle-hole excitations.
\item [$\text{\small F:}$] three collisions; three particle-hole excitations but also all possible recollisions with the already scattered particles; the different possibilities are listed in Section \ref{sec:bound_F}.
\end{itemize}
\noindent The rest of this section is devoted to the proof of Lemma \ref{lem:main_lemma}.

\subsection{Proof of Lemma~\ref{lem:main_lemma}}

\subsubsection{Preliminaries \label{sec:preliminaries}}

For $\varepsilon>0$, we define the two-dimensional index set 
\begin{align}
\I^\varepsilon (N,\rho)  & := \Big\{	(k,l)~:~ 1\le k \le N ,~  N+1 \le l,~  \vert  p_{k} - p_l \vert < \rho^\varepsilon \Big\} \subset \mathbb N^2,
\end{align}
and for $M\in \mathbb N$ the family of sets
\begin{align}
\label{index-sets:A_n}\I_n^{\varepsilon,M}(N,\rho)  & := \Big\{	( k,l )  \in \I^\varepsilon(N,\rho) \ : \ \rho^{-b_n} \le  \vert p_{l} \vert  -  \vert p_{k} \vert < \rho^{-b_{n+1}} \Big\}, ~~~ 0\le n \le M,
\end{align}
where 
$$b_0 = \infty, ~~~~ b_n=\frac{1}{2}-\frac{n-1}{M}\Big(\frac{1}{2}+\varepsilon \Big),~~~1\le n\le M.$$
For notational convenience, we omit from now on the $N$-, $\rho$-, $\varepsilon$- and also the $M$-dependence in the notation: $\I=\I^\varepsilon(N,\rho)$ and $\I_n=\I_n^{\varepsilon,M}(N,\rho)$. The index set $\I$ corresponds to the transitions that have to be controlled in \eqref{sketch_of_estimate}, i.e., collisions with momentum transfer smaller than $\rho^\varepsilon$. The set of pairs of momenta $\{ (p_k,  p_l) \in (2\pi / L )^2 \mathbb Z^4 : (k,l) \in \I_n \}$ are pairwise disjoint, and 
\begin{align}\label{ORTHOGONAL}
\bigcup_{n=0}^{M}  \Big\{ (p_k,  p_l) \in (2\pi / L )^2 \mathbb Z^4 : (k,l) \in \I_n \Big\} = \Big\{ (p_k,  p_l) \in (2\pi / L )^2 \mathbb Z^4 : (k,l) \in \I \Big\}. 
\end{align}
The distance of modulus between the occupied momentum $ p_{k}$ and the new momentum state $p_{l}$ increases in $\I_n$ for increasing $n$. With $2c=1/ (\frac{1}{2} + \varepsilon ) $,
\begin{align}
		\vert p_{l}\vert  - \vert p_{k} \vert  \ge  \rho^{-b_n} = \rho^{-\frac{1}{2}+\frac{n-1}{2cM} } ~~~ \text{for} ~~~ ( k,l ) \in \I_n, ~ 1\le n \le M.
\end{align}
Hence, also the energy shift increases,
\begin{align}
	\label{energy_difference}	E_{l}-E_{k} = \Big(\vert  p_{l}\vert  + \vert p_{k}\vert \Big)  \Big( \vert  p_{l}\vert  - \vert p_{k}\vert \Big)  \ge k_F \rho^{-b_n} = C \rho^{\frac{n-1}{2cM}} \ \ \ \text{for} ~~~ ( k,l)  \in \I_n,
\end{align}
$1\le n \le M$. $\rho^{-b_n}$ corresponds to the factor $\kappa(\rho)$ in (\ref{sketch_of_estimate}).

In the following lemma we state the key estimates that are used in order to prove Lemma \ref{lem:main_lemma}. Recall that $\theta(x)=1$ for $x\geq 0$ and zero otherwise. 
\begin{lemma}\label{lem:v_hat_estimates}
Assume $0<\varepsilon<\frac{1}{2}$ and $M,q\in \mathbb N$. Let $v(x)\in C^{\infty}_0(\mathbb T^2)  \cap C^{\infty}_0(\mathbb R^2)$  and $v^{l,\varepsilon}$, $v^{s,\varepsilon}$ defined as in \eqref{potential_large},\eqref{potential_small}. Then there exist positive constants $C$, $C_q$, $C_{\varepsilon,q}$ such that 
\end{lemma}
\vspace{-0.7cm}
\begin{align}
\label{lemma:v_hat_estimates_fluc}\lim_{\T{TD}} &    \Vol{4} \sum_{k=1}^N \sum_{l=N+1}^{\infty} \Big\vert \mathcal F[v] ( p_k - p_l ) \Big\vert^q  = C_q \rho^{\frac{1}{2}},\\
\label{lem:estimates_v_L} \lim_{\T{TD}} &  \Vol{4} \sum_{k=1}^{N} \sum_{l=N+1}^{\infty} \Big\vert \mathcal F[v^{l,\varepsilon}] (p_k  - p_l ) \Big\vert^q   \le C_{\varepsilon,q} \rho^{-1/\varepsilon}  ,\\
\label{lemma:v_hat_estimates_int_An} \lim_{\T{TD}} &  \Vol{4}\sum_{(k,l)\in \I_n }  \le C \rho^{\frac{1}{2}+\varepsilon-b_{n+1} } \Big( \rho^{-b_{n+1} } - \rho^{-b_n} \Big) \ \ \ \text{for} \ \ \ 0\le n \le M,\\
\label{lemma:v_hat_estimates_Ka} C \leq E_{re}(\rho) = \lim_{\T{TD}}&  \Vol{4} \sum_{k=1}^N\sum_{l=N+1}^{\infty} \frac{\big\vert \mathcal F[v] ( p_k  - p_l ) \big\vert^2 }{E_l-E_k} \theta\Big( \vert p_l \vert - \vert p_k \vert - \rho^{-\frac{1}{2}} \Big)  \le C\rho^{2\varepsilon} + C_\varepsilon \rho^{-1/\varepsilon},\\
\label{lemma:v_hat_estimates_p}\lim_{\T{TD}} &   \Vol{2} \sum_{k=1}^\infty \Big\vert \mathcal F[v]  (   p_k  -    p ) \Big\vert \le C \rho^{2\varepsilon} + C_\varepsilon \rho^{-1/\varepsilon} \ \ \ \ \text{for}\ \ \ p \in (2\pi/L)\mathbb Z^2.
\end{align}
The proof of the lemma is postponed to Section \ref{proof:lem2.3_2.4}. For notational ease, let us abbreviate the number of possible transitions that correspond to the set $\I_n$ by
\begin{align}
\mathcal V_n(N,\rho) := \Vol{4} \sum_{(k,l) \in \I_n} = \Vol{4} \sum_{k=1}^\infty \sum_{l=1}^\infty \chi_{\I_n}\big( ( k,l)\big),
\end{align}
$\chi_A:\mathbb N^2 \to \{0,1\}$ denoting the characteristic function, i.e., $\chi_A((k,l)) = 1$ whenever $(k,l) \in A \subset \mathbb N^2$, otherwise zero.
We readily obtain the following
\begin{corollary}\label{COROLLARY}Given the same assumptions as in Lemma \ref{lem:v_hat_estimates}, and setting $2c= 1/ (\frac{1}{2}+\varepsilon )$, there exists a constant $C>0$ such that
\end{corollary}
\vspace{-0.7cm}
\begin{align}
\lim_{\T{TD}} \mathcal V_0 (N,\rho) & \le C \rho^{-\frac{1}{2}+\varepsilon}, \label{Estimate:A0} \\
\lim_{\T{TD}} \Big( \rho^{-\left(\frac{n-1}{cM}\right) } \mathcal V_n(N,\rho) \Big) & \le C \rho^{-\frac{1}{2}+\varepsilon} \Big( \rho^{\frac{1}{cM}} -  \rho^{\frac{1}{2cM}} \Big)  \ \ \ \text{for} \ 1\le n\le M, \label{Estimate:An}
\end{align}
\begin{remark}
The decomposition of $\I$ into the sets $\I_n$ is optimal in the sense that the r.h.s.\ of all estimates behaves asymptotically almost the same, namely $\propto \rho^{-\frac{1}{2}}$ for small $\varepsilon$ and large $M$.
\end{remark}
\noindent The corollary follows from Lemma \ref{lem:v_hat_estimates} together with the definition of the $b_n$.\ Next, we summarize some straightforward bounds which makes the presentation of the proof of Lemma \ref{lem:main_lemma} more convenient.

\begin{lemma}\label{lem:Estimates_partial_k(s)}
Let $\varepsilon>0$ and $M\in \mathbb N$. Given the same assumptions as in Theorem \ref{thm:main_thm}, the following bounds hold for all $\tau_1,\tau_2\ge 0$ ($\chi_A$ denotes the characteristic function),
\begin{align}
	\label{lem:Estimates_partial_k(s)_1}  \chi_{\I_n} \big( ( k,l )\big)  \ \no \partial_{\tau_1} k_{kl}( \tau_1 ) \varphi_0 \no & \le C \rho^{2\varepsilon},\\
	\label{lem:Estimates_partial_k(s)_2}	
\chi_{\I_n} \big( ( k,l )\big)  \ \big(\no \partial_{\tau_2}k_{kl}(\tau_2)  g_{kl}(\tau_1) \varphi_0 \no + \no \partial_{\tau_2}k_{lk}(\tau_2)  g_{kl}(\tau_1) \varphi_0 \no \big)	& \le C \rho^{2\varepsilon},\\
\label{lem:Estimates_partial_k(s)_3}	
\chi_{\I} \big( ( k,l )\big)   \chi_{\I} \big( ( m,n )\big) \ \no \partial_{\tau_2} k_{lk}(\tau_2) \partial_{\tau_1} k_{mn}(\tau_1) \varphi_0 \no & \le C \rho^{4\varepsilon}, \\
	\label{lem:Estimates_partial_k(s)_4} \chi_{\I} \big( ( k,l )\big)  \chi_{\I} \big( ( m,n )\big)  \ \no \partial_{\tau_1} k_{kl}(\tau_1)k_{mn}(\tau_1) \varphi_0 \no & \le C \rho^{4\varepsilon}.
\end{align}
\end{lemma}
\noindent The proof is obtained by means of Stone's theorem and the assumption $\no \nabla^4 \varphi_0\no \le C$ (for more details, see Section \ref{proof:lem2.3_2.4}).\\

From now on, we always assume that $0<\varepsilon<1/8$ and denote $2c=(\frac{1}{2}+\varepsilon)^{-1}$. Furthermore, we will equally use the letter $\tau$ for indicating the dependence on the variables $\tau=(\tau_1,\tau_2)$ or $\tau=(\tau_1,\tau_2,\tau_3)$. For notational ease, let us also introduce for two real numbers $A$ and $B$: $$A \lesssim B~ \Leftrightarrow ~ \exists~ C>0\ s.t.~ A\le C B,$$
where the constant $C$ may depend on the supremum of $\hat v$ but is independent of any of the relevant parameters ($N$, $L$, $\rho$, $t$, $\varepsilon$ and $M$).

\subsubsection{Derivation of the Bound for $\no \Psi_{\T{A}}(t) \no_{\T{TD}}$ \label{sec:Derivation_A}} 

In order to bound $\no \Psi_{\T{A}}(t) \no_{\T{TD}}$, we first split the interaction potential into the contributions coming from small momentum transfer $\hat{v}^{s,\varepsilon}$ and those from large momentum transfer $\hat{v}^{\ell,\varepsilon}$. The $\hat{v}^{\ell,\varepsilon}$ will then be estimated using \eqref{lem:estimates_v_L}. For the $\hat{v}^{s,\varepsilon}$, we seperate the stationary points of the phase which can be estimated using \eqref{Estimate:A0}. For the nonstationary points, we do one partial integration in the time in order to be able to use \eqref{energy_difference} and \eqref{Estimate:An}.\\

Let $M\ge 1$, and for $0\le n\le M$,
\begin{align}
 \label{Psi_E^n(t)} \Psi_{\T{A}}^{s,n}(t)  & =  \Vol{2}   \sum_{(k_1,l_1) \in \I_n} \hat v^{s,\varepsilon}_{k_1l_1}     \int_0^t d\mu_2(\tau) D(\tau_2) \Big( g_{k_1l_1}(\tau_1)  \varphi_0 \Big) \otimes \Omega_0^{[l_1^*k_1]}, \\
 \Psi_{\T{A}}^\ell(\tau_1) & =  \Vol{2}   \sum_{k_1=1}^N \sum_{l_1=N+1}^{\infty} \hat v^{\ell,\varepsilon}_{k_1l_1}    \Big( g_{k_1l_1}(\tau_1) \varphi_0 \Big) \otimes \Omega_0^{[l_1^*k_1]} .
\end{align}
Using the identity in $\eqref{DECOMPOSITION:V-E}$ and $v_{k_1l_1} = v^{s,\varepsilon}_{k_1l_1} + v^{\ell,\varepsilon}_{k_1l_1}$, this leads to the following decomposition of $\Psi_{\T{A}}(t)$,
\begin{align}
\label{Psi:E(t)}\Psi_{\T{A}}(t) = E_{re}(\rho) \sum_{n=0}^M \Psi_{\T{A}}^{s,n}(t) + E_{re}(\rho) \int_0^{t}d\mu_2(\tau) D(\tau_2) \Psi_{\T{A}}^\ell(\tau_1).
\end{align}
We emphasize that $\Psi_{\T{A}}^{s,n}(t)$ depends on the choice of $M$ through the $M$-dependence of the sets $\I_n$, whereas $\Psi_{\T{A}}(t)$ and $\Psi_{\T{A}}^\ell(t)$ are both $M$-independent. Next, we estimate each term on the r.h.s.\ of \eqref{Psi:E(t)}. In the last one, we find, using $\langle \Omega_0^{[l_1^*k_1]} , \Omega_0^{[n_1^*m_1]}  \rangle = \delta_{l_1 n_1} \delta_{k_1m_1}$ for $ k_1,m_1\le N$, $N+1\le l_1,n_1$, as well as $\no g_{k_1l_1}(\tau_1)\varphi_0\no  = 1$,
\begin{align}
\no \Psi_{\T{A}}^\ell(\tau_1)\no^2 = \Vol{4}  \sum_{k_1=1}^N \sum_{l_1=N+1}^{\infty} \vert \hat v^{\ell,\varepsilon}_{k_1l_1}  \vert^2 \le C_\varepsilon \rho^{-1/(2\varepsilon)},
\end{align}
where the bound has been derived in \eqref{lem:estimates_v_L}. Recalling also \eqref{lemma:v_hat_estimates_Ka} which states that $E_{re} (\rho) \lesssim \rho^{2\varepsilon} + C_\epsilon \rho^{-1/\epsilon}$, and using unitarity of $D(s)$, one obtains for the last term in \eqref{Psi:E(t)}
\begin{align}
\bno E_{re} (\rho) \int_0^{t}d\mu_2(\tau) D(\tau_2) \Psi_{\T{A}}^\ell (\tau_1) \bno_{\T{TD} }\lesssim t^2 \Big( \rho^{2\varepsilon} + C_\epsilon \rho^{-1/\epsilon} \Big) \rho^{-1/(2\varepsilon)}.
\end{align}
In $\Psi_{\T{A}}^{s,0}(t)$, we need to estimate the norm
\begin{align}
\bno \Vol{2}   \sum_{(k_1,l_1) \in \I_0} \hat v^{s,\varepsilon}_{k_1l_1}  \Big( g_{k_1l_1}(\tau_1)   \varphi_0 \Big) \otimes \Omega_0^{[l_1^*k_1]} \bno^2 \lesssim  \Vol{4}   \sum_{( k_1, l_1 ) \in \I_0}  =  \V{0}.
\end{align}
The remaining expression, i.e., the number of transitions corresponding to the set $\I_0$, has been estimated in \eqref{Estimate:A0}. Thus, 
$$\no \Psi_{\T{A}}^{s,0}(t) \no_{\T{TD}} \lesssim  t^2 \rho^{-\frac{1}{4}+ \frac{1}{2}\varepsilon }.$$

\begin{lemma}\label{lem:R_S_A_est}
Let $\Psi_{\T{A}}^{s,n}(t)$ as in (\ref{Psi_E^n(t)}). Then, under the same assumptions as in Theorem~\ref{thm:main_thm}, there exists a constant $C>0$ such that\end{lemma}
\vspace{-0.9cm}
\begin{align}
 \no \Psi_{\T{A}}^{s,n}(t) \no_{\T{TD}} & \lesssim (1+t)^2  \rho^{-\frac{1}{4}+\frac{5}{2}\varepsilon } \sqrt{\rho^{\frac{1}{cM}} -  \rho^{\frac{1}{2cM}} } , \hspace{0.5cm} 1\le n\le M,
\end{align}
\textit{holds for all $t>0$.}\\

One can now use that for $M=\floor{\ln \rho}$ (the largest integer smaller than the number $\ln \rho$),
\begin{align}
\sum_{n=1}^M \no \Psi_{\T{A}}^{s,n}(t) \no_{\T{TD}}  \lesssim (1+t)^2  \rho^{-\frac{1}{4}+\frac{5}{2} \varepsilon } M \rho^{\frac{1}{2cM}} \lesssim (1+t)^2  \rho^{-\frac{1}{4}+3\varepsilon} \label{limit_M}
\end{align}
because $M \rho^{\frac{1}{2cM}}\le \ln \rho \cdot e^{\frac{1}{2c}} \lesssim \rho^{\frac{1}{2}\varepsilon}$ for any $\varepsilon>0$. This proves the bound for $\no \Psi_{\T{A}}(t) \no_{\T{TD}}$ in (\ref{lemma:main-norm-of-E}).\ That taking first the thermodynamic limit and then $M=\floor{\ln \rho}$ is unproblematic (even when $\rho$ tends to $\infty$) is summarized in the following
\begin{remark}\label{INTERCHANGE:LIMITS}$\Psi_{\T{A}}(t)$  as well as $\Psi_{\T{A}}^\ell (\tau)$ in \eqref{Psi:E(t)} are both $M$-independent. There is thus no need to interchange the order of the two limits. One first takes the thermodynamic limit on both sides and then passes to the limit of large $M$.\ Since only the r.h.s.\ of \eqref{Psi:E(t)} depends on the choice of $M$, this provides the desired estimate.
\end{remark}

\begin{proof}[Proof of Lemma~\ref{lem:R_S_A_est}]
We first decompose each of the $\Psi_{\T{A}}^{s,n}(t)$  via partial integration in $\tau_1$. For that, we recall $g_{k_1l_1}(\tau_1) = e^{ i (E_{l_1} - E_{k_1}) \tau_1 } k_{k_1l_1}(\tau_1)$ and rewrite (for $n\ge 1$)
\begin{align}
\Psi_{\T{A}}^{s,n} (t) &  = \Vol{2}  \sum_{( k_1 ,l_1 ) \in \I_n}  \hat v^{s,\varepsilon}_{k_1l_1} \int_0^{t} d\tau_1  \Big( \frac{ \partial_{\tau_1} e^{i(E_{l_1}-E_{k_1})\tau_1}  }{i(E_{l_1}-E_{k_1} )} \Big) \int_0^{\tau_1} d\tau_2 D(\tau_2) \Big( k_{k_1l_1}(\tau_1) \varphi_0 \Big) \otimes \Omega_0^{[l_1^*,k_1]}.\nonumber
\end{align}
Partial integration in $\tau_1$ leads to
\begin{align}\label{DECOMPOSITION:ASN}
  \Psi_{\T{A}}^{s,n}(t) = \int_0^t d\tau D(\tau_1) \Psi_{\T{A,1}}^{s,n}(t,\tau_1) + \int_0^t d\mu_2(\tau) D(\tau_2) \Psi_{\T {A,2}}^{s,n}(\tau_1),
\end{align}
where
\begin{align}
\Psi_{\T{A,1}}^{s,n}(t,\tau_1)
& = \Vol{2}  \sum_{( k_1 ,l_1 ) \in \I_n}  \hat v^{s,\varepsilon}_{k_1l_1} \Big( \frac{ g_{k_1l_1}(t)  - g_{k_1l_1} (\tau_1) }{i(E_{l_1}-E_{k_1})}  \varphi_0 \Big) \otimes \Omega_0^{[l_1^*,k_1]},\\
\Psi_{\T{A,2}}^{s,n}(\tau_1) & =  \Vol{2}  \sum_{( k_1 ,l_1 )\in \I_n}  \hat v^{s,\varepsilon}_{k_1l_1} \Big(   \frac{ e^{i(E_{l_1}-E_{k_1})\tau_1} \partial_{\tau_1} k _{k_1l_1}(\tau_1)}{i(E_{l_1}-E_{k_1})}   \varphi_0 \Big) \otimes \Omega_0^{[l_1^*,k_1]}.
\end{align}
Next, we estimate the the norms of $\Psi_{\T{A,1}}^{s,n}(\tau)$ and $\Psi_{\T{A,2}}^{s,n}(\tau)$:
\begin{align}\label{EXAMPLE:THREE:D}
\no \Psi_{\T{A,1}}^{s,n}(t,\tau_1) \no^2 & \lesssim \Vol{4} \sum_{( k_1, l_1 ) \in \I_n} \frac{1}{(E_{l_1}- E_{k_1})^2} \lesssim \Big( \rho^{-\left( \frac{n-1}{cM} \right) } \V{n} \Big) ,
\end{align}
where we have used \eqref{energy_difference},
\begin{align}
\no \Psi_{\T{A,2}}^{s,n}(\tau_1) \no^2  \lesssim \Vol{4}
  \sum_{( k_1 , l_1 ) \in \I_n}  \frac{\no \partial_{\tau_1} k _{k_1l_1}(\tau_1) \varphi_0 \no^2}{ ( E_{l_1} - E_{k_1})^2}  \lesssim \rho^{4\varepsilon}  \Big( \rho^{-\left( \frac{n-1}{cM} \right) } \V{n} \Big),
\end{align}
where we have made in addition use of \eqref{lem:Estimates_partial_k(s)_1}.\
The remaining expressions have been estimated in \eqref{Estimate:An}.
\end{proof}

\subsubsection{\label{bound:Lemma1.1.b}Derivation of the Bound for $\no\Psi_{\T{B}}(t) \no_{\T{TD}}$}

In the estimate for $\no\Psi_{\T{B}}(t) \no_{\T{TD}}$, we first identify the contribution in $i \int_0^t d\mu_2(\tau)  D (\tau_2) \Psi_{1}(\tau)$ that cancels with the energy correction $E_{re}(\rho) \int_0^{t}d\tau_1  D(\tau_1)\Psi_0$. The remaining terms will then be estimated using similar techniques as in Section \ref{sec:Derivation_A}.\\

For $0\le n\le M$ ($M\ge 1$), let
\begin{align}
\Psi_{\T{B}}^{s,0}(t) & =  \Vol{4}   \sum_{(k_1, l_1) \in \I_0} \vert \hat v^{s,\varepsilon}_{k_1l_1}\vert^2  i \int_0^t d\mu_2(\tau) D(\tau_2)g_{l_1k_1}(\tau_2) g_{k_1l_1}(\tau_1) \Psi_0,\\
\Psi_{\T{B,1}}^{s,n}(t) & = \Vol{4}     \sum_{(k_1, l_1) \in \I_n}  \vert \hat v^{s,\varepsilon}_{k_1l_1} \vert^2  \int_0^td\tau_2  D(\tau_2) \frac{ g_{l_1k_1}(\tau_2) g_{k_1l_1}(t) }{(E_{l_1}-E_{k_1})} \Psi_0, \label{BS:1}\\
\Psi_{\T{B,2}}^{s,n}(t) &  =  \Vol{4}     \sum_{(k_1, l_1)\in \I_n}   \vert \hat v^{s,\varepsilon}_{k_1l_1} \vert^2  \int_0^td \tau_2  D(\tau_2) \frac{ k_{l_1k_1}(\tau_2) k_{k_1l_1}(\tau_2)}{(E_{l_1}-E_{k_1})} \Psi_0, \\
\Psi_{\T{B,3}}^{s,n}(t) & =   \Vol{4}     \sum_{(k_1, l_1)\in \I_n}   \vert \hat v^{s,\varepsilon}_{k_1l_1} \vert^2  \int_0^t d\mu_2(\tau)  D(\tau_2) \frac{ g_{l_1k_1}(\tau_2) e^{i(E_{l_1}-E_{k_1}) \tau_1} \partial_{\tau_1} k_{k_1l_1}(\tau_1)}{(E_{l_1}-E_{k_1})} \Psi_0, \label{BS:3}\\
\Psi_{\T{B}}^\ell(\tau) & =  \Vol{4}    \sum_{l_1=1}^N \sum_{k_1=N+1}^{\infty}\vert \hat v^{\ell,\varepsilon}_{k_1l_1} \vert^2   g_{l_1k_1}(\tau_2) g_{k_1l_1}(\tau_1) \Psi_0.
\end{align}
Via partial integration, this leads to the identity
\begin{align}
 \Psi_{\T{B}}(t) & = E_{re}(\rho)  \int_0^td\tau_1     D(\tau_1)\Psi_0 - \sum_{n=1}^M  \Psi_{\T{B,2}}^{s,n}(t)  \nonumber\\
 & \ \ \ \ + \Psi_{\T{B}}^{s,0}(t) + \sum_{n=1}^M \Big[ \Psi_{\T{B,1}}^{s,n}(t) -  \Psi_{\T{B,3}}^{s,n}(t) \Big] + i\int_0^td\mu_2(\tau)  D(\tau_2) \Psi_{\T{B}}^\ell (\tau).
\end{align}
The first step in estimating the r.h.s.\ is to note that for any $M\ge 1$, the thermodynamic limit of the upper line is bounded in terms of
\begin{align}
\lno E_{re}(\rho)   \int_0^td\tau_1    D(\tau_1)\Psi_0 - \sum_{n=1}^M \Psi_{\T{B,2}}^{s,n}(t) \rno_{\T{TD}} \le C_\varepsilon \rho^{-1/\varepsilon}. \label{RECOLLISIONS}
\end{align}
In $\Psi_{\T{B,2}}^{s,n}(t)$, the fluctuation does not propagate in time (note that $k_{l_1k_1}(\tau_2) k_{k_1l_1}(\tau_2)=1$), and the factor $1/(E_{l_1}-E_{k_1})$ does not make this term small enough. This collision history corresponds to immediate recollisions with the same particle removing the particle-hole excitation which was created in the first scattering. It needs to be canceled directly by the next-to-leading order energy correction in $H^{\text{mf}}$. To see that the above estimate is true, we rewrite
\begin{align} 
 \sum_{n=1}^M \Psi_{\T{B,2}}^{s,n}(t) &  = \sum_{n=1}^M    \Vol{4}  \sum_{(k_1, l_1 ) \in \I_n} \frac{ \vert \hat v^{s,\varepsilon}_{k_1l_1} \vert^2 }{(E_{l_1}-E_{k_1})} \int_0^t d\tau_2  D(\tau_2)   \Psi_0 \\
&  =   \Vol{4}  \sum_{k_1=1}^N  \sum_{l_1=N+1}^\infty  \frac{ \vert \hat v^{ }_{k_1l_1} \vert^2 }{ (E_{l_1}-E_{k_1})}   \theta \Big(\rho^{\varepsilon} -
\vert p_{k_1} - p_{l_1} \vert \Big) \theta\Big( \vert p_{l_1} \vert  - \vert  p_{k_1} \vert - \rho^{-\frac{1}{2} } \Big) \int_0^t	d\tau_2  D(\tau_2) \Psi_0,\nonumber
\end{align}
and thus, recalling definition \eqref{def:F}, we need to estimate
\begin{align}
&\Big\vert E_{re}(\rho) -  \lim_{\T{TD}} \Vol{4}  \sum_{k_1=1}^N  \sum_{l_1=N+1}^\infty  \frac{ \vert \hat v^{ }_{k_1l_1} \vert^2 }{ (E_{l_1}-E_{k_1})}   \theta \Big(\rho^{\varepsilon} -
\vert p_{k_1} - p_{l_1} \vert \Big) \theta\Big( \vert p_{l_1} \vert  - \vert  p_{k_1} \vert - \rho^{-\frac{1}{2} } \Big) \Big\vert \nonumber\\
& =  \lim_{\T{TD}} \Vol{4}  \sum_{k_1=1}^N  \sum_{l_1=N+1}^\infty  \frac{ \vert \hat v^{ }_{k_1l_1} \vert^2 }{ (E_{l_1}-E_{k_1})}   \theta \Big(-\rho^{\varepsilon} +
\vert p_{k_1} - p_{l_1} \vert \Big) \theta\Big( \vert p_{l_1} \vert  - \vert  p_{k_1} \vert - \rho^{-\frac{1}{2} } \Big)  \nonumber \\
& = \lim_{\T{TD}}  \Vol{4}  \sum_{k_1=1}^N  \sum_{l_1=N+1}^\infty  \frac{ \vert \hat v^{\ell,\epsilon }_{k_1l_1} \vert^2 }{ (E_{l_1}-E_{k_1})}    \theta\Big( \vert p_{l_1} \vert  - \vert  p_{k_1} \vert - \rho^{-\frac{1}{2} } \Big)   \nonumber\\
& \lesssim \lim_{\T{TD}} \Vol{4}  \sum_{k_1=1}^N  \sum_{l_1=N+1}^\infty \vert \hat v^{\ell,\epsilon }_{k_1l_1} \vert^2. \label{RECOLLISIONS:2}
\end{align}
By \eqref{lem:estimates_v_L}, one obtains \eqref{RECOLLISIONS}. Note that in the last step, we have used $E_{l_1}-E_{k_1} = (\vert p_{l_1} \vert - \vert p_{k_1} \vert) (\vert p_{l_1} \vert + \vert p_{k_1} \vert) \ge \rho^{-\frac{1}{2}} k_F = C$ (since $\vert p_{l_1} \vert \ge k_F$).
\\
\\
It follows with \eqref{Estimate:A0} that
\begin{align}
\no \Psi_{\T{B}}^{s,0}(t) \no_{\T{TD}} \lesssim  t^2 \lim_{\T{TD}} \V{0} \lesssim t^2 \rho^{-\frac{1}{2}+\varepsilon},
\end{align}
as well as with \eqref{lem:estimates_v_L},
\begin{align}
\no \Psi_{\T{B}}^\ell (\tau) \no_{\T{TD}} \lesssim t^2 \lim_{\T{TD}} \Vol{4}  \sum_{k_1=1}^N \sum_{l_1=N+1}^{\infty}\vert \hat v^{\ell,\varepsilon}_{k_1l_1} \vert^2 \le C_\varepsilon t^2 \rho^{-1/ \varepsilon }.
\end{align}
The bounds for the remaining wave functions are summarized in
\begin{lemma}\label{PSIB:S:LEMMA}
Let  $\Psi_{\T{B,1}}^{s,n}(t)$ and $\Psi_{\T{B,3}}^{s,n}(t)$ as in \eqref{BS:1} and \eqref{BS:3}. Then, under the same assumptions as in Theorem~\ref{thm:main_thm}, there exists a positive constant $C$ such that\end{lemma}
\vspace{-0.7cm}
\begin{align}
\no \Psi_{\T{B,1}}^{s,n}(t)  \no_{\T{TD}} + \no \Psi_{\T{B,3}}^{s,n}(t)  \no_{\T{TD}} \lesssim & (1+t)^2 \rho^{-\frac{1}{2} + \varepsilon} \Big( \rho^{\frac{1}{4}} + C_\varepsilon \rho^{-1/\varepsilon} \Big) \Big( \rho^{\frac{1}{cM}} - \rho^{\frac{1}{2cM}} \Big) ,\hspace{0.5cm} 1\le n \le M,
\end{align}
\textit{holds for all $t\ge 0$.}\\

Taking the sum of all terms, passing to the thermodynamic limit and then choosing again $M=\floor{\ln \rho} \lesssim \rho^{\varepsilon}$, cf.\ \eqref{limit_M}, one finds
\begin{align}
 \sum_{n=1}^{M} \Big( \no \Psi_{\T{B,1}}^{s,n}(t)   \no_{\T{TD}} + \no  \Psi_{\T{B,3}}^{s,n}(t)  \no_{\T{TD}} \Big) & \lesssim (1+t)^2 \Big(  \rho^{-\frac{1}{4} + 2\varepsilon}  + C_\varepsilon \rho^{ -\frac{1}{2} + 2\varepsilon - \frac{1}{\varepsilon} } \Big).
\end{align}
This proves the bound for $\no \Psi_{\T{B}}(t)\no_{\T{TD}}$ in \eqref{lemma:main-norm-of-A1} (recall Remark \ref{INTERCHANGE:LIMITS} and the fact that $\Psi_{\T{B}}(t)$ and $\Psi_{\T{B}}^\ell(\tau)$ do both not depend on $M$).\\

In order to prove Lemma \ref{PSIB:S:LEMMA}, we need the following estimate.
\begin{lemma} \label{D:DERIVATIVE} Let $\psi \in L^2(\mathbb T^2)$ and $\Omega_0$ as in \eqref{FERMI:SEA}.\ Then, there exists a positive constant $C$ such that
\end{lemma}
\vspace{-0.6cm}
\begin{align}\bno \Big( \partial_\tau  D(\tau) \Big) \Big(\psi \otimes \Omega_0 \Big) \bno^2 \lesssim \no \psi\no^2 \Big( E_{re}(\rho)^2 +	\Vol{4}  \sum_{k=1}^N \sum_{l=N+1}^\infty \vert \hat v_{kl} \vert^2 \Big)
\end{align}
\textit{holds for all $\tau>0$.}
\begin{proof} Using $i\partial_\tau D(\tau) = U(-\tau) (H^{\text{mf}} - H) U^{\text{mf}}(\tau)$, $\no \Omega_0\no = 1$ and some basic algebra similar as in \eqref{DECOMPOSITION:V-E} and \eqref{FIRST:ORDER:DEVIATIONS}, it follows easily that
\begin{align}
\bno \Big( \partial_\tau  D(\tau) \Big) \Big(\psi \otimes \Omega_0 \Big) \bno^2 & \lesssim  \no \psi \no^2 ~ E_{re}(\rho)^2  + \bno (V- E) \Big( \psi \otimes \Omega_0 \Big)\bno^2  \nonumber\\
& \lesssim \no \psi \no^2 ~ \Big( E_{re}(\rho)^2 + \Vol{4} \sum_{k=1}^N\sum_{l=N+1}^\infty \vert \hat v_{k l }\vert^2  \Big).
\end{align}
\end{proof}
\begin{proof}[Proof of Lemma \ref{PSIB:S:LEMMA}]
Let
\begin{align}
\Psi_{\T{B,11}}^{s,n}(t) & =  \Vol{4} \sum_{ ( k_1, l_1 ) \in \I_n}  \vert \hat v^{s,\varepsilon}_{k_1l_1} \vert^2  \frac{ D(t) - g_{l_1k_1}(0) g_{k_1l_1} (t) }{i(E_{l_1}-E_{k_1})^2} \Psi_0,\\
 \Psi_{\T{B,12}}^{s,n}(t,\tau_2) &  =  \Vol{4}  \sum_{ ( k_1, l_1 ) \in \I_n}   \vert \hat v^{s,\varepsilon}_{k_1l_1} \vert^2 \frac{ e^{i(E_{l_1}-E_{k_1} )\tau_2 } \partial_{\tau_2} \big( D(\tau_2)   k_{l_1k_1}(\tau_2)\big)  g_{k_1l_1}(t) }{i(E_{l_1}-E_{k_1})^2}\Psi_0,
\end{align}
$1\le n \le M$. One verifies by means of partial integration that
\begin{align} \Psi_{\T{B,1}}^{s,n}(t) = \Psi_{\T{B,11}}^{s,n}(t) - \int_0^td\tau_2 \Psi_{\T{B,12}}^{s,n}( t,\tau_2).
\end{align}
By \eqref{energy_difference},
\begin{align}
\no  \Psi_{\T{B,11}}^{s,n}(t) \no \lesssim \Big( \rho^{-\left( \frac{n-1}{cM} \right) } \V{n} \Big),
\end{align}
and, by using Lemma \ref{D:DERIVATIVE} together with \eqref{Estimate:An}, \eqref{lem:Estimates_partial_k(s)_2}, 
\begin{align}
\no \Psi_{\T{B,12}}^{s,n}(t,\tau_2) \no & \lesssim  \rho^{-\left( \frac{n-1}{cM} \right) } \Vol{4} \sum_{( k_1 , l_1 ) \in \I_n} \Big( \bno \Big( \partial_{\tau_2}  D(\tau_2) \Big)  k_{l_1k_1}(\tau_2) g_{k_1l_1}(t) \Psi_0 \bno 
 +  \bno \partial_{\tau_2}  k_{l_1k_1}(\tau_2) g_{k_1l_1}(t) \varphi_0 \bno \Big) \nonumber\\
& \lesssim \Big( \rho^{-\left( \frac{n-1}{c M} \right) } \V{n} \Big)  \Big( E_{re}(\rho)  + \Big(\Vol{4}  \sum_{k=1}^N \sum_{l=N+1}^\infty \vert \hat v_{kl} \vert^2 \Big)^{\frac{1}{2}} + \rho^{2\varepsilon} \Big).
\end{align}
Let similarly
\begin{align}
 \Psi_{\T{B,31}}^{s,n} (\tau_1) &  =    \Vol{4}    \sum_{(k_1,l_1)\in \I_n}   \vert \hat v^{s,\varepsilon}_{k_1l_1} \vert^2 \frac{ \Big( D(\tau_1) g_{l_1k_1}(\tau_1) - g_{l_1k_1}(0) \Big) e^{i(E_{l_1}-E_{k_1})\tau_1} \partial_{\tau_1}k_{k_1l_1}(\tau_1)}{i(E_{l_1}-E_{k_1})^2}  \Psi_0, \\
  \Psi_{\T{B,32}}^{s,n}(\tau) & =    \Vol{4}   \sum_{(k_1,l_1) \in \I_n}   \vert \hat v^{s,\varepsilon}_{k_1l_1} \vert^2 \frac{  e^{i(E_{k_1}-E_{l_1})(\tau_2-\tau_1)}\partial_{\tau_2} \Big( D(\tau_2) k_{l_1k_1}(\tau_2) \Big) \partial_{\tau_1}k_{k_1l_1}(\tau_1)}{i (E_{l_1}-E_{k_1})^2} \Psi_0, 
\end{align}
such that by partial integration in $\tau_2$,
\begin{align} \Psi_{\T{B,3}}^{s,n}(t)=  \int_0^t d\tau_1  \Psi_{\T{B,31}}^{s,n}(\tau_1) - \int_0^td \mu_2(\tau)  \Psi_{\T{B,32}}^{s,n}(\tau).
\end{align} 
From \eqref{energy_difference} together with \eqref{lem:Estimates_partial_k(s)_1}, it follows that
\begin{align}
\no \Psi_{\T{B,31}}^{s,n}(\tau)\no & \lesssim  \Vol{4} \sum_{( k_1, l_1) \in \I_n} \frac{ \no \partial_{\tau_1}k_{k_1l_1}(\tau_1) \varphi_0 \no }{(E_{l_1}-E_{k_1} )^2} \lesssim \rho^{2\varepsilon} \Big(\rho^{-\left( \frac{n-1}{cM}\right)} \V{n} \Big)  ,
\end{align}
as well as in combination with Lemma \ref{D:DERIVATIVE} and \eqref{lem:Estimates_partial_k(s)_2},
\begin{align}
\no \Psi_{\T{B,32}}^{s,n}(t)\no & \lesssim \rho^{-\left( \frac{n-1}{cM} \right) } \Vol{4} \sum_{ ( k_1 , l_1 ) \in \I_n}   \Big( \bno \Big( \partial_{\tau_2}  D(\tau_2) \Big) k_{l_1k_1}(\tau_2) k_{k_1l_1}(\tau_1) \Psi_0 \bno + \rho^{2\varepsilon} \Big) \\
& \lesssim \Big(  \rho^{-\left( \frac{n-1}{cM} \right)}  \V{n} \Big)  \Big( E_{re}(\rho)  + \Big( \Vol{4}  \sum_{k=1}^N \sum_{l=N+1}^\infty \vert \hat v_{kl} \vert^2 \Big)^{\frac{1}{2}} + \rho^{2\varepsilon} \Big).
\end{align}
Lemma \ref{PSIB:S:LEMMA} then follows from \eqref{lemma:v_hat_estimates_fluc}, \eqref{lemma:v_hat_estimates_Ka} and \eqref{Estimate:An}.
\end{proof}

\subsubsection{Derivation of the Bound for $ \no \Psi_{\T{C}}(t) \no_{\T{TD}} $\label{sec:derivation_C}}
For estimating $ \no \Psi_{\T{C}}(t) \no_{\T{TD}}$ there is the extra difficulty that there is an additional sum appearing. This can be dealt with by using \eqref{lemma:v_hat_estimates_p}.\\

Let $M\ge 1$. We define for $0\le n \le M$,
\begin{align}\label{PSI:CS}
\Psi_{\T{C}}^{s,n}(t) = & \Vol{4}  \sum_{(k_1,l_1)\in \I_n} \sum_{k_2=1}^N \hat v_{k_2k_1} \hat v_{k_1l_1}^{s,\varepsilon}  \int_0^t d\mu_2(\tau)  D(\tau_2)  \Big( g_{k_2k_1}(\tau_2) g_{k_1l_1}(\tau_1)  \varphi_0 \Big) \otimes \Omega_0^{[k_2l_1^*]}, \\ 
\Psi_{\T{C}}^\ell(\tau) = & \Vol{4}   \sum_{l_1=N+1}^\infty \sum_{k_1,k_2=1}^N \hat v_{k_2k_1} \hat v_{k_1l_1}^{\ell, \varepsilon} \Big(  g_{k_2k_1}(\tau_2) g_{k_1l_1}(\tau_1)   \varphi_0 \Big) \otimes \Omega_0^{[k_2l_1^*]},
\end{align}
such that
\begin{align}
\Psi_{\T{C}}(t) = \sum_{n=0}^M \Psi_{\T{C}}^{s,n}(t)  + \int_0^t d\mu_2(\tau) D(\tau_2)  \Psi_{\T{C}}^\ell(\tau). \label{PSI:C}
\end{align}
Note that $\lsp \Omega_0^{[k_2l_1^*]} , \Omega_0^{[m_2n_1^*]} \rsp = \delta_{k_2m_2} \delta_{l_1n_1}$ for $k_2,m_2\le N$, $N+1 \le l_1,n_1$. Using this and $|\hat{v} \vert \leq C $, one finds
\begin{align}
\no \Psi_{\T{C}}^\ell (\tau) \no^2 & \lesssim   \Vol{8} \sum_{k_1=1}^N \sum_{l_1=N+1}^\infty \vert \hat v_{k_1l_1}^{\ell,\varepsilon} \vert  \sum_{k_2=1}^N  \vert \hat v_{k_2k_1} \vert  \sum_{m_1=1}^N \sum_{n_1=N+1}^\infty  \vert \hat v_{m_1n_1}^{\ell,\varepsilon} \vert  \sum_{m_2=1}^N  \vert \hat v_{m_2m_1} \vert \lsp \Omega_0^{[k_2l_1^*]} , \Omega_0^{[m_2n_1^*]} \rsp \\
& \lesssim \Vol{4} \sum_{k_1 =1}^N \sum_{l_1 =N+1}^\infty \vert \hat v_{k_1l_1}^{\ell , \varepsilon} \vert  \Vol{2} \sum_{k_2=1}^N  \vert \hat v_{k_2k_1} \vert  \Vol{2} \sum_{m_1=1}^N \vert \hat v_{m_1l_1}^{\ell ,\varepsilon} \vert.
\end{align}
Then, by means of \eqref{lem:estimates_v_L} and \eqref{lemma:v_hat_estimates_p}, $\no \Psi_{\T{C}}^\ell (\tau)\no_{\T{TD}} \le  C_\varepsilon \rho^{\varepsilon-\frac{1}{\varepsilon}}$, and thus
\begin{align}
\bno \int_0^t d\mu_2(\tau)  D(\tau_2)  \Psi_{\T{C}}^\ell(\tau) \bno_{\T{TD}}\le C_\varepsilon t^2 \rho^{\varepsilon-\frac{1}{\varepsilon}}.
\end{align}
Similarly, we can estimate the norm in $\Psi_{\T{C}}^{s,0}(t)$. Using \eqref{lemma:v_hat_estimates_p}, we find
\begin{align}
& \bno \Vol{4}  \sum_{( k_1, l_1 ) \in \I_n} \sum_{k_2=1}^N \hat v_{k_2k_1} \hat v_{k_1l_1}^{s,\varepsilon}   g_{k_2k_1}(\tau_2) g_{k_1l_1}(\tau_1)  \varphi_0\otimes \Omega_0^{[k_2l_1^*]} \bno^2 \nonumber\\
& \hspace{2cm} \lesssim  \Vol{4} \sum_{( k_1, l_1 ) \in \I_0}  \Vol{2}  \sum_{k_2=1}^N  \vert \hat v_{k_2k_1} \vert   \Vol{2}  \sum_{m_1=1}^N \vert \hat v_{m_1n_1}^{s,\varepsilon} \vert  \lesssim \V{0}  \Big(\rho^{2\varepsilon}+C_\varepsilon\rho^{-1/\varepsilon} \Big) \rho^{2\varepsilon} .
\end{align}
In combination with \eqref{Estimate:A0} this gives
\begin{align}
\no \Psi_{\T{C}}^{s,0}(t) \no_{\T{TD}} \lesssim t^2  \rho^{-\frac{1}{4}+\frac{3}{2}\varepsilon}   \Big( \rho^{\varepsilon} + C_\varepsilon\rho^{-1/(2 \varepsilon ) } \Big).
\end{align}

\begin{lemma}\label{LEMMA:PSI:C:S} Let  $\Psi_{\T{C}}^{s,n}(t)$ as in \eqref{PSI:CS}. Then, under the same assumptions as in Theorem~\ref{thm:main_thm}, there exists a positive constant $C$ such that\end{lemma}
\vspace{-0.7cm}
\begin{align}
\no \Psi_{\T{C}}^{s,n}(t) \no_{\T{TD}} & \lesssim (1+t)^2 \rho^{-\frac{1}{4}+\frac{7}{2}\varepsilon} \sqrt{ \rho^{\frac{1}{cM}}-\rho^{\frac{1}{2cM}} } \Big( \rho^{\varepsilon} +  C_\varepsilon \rho^{-1/(2\varepsilon) } \Big)  , \hspace{0.5cm} 1\le n \le M, 
\end{align}
\textit{holds for all $t>0$.}\\

Here, we can use again that $M \rho^{\frac{1}{2cM}} \lesssim \rho^{\frac{1}{2}\varepsilon}$ for $M=\floor{\ln \rho}$, and thus obtain
\begin{align}
	\sum_{n=1}^{M} \no   \Psi_{\T{C}}^{s,n}(t) \no_{\T{TD}} & \lesssim (1+t)^2 \rho^{-\frac{1}{4}+2\varepsilon}  \Big( \rho^{\varepsilon} +  C_\varepsilon \rho^{-1/(2\varepsilon)} \Big).
\end{align}
This proves the bound for $\no \Psi_{\T{C}}(t) \no_{\T{TD}}$ in \eqref{lemma:main-norm-of-C}.} 
\begin{proof}[Proof of Lemma \ref{LEMMA:PSI:C:S}] We define for $1\le n \le M$,
\begin{align}
\Psi_{\T{C,1}}^{s,n}(t,\tau_1) & =  \Vol{4}    \sum_{( k_1, l_1) \in \I_n} \sum_{k_2=1}^N \hat v_{k_2k_1}  \hat  v_{k_1l_1}^{s,\varepsilon} \Big( g_{k_2k_1}(\tau_1) \frac{ g_{k_1l_1}(t) - g_{k_1l_1}(\tau_1)  }{ i(E_{l_1}-E_{k_1}) }   \varphi_0 \Big) \otimes \Omega_0^{[k_2l_1^*]}, \\
\Psi_{\T{C,2}}^{s,n}(\tau)&  =   \Vol{4}    \sum_{( k_1, l_1) \in \I_n} \sum_{k_2=1}^N \hat  v_{k_2k_1}  \hat  v_{k_1l_1}^{s,\varepsilon} \Big( g_{k_2k_1}(\tau_2) \frac{ e^{i(E_{l_1}-E_{k_1})\tau_1} \partial_{\tau_1} k_{k_1l_1}(\tau_1) }{ i(E_{l_1}-E_{k_1}) }  \varphi_0 \Big) \otimes \Omega_0^{[k_2l_1^*]},
\end{align}
and find via partial integration,
\begin{align}
\Psi_{\T{C}}^{(n)}(t) = \int_0^t d\tau_1 D(\tau_1) \Psi_{\T{C,1}}^{s,n}(t,\tau_1) -  \int_0^t d\mu_2(\tau) D(\tau_2) \Psi_{\T{C,2}}^{s,n}(\tau).
\end{align} 
For estimating the wave functions on the r.h.s., we use again $\lsp \Omega_0^{[k_2l_1^*]} , \Omega_0^{[m_2n_1^*]} \rsp = \delta_{k_2m_2} \delta_{l_1n_1}$, \eqref{energy_difference} and \eqref{lemma:v_hat_estimates_p},
\begin{align}
\no \Psi_{\T{C,1}}^{s,n}(\tau) \no^2 & \lesssim 
     \Vol{4}   \sum_{( k_1,  l_1) \in \I_n}  \rho^{ -\left( \frac{n-1}{cM} \right)}  \Vol{2}  \sum_{k_2=1}^N \vert \hat v_{k_2k_1} \vert 	\Vol{2}     \sum_{m_1=1}^N \vert \hat v_{m_1l_1}^{s,\varepsilon} \vert \nonumber\\
& \lesssim \rho^{2\varepsilon} \Big( \rho^{ -\left( \frac{n-1}{cM} \right)} \V{n} \Big) \Big(\rho^{2\varepsilon}+C_\varepsilon\rho^{-1/\varepsilon} \Big)  .
\end{align}
Similarly, using in addition \eqref{lem:Estimates_partial_k(s)_1},
\begin{align}
\no \Psi_{\T{C,2} }^{s,n}(t,\tau_1) \no^2 & \lesssim \rho^{4\varepsilon}  \Vol{4}  \sum_{( k_1, l_1) \in \I_n}  \rho^{-\left( \frac{n-1}{cM}\right)}  \Vol{2}  \sum_{k_2=1}^N \vert \hat v_{k_2k_1} \vert \Vol{2}  \sum_{m_1=1}^N \vert \hat v_{m_1l_1}^{s,\varepsilon} \vert \nonumber \\
& \lesssim \rho^{6\varepsilon} \Big( \rho^{-\left( \frac{n-1}{cM}\right)} \V{n} \Big) \Big(\rho^{2\varepsilon} +C_\varepsilon\rho^{-1/\varepsilon } \Big) .
\end{align}
Application of \eqref{Estimate:An} completes the proof of the lemma.
\end{proof}

\subsubsection{Derivation of the Bound for $ \no \Psi_{\T{D}}(t) \no_{\T{TD}} $\label{sec:derivation_D}}
The term $\Psi_{\T{D}}(t)$ has a similar structure as $\Psi_{\T{C}}(t)$ and can be estimated in a similar way.\\

Let $M\ge1$ and for $0\le n \le M$,
\begin{align}
\Psi_{\T{D}}^{s,n}(t)& =  \Vol{4} \sum_{( k_1, l_1) \in \I_n} \sum_{l_2=N+1}^\infty \hat v_{l_1l_2}	\hat v_{k_1l_1}^{s,\varepsilon}  \int_0^t d\mu_2(\tau) D(\tau_2)  \Big( g_{l_1l_2}(\tau_2) g_{k_1l_1}(\tau_1) \varphi_0 \Big) \otimes \Omega_0^{[l_2^*k_1]} , \label{PSI:D:S} \\
\Psi_{\T{D}}^\ell(\tau) & =   \Vol{4} \sum_{k_1=1}^N  \sum_{l_1,l_2=N+1}^\infty  \hat v_{l_1l_2} \hat v_{k_1l_1}^{\ell,\varepsilon} \Big(  g_{l_1l_2}(\tau_2) g_{k_1l_1}(\tau_1) \varphi_0 \Big) \otimes \Omega_0^{[l_2^*k_1]} .
\end{align}
This leads to the identity
\begin{align}
\Psi_{\T{D}}(t) = \sum_{n=0}^M \Psi_{\T{D}}^{s,n}(t)  +  \int_0^t d\mu_2(\tau) D(\tau_2) \Psi_{\T{D}}^\ell (\tau).
\end{align}
Using $\lsp \Omega_0^{[l_2^*k_1]} ,\Omega_0^{[n_2^*m_1]}  \Omega_0 \rsp = \delta_{l_2n_2} \delta_{k_1m_1}$ which holds for $l_2,n_2\ge N+1$ and $k_1,m_1\le N$, one finds
\begin{align}
\no \Psi_{\T{D}}^\ell (t) \no^2 & \lesssim  
   \Vol{4} \sum_{k_1=1}^N \sum_{l_1=N+1}^\infty \vert \hat v_{k_1l_1}^{\ell,\varepsilon} \vert  \Vol{2}   \sum_{l_2=N+1}^\infty  \vert \hat v_{l_1l_2} \vert    \Vol{2}   \sum_{n_1=N+1}^\infty \vert \hat v_{k_1n_1}^{\ell,\varepsilon} \vert .
\end{align}
Then, \eqref{lem:estimates_v_L} in combination with \eqref{lemma:v_hat_estimates_p} leads to $\no \Psi_{\T{D}}^\ell(\tau) \no_{\T{TD}} \le  C_\varepsilon \rho^{\varepsilon-\frac{1}{\varepsilon}}$, and hence,
\begin{align}
\bno \int_0^t d\mu_2(\tau) D(\tau_2) \Psi_{\T{D}}^\ell (\tau) \bno_{ \T{TD} } \le C_\varepsilon t^2 \rho^{\varepsilon-\frac{1}{\varepsilon} }.
\end{align}
We similarly estimate the norm in $\Psi_{\T{D}}^{s,0}(t)$,
\begin{align}
& \bno \Vol{4} \sum_{(k_1, l_1)\in \I_0} \sum_{l_2=N+1}^\infty \hat v_{l_1l_2}	\hat v_{k_1l_1}^{s,\varepsilon} \Big( g_{l_1l_2}(\tau_2) g_{k_1l_1}(\tau_1) \varphi_0 \Big) \otimes \Omega_0^{[l_2^*k_1]} \bno^2\nonumber \\
& \hspace{5cm} \lesssim \Vol{4} \sum_{( k_1, l_1 ) \in \I_0}   \Vol{2}  \sum_{l_2=N+1}^\infty  \vert \hat v_{l_1l_2} \vert  \Vol{2}  \sum_{n_1=N+1}^\infty \vert \hat v_{k_1n_1}^{s,\varepsilon} \vert.
\end{align}
With \eqref{lemma:v_hat_estimates_p} and  \eqref{Estimate:A0}, it follows that
\begin{align}
\no \Psi_{\T{D}}^{s,0}(t) \no_{\T{TD}} & \lesssim t^2 \rho^{-\frac{1}{4} + \frac{3}{2}\varepsilon} \Big(\rho^{\varepsilon} + C_\varepsilon \rho^{-1/(2\varepsilon)} \Big).
\end{align}
\begin{lemma}\label{LEMMA:PSI:D:S}Let  $\Psi_{\T{D}}^{s,n}(t)$ as in \eqref{PSI:D:S}. Then, under the same assumptions as in Theorem~\ref{thm:main_thm}, there exists a positive constant $C$ such that\end{lemma}
\vspace{-0.7cm}
\begin{align}
\no \Psi_{\T{D}}^{s,n}(t) \no_{\T{TD}} & \lesssim (1+t)^2 \rho^{-\frac{1}{4}+ \frac{7}{2}\varepsilon + \frac{1}{2cM}} \Big( \rho^{\varepsilon} +  C_\varepsilon \rho^{-1/(2\varepsilon)} \Big) , \hspace{0.5cm} 1\le n \le M,
\end{align}
\textit{holds for all $t\ge 0$.}\\

It follows exactly as below Lemma \ref{LEMMA:PSI:C:S} that for $M=\floor{\ln \rho}$ (the largest integer smaller than $\ln \rho$),
\begin{align}
\sum_{n=1}^{M} \no  \Psi_{\T{D}}^{s,n}(t) \no_{\T{TD}} 
& \lesssim (1+t)^2 \rho^{-\frac{1}{4}+ 2\varepsilon } \Big( \rho^{\varepsilon}  +  C_\varepsilon \rho^{-1/(2\varepsilon) } \Big),
\end{align}
which proves the bound for $\no \Psi_{\T{D}}(t)\no_{\T{TD}}$ in \eqref{lemma:main-norm-of-B1}.
\begin{proof}[Proof of Lemma \ref{LEMMA:PSI:D:S}]
For $1\le n \le M$, we set
\begin{align}
\Psi_{\T{D,1}}^{s,n}(t,\tau_1) & =  \Vol{4}  \sum_{(k_1, l_1) \in \I_n} \sum_{l_2=N+1}^\infty \hat v_{l_1l_2}  \hat  v_{k_1l_1}^{s,\varepsilon} \Big( g_{l_1l_2}(\tau_1)  \frac{  g_{k_1l_1}(t) - g_{k_1l_1}(\tau_1)}{ i(E_{l_1}-E_{k_1}) } \varphi_0 \Big) \otimes \Omega_0^{[l_2^*k_1]}, \\
\Psi_{\T{D,2}}^{s,n}(\tau) & =   \Vol{4} \sum_{( k_1, l_1 ) \in \I_n} \sum_{l_2=N+1}^\infty \hat v_{l_1l_2}  \hat  v_{k_1l_1}^{s,\varepsilon}  \Big( g_{l_1l_2}(\tau_2)  \frac{  e^{i(E_{l_1}-E_{k_1})\tau_1}\partial_{\tau_1} k_{k_1l_1}(\tau_1) }{ i(E_{l_1}-E_{k_1}) } \varphi_0 \Big) \otimes \Omega_0^{[l_2^*k_1]}.
\end{align}
Partial integration leads to
\begin{align} \Psi_{\T{D}}^{s,n}(t) = \int_0^t d\tau_1 D(\tau_1) \Psi_{\T{D,1}}^{s,n}(t,\tau) - \int_0^t d\mu_2(\tau) D(\tau_2) \Psi_{\T{D,2}}^{s,n}(\tau).
\end{align}
It remains to compute the norm of the wave functions on the r.h.s.\ Using \eqref{energy_difference},
\begin{align}
\no \Psi_{\T{D,1} }^{s,n}(t,\tau_1) \no^2 
 & \lesssim    \rho^{ -\left( \frac{n-1}{cM} \right)}  \Vol{4}  \sum_{(k_1, l_1) \in \I_n}    \Vol{2}   \sum_{l_2=N+1}^\infty \vert \hat v_{l_1l_2} \vert  \Vol{2} \sum_{n_1=N+1}^\infty \vert \hat v_{k_1n_1}^{s,\varepsilon} \vert ,
\end{align}
and similarly, using in addition \eqref{lem:Estimates_partial_k(s)_1},
\begin{align}
\no \Psi_{\T{D,2}}^{s,n}(\tau) \no^2
 & \lesssim \rho^{4\varepsilon} \rho^{ -\left( \frac{n-1}{c M} \right)}  \Vol{4}  \sum_{( k_1, l_1) \in \I_n}   \Vol{2}  \sum_{l_2=N+1}^\infty \vert \hat v_{l_1l_2} \vert \Vol{2}    \sum_{n_1=N+1}^\infty \vert \hat v_{k_1n_1}^{s,\varepsilon} \vert .
\end{align}
The lemma then follows from \eqref{lemma:v_hat_estimates_p} together with \eqref{Estimate:An}.

\end{proof}

\subsubsection{Derivation of the Bound for $\no \Psi_{\T{E}}(t) \no_{\T{TD}} $} 
The term $\Psi_{\T{E}}(t)$ is more difficult to estimate since it involves four sums. In order to get the desired bound, it is not enough to do one partial integration.\ We have to split the term more carefully into different contributions and for some of them perform an additional partial integration. This gives an additional phase cancellation which is enough for the desired bound.\\

For $0\le n,m\le M$, we define
\begin{align}
 \Psi_{\T{E}}^{s,nm}(t)& =   \Vol{4} \sum_{(k_1,l_1)\in \I_n } \sum_{ ( k_2 , l_2 ) \in \I_m} \hat v^{s,\varepsilon}_{k_2l_2} \hat v^{s,\varepsilon}_{k_1l_1}  \int_0^t d\mu_3(\tau) D(\tau_3) \Big(  g_{k_2l_2}(\tau_2) g_{k_1l_1}(\tau_1)   \varphi_0 \Big) \otimes \Omega_0^{[l_2^*k_2l_1^*k_1]},\label{PSI:E:S}\\
	\Psi_{\T{E}}^\ell(\tau) & =  \Vol{4}  \sum_{k_1,k_2=1}^N \sum_{l_1,l_2=N+1}^{\infty} \hat w^{\ell ,\varepsilon}_{k_2l_2k_1l_1 } \Big(  g_{k_2l_2}(\tau_2) g_{k_1l_1}(\tau_1)   \varphi_0 \Big) \otimes \Omega_0^{[l_2^*k_2l_1^*k_1]},
\end{align}
where 
\begin{align}\hat w^{\ell ,\varepsilon}_{k_2l_2k_1l_1} :=  \hat v^{ \ell,\varepsilon}_{k_2l_2} \hat v^{s,\varepsilon}_{k_1l_1} +  \hat v^{s,\varepsilon}_{k_2l_2} \hat v_{k_1l_1}^{\ell ,\varepsilon} + \hat v^{\ell ,\varepsilon}_{k_2l_2} \hat v^{\ell ,\varepsilon}_{k_1l_1}.\label{HAT:W:ELL}
\end{align}
This leads to
\begin{align}\label{PSI_E}
\Psi_{\T{E}}(t) = E_{re}(\rho) \sum_{n,m=0}^{M}  \Psi_{\T{E}}^{s,nm}(t) + E_{re}(\rho) \int_0^{t}d\mu_3(\tau)  D(\tau_3)   \Psi_{\T{E}}^\ell(\tau),
\end{align}
since $\hat v_{k_1l_1} \hat v_{k_2l_2} = \hat v^{s,\varepsilon}_{k_2l_2} \hat v^{s,\varepsilon}_{k_1l_1} + \hat w^{\ell ,\varepsilon}_{k_2l_2k_1l_1}$.\ Using that for $k_1,k_2,m_1,m_2\le N$ and $N+1 \le l_1,l_2,n_1,n_2$,
\begin{align}
 \lsp  \Omega_0^{[l_2^*k_2l_1^*k_1]},  \Omega_0^{[n_2^*m_2n_1^*m_1]} \rsp & = (\delta_{k_2m_2} \delta_{k_1m_1} + \delta_{k_2m_1} \delta_{k_1m_2} ) \delta^\perp_{k_2k_1} \delta^\perp_{m_2m_1} \times \nonumber \\
&~~~ \times (\delta_{l_2n_2} \delta_{l_1n_1} + \delta_{l_2n_1} \delta_{l_1n_2} ) \delta_{l_2l_1}^\perp \delta_{n_2n_1}^\perp ,
\end{align}
where $\delta_{kl}^\perp = 1-\delta_{kl}$, we find
\begin{align}
 \no \Psi_{\T{E}}^\ell(\tau) \no^2 & \lesssim  \Vol{8} \sum_{k_1,k_2=1}^N \sum_{l_1,l_2=N+1}^{\infty} \Big( \vert  \hat v^{\ell ,\varepsilon}_{k_2l_2} \vert \ \vert  \hat v^{s,\varepsilon}_{k_1l_1} \vert + \vert \hat v^{s,\varepsilon}_{k_2l_2} \vert \ \vert \hat v_{k_1l_1}^{\ell,\varepsilon} \vert + \vert \hat v^{\ell ,\varepsilon}_{k_2l_2} \vert \ \vert \hat v^{\ell ,\varepsilon}_{k_1l_1} \vert \Big) .  
\end{align}
By means of \eqref{lemma:v_hat_estimates_fluc} and \eqref{lem:estimates_v_L}, $\no \Psi_{\T{E}}^\ell(\tau) \no_{\T{TD}} \le C_\varepsilon \rho^{\frac{1}{4} - \frac{1}{2\varepsilon}} $, such that together with \eqref{lemma:v_hat_estimates_Ka} we find for the last term in \eqref{PSI_E} that
\begin{align}
\bno E_{re}(\rho) \int_0^t d\mu_3( \tau) D(\tau_3)  \Psi_{\T{E}}^\ell(\tau)  \bno_{\T{TD}} \le C_\varepsilon t^3 \rho^{2\varepsilon+\frac{1}{4} - \frac{1}{2\varepsilon}}.
\end{align}
Similarly, one finds for $ \Psi_{\T{E}}^{s,00}(t)$,
\begin{align}
& \bno \Vol{4} \sum_{(k_1,l_1)\in \I_0 } \sum_{(k_2, l_2 ) \in \I_0} \hat v^{s,\varepsilon}_{k_2l_2} \hat v^{s,\varepsilon}_{k_1l_1}    \Big( g_{k_2l_2}(\tau_2) g_{k_1l_1}(\tau_1)   \varphi_0 \Big) \otimes \Omega_0^{[l_2^*k_2l_1^*k_1]} \bno^2 \lesssim \V{0}^2.
\end{align}
Hence, with \eqref{Estimate:A0},
\begin{align}
\no  \Psi_{\T{E}}^{s,00} (t) \no_{\T{TD} } \lesssim t^3 \rho^{-\frac{1}{2} + 4\varepsilon}.
\end{align}
\begin{lemma}\label{lem:Psi_D1(t)} Let $\Psi_{\T{E}}^{s,nm}(t)$ be defined as in \eqref{PSI:E:S}. Then, under the same assumptions as in Theorem \ref{thm:main_thm}, there exists a positive constant $C$ such that
\end{lemma}
\vspace{-0.6cm}
\begin{align}
\no \Psi_{\T{E}}^{s,0n}(t) \no_{\T{TD}} + \no \Psi_{\T{E}}^{s,n0}(t) \no_{\T{TD}} & \lesssim (1+t)^3 \rho^{-\frac{1}{2} + 3\varepsilon } \sqrt{ \rho^{\frac{1}{cM}} -\rho^{\frac{1}{2cM}} } , \hspace{0.5cm}  1\le n \le M,\\
\no \Psi_{\T{E}}^{s,nm} (t) \no_{\T{TD}} & \lesssim (1+t)^3 \rho^{-\frac{1}{2} + 5\varepsilon} \Big( \rho^{\frac{1}{cM}} -\rho^{\frac{1}{2cM}} \Big),  \hspace{0.5cm} 1\le n,m\le M,
\end{align}
\textit{holds for all $t>0$.}\\
 
With $M=\floor{\ln \rho}$ we then obtain that (recall \eqref{limit_M} and Remark \ref{INTERCHANGE:LIMITS}) 
\begin{align}
\sum_{n=1}^{M} \Big(\no \Psi_{\T{E}}^{s,n0}(t)\no_{\T{TD}} + \no \Psi_{\T{E}}^{s,0n}(t) \no_{\T{TD}} \Big) & \lesssim  (1+t)^3 \rho^{-\frac{1}{2} + 4\varepsilon }  ,   \\
\sum_{n=1}^M\no \Psi_{\T{E}}^{s,nm}(t) \no_{\T{TD}} &  \lesssim (1+t)^3 \rho^{-\frac{1}{2} + 6\varepsilon }.
\end{align}
This proves the bound for $\no \Psi_{\T{E}}(t) \no_{\T{TD}}$ in \eqref{lemma:main-norm-D1}.

\begin{proof}[Proof of Lemma~\ref{lem:Psi_D1(t)}]

We define  for $1\le n \le M$,\allowdisplaybreaks
\begin{align*}
\Psi_{\T{E,1}}^{s,n0}(t,\tau_1) & =  \Vol{4}   \sum_{( k_1, l_1) \in \I_n } \sum_{( k_2, l_2) \in \I_0}  \hat v^{s,\varepsilon}_{k_2l_2}  \hat v^{s,\varepsilon}_{k_1l_1} \Big( g_{k_2l_2}(\tau_1) \frac{ g_{k_1l_1}(t) - g_{k_1l_1}(\tau_1) }{i ( E_{l_1}-E_{k_1} ) }  \varphi_0 \Big) \otimes \Psi_0^{[l_2^*k_2l_1^*k_1]},\\
\Psi_{\T{E,2}}^{s,n0}(\tau) & = \Vol{4}  \sum_{ (k_1, l_1) \in \I_n } \sum_{ (k_2,l_2) \in \I_0}  \hat v^{s,\varepsilon}_{k_2l_2}   \hat v^{s,\varepsilon}_{k_1l_1}
\Big( g_{k_2l_2}(\tau_2) \frac{ e^{i( E_{l_1}-E_{k_1} ) \tau_1}   \partial_{\tau_1} k_{k_1l_1}(\tau_1) }{i ( E_{l_1}-E_{k_1} ) }    \varphi_0 \Big) \otimes \Psi_0^{[l_2^*k_2l_1^*k_1]},\\
\Psi_{\T{E,1}}^{s,0n}(\tau) &=   \Vol{4}  \sum_{ ( k_1,l_1) \in \I_0 } \sum_{ (k_2, l_2) \in \I_n}   \hat v^{s,\varepsilon}_{k_2l_2}  \hat v^{s,\varepsilon}_{k_1l_1} \Big( \frac{ g_{k_2l_2}(\tau_1) - g_{k_2l_2}(\tau_2) }{i ( E_{l_2}-E_{k_2} )}  g_{k_1l_1}(\tau_1) \varphi_0 \Big) \otimes \Psi_0^{[l_2^*k_2l_1^*k_1]},\\
\Psi_{\T{E,2}}^{s,0n}(\tau) & =  \Vol{4}   \sum_{( k_1,l_1) \in \I_0 } \sum_{ ( k_2 , l_2 ) \in \I_n}  \hat v^{s,\varepsilon}_{k_2l_2}  \hat v^{ s,\varepsilon}_{k_1l_1} \Big( \frac{ e^{i ( E_{l_2} - E_{k_2}) \tau_2 }  \partial_{\tau_2} k_{k_2l_2}(\tau_2) }{i ( E_{l_2}-E_{k_2}) }  g_{k_1l_1}(\tau_1) \varphi_0 \Big) \otimes \Psi_0^{[l_2^*k_2l_1^*k_1]}.
\end{align*}
By partial integration, this leads to
\begin{align}
\Psi_{\T{E}}^{s,n0}(t)&  =   E_{re}(\rho) \int_0^t d\mu_2(\tau) D(\tau_2)  \Psi_{\T{E,1}}^{s,n0}(t,\tau_1) - E_{re}(\rho) \int_0^t d\mu_3(\tau)  D(\tau_3) \Psi_{\T{E,2}}^{s,n0}(\tau),\\
\Psi_{\T{E}}^{s,0n}(t) & =   E_{re}(\rho) \int_0^t d\mu_2(\tau) D(\tau_2)  \Psi_{\T{E,1}}^{s,0n}(\tau) - E_{re}(\rho) \int_0^t  d\mu_3(\tau) D(\tau_3) \Psi_{\T{E,2}}^{s,0n}(\tau).
\end{align}
Furthermore,  for all $1\le n,m\le M$, and $X\in \{ 1,2,3\}$, we set
\begin{align}
\Psi_{\T{E,X} }^{s,nm}(t,\tau) & =  \Vol{4}    \sum_{ (k_1, l_1) \in \I_n } \sum_{ (k_2, l_2) \in \I_m} \hat v^{s,\varepsilon}_{k_2l_2} \hat v^{s,\varepsilon}_{k_1l_1}
\Big( G_{k_2l_2k_1l_2}^{\T{(X)}}(t,\tau)  \varphi_0 \Big) \otimes \Omega_0^{[l_2^*k_2l_1^*k_1]} ,
\end{align}
where we introduce the operators $G^{\T{(X)}}_{k_2l_2k_1l_1}(t,\tau): \mathcal H_y\to \mathcal H_y$, defined by
\begin{align}
G_{k_2l_2k_1l_1}^{\T{(1)}}(t,\tau) & =   \frac{g_{k_2l_2}(t) g_{k_1l_1}(t) - g_{k_2l_2}(\tau_1) g_{k_1l_1}(\tau_1)}{i (E_{l_2}-E_{k_2}) i(E_{l_2}-E_{k_2}+ E_{l_1}-E_{k_1})}  - \frac{ g_{k_2l_2}(\tau_1) \Big(  g_{k_1l_1}(t)- g_{k_1l_1}(\tau_1) \Big) }{ i(E_{l_2}-E_{k_2}) i (E_{l_1}-E_{k_1}) }, \label{DEF:G:1}\\
G_{k_2l_2k_1l_1}^{\T{(2)}}(t,\tau) & =  \frac{e^{i (E_{l_1} - E_{k_1} + E_{l_2} - E_{k_2} )\tau_1} \partial_{\tau_1} \Big( k_{k_2l_2}(\tau_1) k_{k_1l_1}(\tau_1) \Big) } {i (E_{l_2}-E_{k_2}) i(E_{l_2}-E_{k_2}+ E_{l_1}-E_{k_1})} 
- \frac{ g_{k_2l_2}(\tau_2) e^{i ( E_{l_1}-E_{k_1})\tau_1 } \Big( \partial_{\tau_1} k_{k_1l_1} (\tau_1) \Big) }{	i	(E_{l_2}-E_{k_2} ) i ( E_{l_1}-E_{k_1})} \nonumber\\
& - \frac{ e^{i ( E_{l_2} - E_{k_2})\tau_1} \Big(\partial_{\tau_1}k_{k_2l_2}(\tau_1) \Big) \Big( g_{k_1l_1}(t) - g_{k_1l_1}(\tau_1) \Big) }{ 	i	( E_{l_2}-E_{k_2} )	i ( E_{l_1}-E_{k_1}) }, \label{DEF:G:2}\\
G_{k_2l_2k_1l_1}^{\T{(3)}}(t,\tau) &  =   \frac{e^{i (E_{l_2}-E_{k_2} ) \tau_2}\Big( \partial_{\tau_2}k_{k_2l_2}(\tau_2)\Big) e^{i (E_{k_1}-E_{l_1})\tau_1} \Big( \partial_{\tau_1} k_{k_1 l_1} (\tau_1) \Big)}{	i	( E_{k_2}-E_{l_2} ) i ( E_{l_1}-E_{k_1} ) }. \label{DEF:G:3}
\end{align}
With a two-fold partial integration one now finds 
\begin{align}
\Psi_{\T{E}}^{s,nm}(t) = E_{re}(\rho) \Big[& \int_0^t d\tau_1 D(\tau_1) \Psi_{\T{E,1}}^{s,nm}(t,\tau) \nonumber \\
& ~~~ -  \int_0^t d\mu_2(\tau)D(\tau_2) \Psi_{\T{E,2}}^{s,nm}(t,\tau) + \int_0^t d\mu_3(\tau)D(\tau_3) \Psi_{\T{E,3}}^{s,nm}(t,\tau)\Big].
\end{align}
It remains to compute the norm of the above wave functions.
Recalling that the scalar product produces four Kronecker-deltas,
\begin{align}
 \no \Psi_{\T{E,1}}^{s,n0}(t,\tau_1) \no^2 & \lesssim  \Vol{4}    \sum_{(k_1,l_1) \in \I_n  }  \frac{1 }{E_{l_1}-E_{k_1}} 
 \Vol{4}   \sum_{ ( k_2 , l_2 ) \in \I_0}   \times \nonumber \\
& \hspace{2.5cm}\times   \sum_{ ( m_1 , n_1 ) \in \I_n  }   \frac{1 }{ E_{n_1}-E_{m_1}} 
 \sum_{ (m_2 ,n_2 ) \in \I_0}  \lsp  \Omega_0^{[l_2^*k_2l_1^*k_1]},  \Omega_0^{[n_2^*m_2n_1^*m_1]} \rsp \nonumber \\
& \lesssim  \Big( \rho^{-\left( \frac{n-1}{cM} \right)} \V{n} \Big)  \V{0}.
\end{align}
Using in addition \eqref{lem:Estimates_partial_k(s)_1},
\begin{align}
\no \Psi_{\T{E,2}}^{s,n0}(\tau) \no^2 & \lesssim   \Vol{4}    \sum_{ ( k_1 , l_1 ) \in \I_n  }  \frac{\no \partial_{\tau_1} k_{k_1l_1}(\tau_1) \varphi_0\no }{E_{l_1}-E_{k_1}}  \Vol{4}   \sum_{ ( k_2 , l_2) \in \I_0} \times \nonumber \\
& \times   \sum_{ ( m_1, n_1) \in \I_n  }   \frac{\no \partial_{\tau_1} k_{m_1n_1}(\tau_1) \varphi_0\no}{ E_{n_1}-E_{m_1}}  \sum_{( m_2, n_2) \in \I_0}   \lsp  \Omega_0^{[l_2^*k_2l_1^*k_1]},  \Omega_0^{[n_2^*m_2n_1^*m_1]} \rsp \nonumber \\
& \lesssim \rho^{4\varepsilon}     \Big( \rho^{-\left( \frac{n-1}{cM} \right)} \V{n} \Big)  \V{0}.
\end{align}
By means of \eqref{Estimate:An}, this shows the first part of the first bound in the lemma. We omit the proof for $\Psi_{\T{E}}^{s,0n}(\tau )$ since it works exactly the same way as for $\Psi_{\T{E}}^{s,n0}(\tau )$.\\
\\
For the second bound of the lemma, note that by using \eqref{energy_difference} and Lemma \ref{lem:Estimates_partial_k(s)}, one finds
\begin{align} 
\chi_{\I_n} \big( (k_1,l_1) \big) \chi_{\I_m}\big(  (k_2,l_2 ) \big)   \no G^{\T{(1)}}_{k_2l_2k_1l_1}(t,\tau) \varphi_0 \no & \lesssim \rho^{-\left( \frac{n-1}{2cM} \right) } \rho^{-\left( \frac{m-1}{2cM} \right) } ,	\label{BOUND:G:1} \\
\chi_{\I_n} \big( (k_1,l_1) \big) \chi_{\I_m}\big(  (k_2,l_2 ) \big) \no G^{(\T{2})}_{k_2l_2k_1l_1}(t,\tau ) \varphi_0 \no  & \lesssim \rho^{4\varepsilon} \rho^{-\left( \frac{n-1}{2cM} \right) } \rho^{-\left( \frac{m-1}{2cM} \right) } , \label{BOUND:G:2} \\
\chi_{\I_n} \big( (k_1,l_1) \big) \chi_{\I_m}\big(  (k_2,l_2 ) \big) \no G^{(\T{3})}_{k_2l_2k_1l_1}(t,\tau ) \varphi_0 \no  &\lesssim \rho^{4\varepsilon} \rho^{-\left( \frac{n-1}{2cM} \right) } \rho^{-\left( \frac{m-1}{2cM} \right) } \label{BOUND:G:3} .
\end{align} 
Then, for all $1\le n \le M$ and $X\in \{1,2,3\}$,
\begin{align}
\no \Psi_{\T{E,X}}^{s,nm}(t,\tau) \no^2  & \lesssim  \Vol{8}   \sum_{( k_1 , l_1 )\in \I_n }   \sum_{ ( k_2 ,l_2) \in \I_m } \sum_{ ( m_1 , n_1) \in \I_n }   \sum_{ ( m_2 , n_2) \in \I_m }  \no G^{\T{(X)}}_{k_2l_2k_1l_1}(t,\tau) \varphi_0 \no \times \nonumber \\
& \hspace{4cm} \times \no G^{\T{(X)}}_{m_2n_2m_1n_1}(t,\tau) \varphi_0 \no~ \lsp  \Omega_0^{[l_2^*k_2l_1^*k_1]},  \Omega_0^{[n_2^*m_2n_1^*m_1]} \rsp \nonumber \\
& \lesssim \rho^{8\varepsilon}  \Big( \rho^{-\left( \frac{n-1}{cM} \right)}\V{n} \Big)  \Big( \rho^{-\left( \frac{m-1}{cM} \right)} \V{m} \Big).
\end{align}
The bounds for the remaining expressions have been derived in \eqref{Estimate:An}.
\end{proof}

\subsubsection{Derivation of the Bound for $ \no \Psi_{ \T{F} }(t) \no_{ \T{TD} }$ \label{sec:bound_F}}
The estimate for $\no \Psi_{ \T{F} }(t) \no_{ \T{TD} }$ is more tedious, since we have to deal with an additional $(V-E)$, i.e., we have to take into account one more collision, starting from $\Psi_4$. This leads to many possible collision histories, which we write down in \eqref{PSI:FS1}-\eqref{PSI:F:L:3} below. After that, we use the same techniques as in the previous sections, taking care in addition of the more tedious combinatorics.\\

We first rewrite the potential
\begin{align}
	(V-E) = \Vol{2} \underset{(l_3\neq k_3)}{\sum_{l_3=1}^{\infty} \sum_{k_3=1}^{\infty}} \hat v_{k_3l_3} e^{i ( p_{k_3} - p_{l_3} ) y } a^*(p_{l_3}) a(p_{k_3})
\end{align}
in terms of fermionic creation and annihilation operators, cf.\ \eqref{creation_annihilation_op}. This can be used to decompose the wave function $\Psi_{ \T{F} }(t)$ (for $M\ge 1$) in terms of
\begin{align}
\Psi_{\T{F} }(t) = & \sum_{n,m=0}^M \Big( \Psi_{\T{F,1} }^{s,nm}(t) + \Psi_{\T{F,2} }^{s,nm}(t) +  \Psi_{ \T{F,3} }^{s,nm}(t)\Big)\nonumber \\
&\hspace{2cm}  + \int_0^t d\mu_3(\tau) D(\tau_3) \Big( \Psi_{\T{F,1}}^\ell (\tau) + \Psi_{\T{F,2}}^\ell (\tau)  + \Psi_{\T{F,3}}^\ell (\tau) \Big),
\end{align} 
where (recall the definition \eqref{HAT:W:ELL} for $\hat{w}^{\ell,\varepsilon}$)\allowdisplaybreaks
\begin{align}
\Psi_{\T{F,1}}^{s,nm}(t) & =  \Vol{6} \sum_{ (k_1, l_1) \in \I_n} \sum_{( k_2,l_2)\in \I_m} \sum_{k_3=1}^N \sum_{l_3=N+1}^\infty  \hat v_{k_3l_3} \hat v^{s,\varepsilon}_{k_2l_2}   \hat v^{s,\varepsilon}_{k_1l_1} \times \label{PSI:FS1} \\
& \hspace{4.0cm} \times \int_0^{t}d\mu_3(\tau) D(\tau_3) \Big( g_{k_3l_3}(\tau_3) g_{k_2l_2}(\tau_2) g_{k_3l_3}(\tau_1)\varphi_0 \Big) \otimes \Omega_0^{[l_3^* k_3 l_2^* k_2 l_1^* k_1 ]} ,\nonumber \\
\Psi_{\T{F,2}}^{s,nm}(t) &=  \Vol{6} \sum_{ ( k_1, l_1) \in \I_n} \sum_{ ( k_2 , l_2 ) \in \I_m} \sum_{l_3=N+1}^{\infty}  \hat v_{l_2l_3} \hat v^{s,\varepsilon}_{k_2l_2}   \hat v^{s,\varepsilon}_{k_1l_1}  \times \nonumber \\
& \hspace{4cm}  \times \int_0^{t}d\mu_3(\tau)  D(\tau_3)  \Big( g_{l_2l_3}(\tau_3) g_{k_2l_2}(\tau_2) g_{k_1l_1}(\tau_1)  \varphi_0 \Big) \otimes \Omega_0^{[l_3^* k_2 l_1^* k_1]}  \nonumber\\
& +  \Vol{6} \sum_{ ( k_1, l_1) \in \I_n} \sum_{( k_2 , l_2 ) \in \I_m} \sum_{l_3=N+1}^{\infty}  \hat v_{l_1l_3} \hat v^{s,\varepsilon}_{k_2l_2}   \hat v^{s,\varepsilon}_{k_1l_1} \times \nonumber \\
& \hspace{4cm}  \times  \int_0^{t}d\mu_3(\tau)  D(\tau_3) \Big( g_{l_1 l_3 }(\tau_3) g_{k_2l_2}(\tau_2) g_{k_1l_1}(\tau_1)\varphi_0 \Big) \otimes \Omega_0^{[l_3^*l_2^*k_2k_1]} \nonumber\\
& + \Vol{6} \sum_{ ( k_1, l_1) \in \I_n} \sum_{ ( k_2 , l_2 ) \in \I_m} \sum_{k_3=1}^{N}  \hat v_{k_3 k_2} \hat v^{s,\varepsilon}_{k_2l_2}   \hat v^{s,\varepsilon}_{k_1l_1}  \times  \nonumber \\
& \hspace{4cm} \times  \int_0^{t}d\mu_3(\tau)  D(\tau_3) \Big( g_{k_3k_2}(\tau_3) g_{k_2l_2}(\tau_2) g_{k_1 l_1}(\tau_1) \varphi_0 \Big) \otimes \Omega_0^{[ k_3 l_2^* l_1^* k_1]} \nonumber\\
& + \Vol{6} \sum_{ ( k_1, l_1) \in \I_n} \sum_{ ( k_2 , l_2 ) \in \I_m} \sum_{k_3=1}^{N}  \hat v_{k_3 k_1} \hat v^{s,\varepsilon}_{k_2l_2}   \hat v^{s,\varepsilon}_{k_1l_1} \times  \label{PSI:FS2} \\
& \hspace{4cm} \times \int_0^{t}ds\mu_3(\tau)  D(\tau_3) \Big( g_{k_3 k_1 }(\tau_3) g_{k_2l_2}(\tau_2) g_{k_1l_1}(\tau_1)\varphi_0 \Big) \otimes \Omega_0^{[k_3 l_2^* k_2 l_1^*]} ,\nonumber\\
\Psi_{\T{F,3}}^{s,nm}(t) & =   \Vol{6}  \sum_{ ( k_1, l_1) \in \I_n} \sum_{ ( k_2 , l_2 ) \in \I_m} \hat v_{l_2k_2}  \hat v^{s,\varepsilon}_{k_2l_2}  \hat v^{s,\varepsilon}_{k_1l_1} \int_0^{t} d\mu_3(\tau)   D(\tau_3) \Big( g_{l_2k_2}(\tau_3) g_{k_2l_2}(\tau_2) g_{k_1l_1}(\tau_1) \varphi_0 \Big) \otimes \Omega_0^{[ l_1^* k_1]} \nonumber\\
& +  \Vol{6} \sum_{ ( k_1 , l_1 ) \in \I_n} \sum_{ ( k_2 , l_2 ) \in \I_m}  \hat v_{l_1k_1}   \hat v^{s,\varepsilon}_{k_2l_2}   \hat v^{s,\varepsilon}_{k_1l_1}   \int_0^{t}d\mu_3(\tau)  D(\tau_3) \Big( g_{l_1k_1}(\tau_3) g_{k_2l_2}(\tau_2) g_{k_1l_1}(\tau_1)  \varphi_0 \Big) \otimes \Omega_0^{[ l_2^* k_2 ]} \nonumber\\
& +  \Vol{6} \sum_{ ( k_1, l_1) \in \I_n} \sum_{ ( k_2 , l_2 ) \in \I_m}   \hat v_{l_2 k_1} \hat v^{s,\varepsilon}_{k_2l_2}   \hat v^{s,\varepsilon}_{k_1l_1} \int_0^{t}d\mu_3(\tau)  D(\tau_3) \Big( g_{l_2k_1}(\tau_3) g_{k_2l_2}(\tau_2) g_{k_1l_1}(\tau_1) \varphi_0 \Big) \otimes \Omega_0^{[ k_2 l_1^* ]}   \nonumber\\
& +  \Vol{6} \sum_{ ( k_1, l_1) \in \I_n} \sum_{ ( k_2 , l_2 ) \in \I_m}  \hat v_{l_1 k_2} \hat v^{s,\varepsilon}_{k_2l_2}   \hat v^{s,\varepsilon}_{k_1l_1} \int_0^{t}d\mu_3(s)  D(\tau_3) \Big( g_{l_1 k_2}(\tau_3) g_{k_2l_2}(\tau_2) g_{k_1l_1}(\tau_1)   \varphi_0 \Big) \otimes \Omega_0^{[l_2^* k_1]}, \label{PSI:F:S3}
\end{align}
and \allowdisplaybreaks
\begin{align}
\Psi_{\T{F,1}}^\ell(\tau) & =  \Vol{6}  \sum_{k_1,k_2,k_3=1}^N  \sum_{l_1,l_2,l_3=N+1}    \hat v_{k_3l_3} \hat w^{\ell,\varepsilon}_{k_2l_2k_1l_1} \Big(  g_{k_3l_3}(\tau_3) g_{k_2l_2}(\tau_2) g_{k_1l_1}(\tau_1) \varphi_0 \Big) \otimes \Omega_0^{[l_3^* k_3 l_2^* k_2 l_1^*k_1]}, \label{F:ELL:1}\\
\Psi_{\T{F,2}}^\ell(\tau) &=  \Vol{6} \sum_{k_1,k_2=1}^N  \sum_{l_1,l_2=N+1}^\infty \sum_{l_3=N+1}^{\infty}  \hat v_{l_2 l_3} \hat w^{\ell,\varepsilon}_{k_2l_2k_1l_1} \Big(  g_{l_2 l_3}(\tau_3) g_{k_2l_2}(\tau_2) g_{k_1l_1}(\tau_1)  \varphi_0 \Big) \otimes \Omega_0^{[l_3^* k_2 l_1^* k_1 ]}   \nonumber\\
& +  \Vol{6} \sum_{k_1,k_2=1}^N  \sum_{l_1,l_2=N+1}^\infty \sum_{l_3=N+1}^{\infty}   \hat v_{l_1 l_3 } \hat w^{\ell,\varepsilon}_{k_2l_2k_1l_1}  \Big( g_{l_1 l_3 }(\tau_3) g_{k_2l_2}(\tau_2) g_{k_1l_1}(\tau_1) \varphi_0 \Big) \otimes \Omega_0^{[l_3^* l_2^* k_2 k_1 ]}   \nonumber\\
& +  \Vol{6}  \sum_{k_1,k_2=1}^N  \sum_{l_1,l_2=N+1}^\infty  \sum_{k_3=1}^{N}  \hat v_{k_3 k_2}  \hat w^{\ell,\varepsilon}_{k_2l_2k_1l_1} \Big(  g_{k_3 k_2}(\tau_3) g_{k_2l_2}(\tau_2) g_{k_1l_1}(\tau_1)  \varphi_0 \Big) \otimes \Omega_0^{[ k_3  l_2^*  l_1^* k_1]}  \nonumber\\
& +  \Vol{6} \sum_{k_1,k_2=1}^N  \sum_{l_1,l_2=N+1}^\infty  \sum_{k_3=1}^{N}  \hat v_{k_3 k_1} \hat w^{\ell,\varepsilon}_{k_2l_2k_1l_1}  \Big(  g_{k_3 k_1 }(\tau_3) g_{k_2l_2}(\tau_2) g_{k_1l_1}(\tau_1) \varphi_0\Big) \otimes  \Omega_0^{[ k_3  l_2^*k_2l_1^*] }  , \label{F:2:ELL}\\
\Psi_{\T{F,3}}^\ell	(\tau) & =  \Vol{6}  \sum_{k_1,k_2=1}^N  \sum_{l_1,l_2=N+1}^\infty  \hat v_{l_2k_2}  \hat w^{\ell,\varepsilon}_{k_2l_2k_1l_1}  \Big(  g_{l_2k_2}(\tau_3) g_{k_2l_2}(\tau_2) g_{k_1l_1}(\tau_1) \varphi_0 \Big) \otimes \Omega_0^{ [l_1^* k_1]}  \nonumber\\
& +  \Vol{6} \sum_{k_1,k_2=1}^N  \sum_{l_1,l_2=N+1}^\infty  \hat v_{l_1k_1}  \hat w^{\ell,\varepsilon}_{k_2l_2k_1l_1}  \Big( g_{l_1k_1}(\tau_3) g_{k_2l_2}(\tau_2) g_{k_1l_1}(\tau_1) \varphi_0 \Big) \otimes  \Omega^{[l_2^* k_2]}    \nonumber\\
& + \Vol{6} \sum_{k_1,k_2=1}^N  \sum_{l_1,l_2=N+1}^\infty  \hat v_{l_2 k_1}\hat w^{ \ell,\varepsilon}_{k_2l_2k_1l_1}   \Big( g_{l_2k_1}(\tau_3) g_{k_2l_2}(\tau_2) g_{k_1l_1}(\tau_1) \varphi_0 \Big) \otimes \Omega_0^{[ k_2 l_1^* ]} \nonumber\\
& +  \Vol{6} \sum_{k_1,k_2=1}^N  \sum_{l_1,l_2=N+1}^\infty  \hat v_{l_1 k_2 } \hat w^{\ell,\varepsilon}_{k_2l_2k_1l_1}  \Big(  g_{l_1 k_2 }(\tau_3) g_{k_2l_2}(\tau_2) g_{k_1l_1}(\tau_1)   \varphi_0 \Big) \otimes \Omega_0^{[ l_2^* k_1 ]} . \label{PSI:F:L:3}
\end{align}
The different contributions in $\Psi_{\T{F}}(t)$ correspond to the different collision histories in $(V-E) \Psi_{\T{4}}$.\\
\\
\textbf{Bounds for} $\no \Psi_{\T{F,1}}^{s,nm}(t) \no_{\T{TD}}$ \textbf{and} $\no \Psi_{\T{F,1}}^\ell(\tau) \no_{\T{TD}}$\textbf{.}\ We use that for $k_i,m_i \le N$, $N+1\le l_i,n_i$ ($i=1,2,3$), the scalar product
\begin{align*}
\langle  \Omega_0^{[l_3^* k_3 l_2^* k_2 l_1^*k_1]}   , \Omega_0^{[n_3^* m_3 n_2^* m_2 n_1^* m_1]}  \rangle_{\mathcal H_N} & = \Big( \sum_{\sigma \in S_3} \delta_{l_3 n_{\sigma(3)} } \delta_{l_2 n_{\sigma(2)} } \delta_{l_1 n_{\sigma(1)}}  \Big) 
\delta^\perp_{l_3l_2} \delta^\perp_{l_2l_1}   \delta^\perp_{l_1l_3} \delta^\perp_{n_3n_2} \delta^\perp_{n_2n_1}   \delta^\perp_{n_1n_3}\times
\\
 \times & \Big( \sum_{\sigma \in S_3} \delta_{k_3 m_{\sigma(3)} } \delta_{k_2 m_{\sigma(2)} } \delta_{k_1 m_{\sigma(1)} } \Big) \delta^\perp_{k_3k_2} \delta^\perp_{k_2k_1}   \delta^\perp_{k_1k_3} \delta^\perp_{m_3m_2} \delta^\perp_{m_2m_1}   \delta^\perp_{m_1m_3},
\end{align*}
produces six Kronecker deltas in each summand, in order to find
\begin{align}
\no\Psi_{\T{F,1}}^\ell(\tau)\no^2 & \lesssim \Vol{8} \sum_{k_1,k_2=1}^N \sum_{l_1,l_2=N+1}^{\infty}   \vert  \hat w^{\ell , \varepsilon}_{k_2l_2k_1l_1}  \vert \Vol{4}  \sum_{ k_3=1}^N  \sum_{ l_3=N+1}^\infty \vert \hat v_{k_3l_3}  \vert .
\end{align}
Then, by means of \eqref{lemma:v_hat_estimates_fluc} and \eqref{lem:estimates_v_L}, $\no\Psi_{\T{F,1}}^\ell (\tau)\no_{\T{TD}} \le C_\varepsilon \rho^{\frac{1}{4}-\frac{1}{2\varepsilon}+\varepsilon}$, which leads to
\begin{align}
\bno \int_0^t d\mu_3(\tau) D(\tau_3) \Psi_{\T{F,1}}^\ell(\tau) \bno_{\T{TD}} \le C_\varepsilon t^3   \rho^{\frac{1}{2}-\frac{1}{2\varepsilon}}.
\end{align}
Similarly in $\Psi_{\T{F,1}}^{s,00}(t)$, we estimate the norm,
\begin{align}
& \bno \Vol{6} \sum_{( k_1, l_1) \in \I_0} \sum_{( k_2, l_2 ) \in \I_0} \sum_{k_3=1}^N \sum_{l_3=N+1}^\infty  \hat v_{k_3l_3} \hat v^{s,\varepsilon}_{k_2l_2}   \hat v^{s,\varepsilon}_{k_1l_1}  \Big( g_{k_3l_3}(\tau_3) g_{k_2l_2}(\tau_2) g_{k_1l_1}(\tau_1)\varphi_0 \Big) \otimes \Omega_0^{[l_3^* k_3 l_2^* k_2 l_1^* k_1 ]} \bno^2 \nonumber\\
& \hspace{9cm} \lesssim \V{0}^2  \Vol{4} \sum_{k_3=1}^N \sum_{l_3=N+1}^\infty \vert \hat v_{k_3l_3} \vert .
\end{align}
Using \eqref{lemma:v_hat_estimates_fluc} in combination with \eqref{Estimate:A0}, leads to
\begin{align}
\no \Psi_{\T{F,1}}^{s,00}(t)\no_{\T{TD}} & \lesssim t^3 \rho^{-\frac{1}{4}+\varepsilon}.
\end{align}
\begin{lemma}
\label{lem:Psi_FS1(t)} Let $\Psi_{\T{F,1} }^{s,nm}(t)$ be defined as in \eqref{PSI:FS1}. Then, under the same assumptions as in Theorem \ref{thm:main_thm}, there exists a positive constant $C$ such that
\end{lemma}
\vspace{-0.7cm}
\begin{align}
 \no \Psi_{\T{F,1}}^{s,0n}(t) \no_{\T{TD}} +  \no \Psi_{\T{F,1}}^{s,n0}(t) \no_{\T{TD}} & \lesssim (1+t)^3 \rho^{-\frac{1}{4}+3\varepsilon} \sqrt{ \rho^{\frac{1}{cM}} - \rho^{\frac{1}{2cM}} } , \hspace{0.5cm} 1\le n \le M, \\
  \no \Psi_{\T{F,1}}^{s,nm}(t) \no_{\T{TD}}   & \lesssim  (1+t)^3 \rho^{-\frac{1}{4}+5\varepsilon} \Big( \rho^{\frac{1}{cM}} - \rho^{\frac{1}{2cM}} \Big) ,\hspace{0.5cm} 1\le n,m \le M.
\end{align}  
\textit{holds for all $t>0$.}\\

Next, we set again $M=\floor{\ln \rho}$ and find, using $M \rho^{\frac{1}{2cM}} \lesssim \rho^{\varepsilon}$ as well as $M^2 \rho^{\frac{1}{cM}} \lesssim \rho^{\varepsilon}$,
\begin{align}
\sum_{n=1}^{M}  \Big( \no \Psi_{\T{F,1}}^{s,0n}(t) \no_{\T{TD}} + \no \Psi_{\T{F,1}}^{s,n0}(t) \no_{\T{TD}} \Big)&  \lesssim (1+t)^3 \rho^{-\frac{1}{4}+4\varepsilon}, \\
\sum_{n,m=1}^{M} \no \Psi_{\T{F,1}}^{s,nm}(t) \no_{\T{TD}} & \lesssim  (1+t)^3 \rho^{-\frac{1}{4}+6\varepsilon}.
\end{align}
\begin{proof}[Proof of Lemma \ref{lem:Psi_FS1(t)}]
We define for $1\le n \le M$,
\begin{align}
\Psi_{\T{F,11}}^{s,n0}(t,\tau) & =  \Vol{6} \sum_{(k_1, l_1) \in \I_n }  \sum_{(k_2, l_2)\in \I_0 } \sum_{k_3=1}^N \sum_{l_3=N+1}^\infty  \hat v_{k_3l_3} \hat v^{s,\varepsilon}_{k_2l_2}   v^{s,\varepsilon}_{k_1l_1}   \times \\
& \hspace{4cm} \times \Big( g_{k_3l_3}(\tau_2) g_{k_2l_2}(\tau_1)\frac{ g_{k_1l_1}(t) - g_{k_1l_1}(\tau_1) }{i	 ( E_{l_1}-E_{k_1} ) }  \varphi_0 \Big) \otimes \Omega_0^{[l_3^* k_3 l_2^*k_2l_1^*k_1]} , \nonumber \\
\Psi_{\T{F,12}}^{s,n0}(\tau) & = \Vol{6} \sum_{ (k_1, l_1)  \in \I_n }  \sum_{(k_2, l_2) \in \I_0 } \sum_{k_3=1}^N \sum_{l_3=N+1}^\infty  \hat v_{k_3l_3}\hat v^{s,\varepsilon}_{k_2l_2}   \hat v^{s,\varepsilon}_{k_1l_1}  \times \\
& \hspace{4cm} \times \Big(  g_{k_3l_3}(\tau_3)  g_{k_2l_2}(\tau_2) \frac{ e^{i ( E_{l_1}-E_{k_1} ) \tau_1}   \partial_{\tau_1} k_{k_1l_1}(\tau_1) }{i( E_{l_1}-E_{k_1} )}       \varphi_0 \Big) \otimes \Omega_0^{[l_3^* k_3 l_2^*k_2l_1^*k_1]}, \nonumber\\
\Psi_{\T{F,11}}^{s,0n}(\tau) & =  \Vol{6} \sum_{ (k_1, l_1)  \in \I_0 }  \sum_{(k_2, l_2) \in \I_n } \sum_{k_3=1}^N \sum_{l_3=N+1}^\infty  \hat v_{k_3l_3} \hat v^{s,\varepsilon}_{k_2l_2} v^{s,\varepsilon}_{k_1l_1}  \times \\
& \hspace{4cm} \times \Big( g_{k_3l_3}(\tau_2)  \frac{ g_{k_2l_2}(\tau_1) - g_{k_2l_2}(\tau_2) }{i ( E_{l_2}-E_{k_2} ) }  g_{k_1l_1}(\tau_1) \varphi_0 \Big) \otimes \Omega_0^{[l_3^* k_3 l_2^*k_2l_1^*k_1]}, \nonumber\\
\Psi_{\T{F,12}}^{s,0n}(\tau) & =  \Vol{6} \sum_{(k_1, l_1) \in \I_0 }  \sum_{( k_2, l_2) \in \I_n } \sum_{k_3=1}^N \sum_{l_3=N+1}^\infty  \hat v_{k_3l_3}  \hat v^{s,\varepsilon}_{k_2l_2}  v^{s,\varepsilon}_{k_1l_1}   \times \\
& \hspace{4cm} \times \Big( g_{k_3l_3}(\tau_3)\frac{ e^{i ( E_{l_2}-E_{k_2} ) \tau_2 }  \partial_{\tau_2} k_{k_2l_2}(\tau_2) }{i ( E_{l_2}-E_{k_2} ) }  g_{k_1l_1}(\tau_1)  \varphi_0 \Big) \otimes \Omega_0^{[l_3^* k_3 l_2^*k_2l_1^*k_1]}.\nonumber
\end{align}
Via partial integration we obtain
\begin{align}
\Psi_{ \T{F,1} }^{s,n0} (t) = & \int_0^t d\mu_2(\tau) D(\tau_2)  \Psi_{ \T{F,11} }^{s,n0} (t,\tau)  -  \int_0^t d\mu_3(\tau)  D(\tau_3) \Psi_{ \T{F,12} }^{ s,n0 }(\tau),\\
\Psi_{ \T{F,1} }^{s,0n} (t) = & 	\int_0^t d\mu_2(\tau) D(\tau_2)  \Psi_{ \T{F,11} }^{s,0n} (\tau) - \int_0^t  d\mu_3(\tau) D(\tau_3) \Psi_{\T{F,12}}^{s,0n} (\tau).
\end{align}
We further set for $1\le n,m\le M$, and for $G^{\T{(X)}}_{k_2l_2k_1l_1}(t,\tau)$, $X\in \{1,2,3\}$, defined as in \eqref{DEF:G:1}-\eqref{DEF:G:3},
\begin{align*}
\Psi_{\T{F,1X} }^{s,nm}(t,\tau) & =  \Vol{6} \sum_{( k_1 ,l_1) \in \I_n }  \sum_{ (k_2, l_2 ) \in \I_m } \sum_{k_3=1}^N \sum_{l_3=N+1}^\infty   \hat v_{k_3l_3} \hat v^{s,\varepsilon}_{k_2l_2}  \hat v^{s,\varepsilon}_{k_1l_1} \Big( g_{k_3l_3}(\tau_{\T{X}}) G_{k_2l_2k_1l_1}^{\T{(X)}}(t,\tau) \varphi_0 \Big) \otimes \Omega_0^{[l_3^* k_3 l_2^*k_2l_1^*k_1]}.
\end{align*}
Partial integration leads again to
\begin{align*}
\Psi_{\T{F,1}}^{s,nm}(t)  = \int_0^t d\tau_1 D(\tau_1) \Psi_{\T{F,11}}^{s,nm }(t,\tau)  - \int_0^t d\mu_2(\tau) D(\tau_2) \Psi_{\T{F,12} }^{s,nm }(t,\tau)  + \int_0^t d\mu_3(\tau) D(\tau_3) \Psi_{\T{F,13}}^{s,nm}(t,\tau).
\end{align*}
Using \eqref{energy_difference} in combination with \eqref{lem:Estimates_partial_k(s)_1}, we find for $Y\in \{1,2\}$,
\begin{align}
\no  \Psi^{s,n0}_{\T{F,1Y}} (\tau)\no^2 & \lesssim \Vol{12} \sum_{(k_1, l_1) \in \I_n } \Big( \frac{1+\no \partial_{\tau_1}k_{k_1l_1}(\tau_1)\varphi_0\no }{  E_{l_1}-E_{k_1}  } \Big) \sum_{( k_2 ,l_2)\in \I_0 } \sum_{k_3=1}^N \sum_{l_3=N+1}^\infty \vert \hat v_{k_3l_3}\vert \times \nonumber\\
& \hspace{-1.8cm}\times \sum_{(m_1,n_1) \in \I_n }  \Big( \frac{ 1+\no \partial_{\tau_1}k_{m_1n_1}(\tau_1)\varphi_0\no  }{ E_{n_1}-E_{m_1}  } \Big) \sum_{(m_2, n_2) \in \I_0 } \sum_{m_3=1}^N \sum_{n_3=N+1}^\infty \vert \hat v_{m_3n_3}\vert     \langle  \Omega_0^{[l_3^* k_3 l_2^* k_2 l_1^*k_1]}   ,  \Omega_0^{[n_3^* m_3 n_2^* m_2 n_1^* m_1]}  \rangle , \nonumber\\
& \lesssim  \rho^{4\varepsilon} \Big( \rho^{-\left(\frac{n-1}{cM}\right)} \V{n} \Big) \V{0} \frac{1}{L^4} \sum_{k_3=1}^N\sum_{l_3=N+1}^\infty \vert \hat v_{k_3l_3}\vert .
\end{align}
In complete analogy one finds the same bound for $\Psi_{\T{F,1Y}}^{s,0n}(\tau)$, $Y\in \{1,2\}$. Next, using \eqref{energy_difference} together with \eqref{BOUND:G:1}-\eqref{BOUND:G:3}, we find for $X\in \{1,2,3\}$,
\begin{align}
\no \Psi_{\T{F,1X}}^{s,nm}(\tau)\no^2  & \lesssim   \frac{1}{L^{12}} \sum_{(k_1, l_1) \in \I_n }     \sum_{( k_2, l_2) \in \I_m } \sum_{k_3=1}^N \sum_{l_3=N+1}^\infty \vert \hat v_{k_3l_3}\vert   \no G^{\T{(X)}}_{k_2 l_2 k_1 l_1}(t,\tau)\varphi_0 \no \times   \nonumber\\
& \hspace{-1.8cm} \times \sum_{(m_1,n_1) \in \I_n }    \sum_{( m_2,n_2 )\in \I_m } \sum_{m_3=1}^N \sum_{n_3=N+1}^\infty \vert \hat v_{m_3n_3}\vert   \no G^{\T{(X)}}_{m_2 n_2 m_1 n_1}(t,\tau)\varphi_0\no \   \langle  \Omega_0^{[l_3^* k_3 l_2^* k_2 l_1^*k_1]}   ,  \Omega_0^{[n_3^* m_3 n_2^* m_2 n_1^* m_1]}  \rangle , \nonumber\\
 & \lesssim \rho^{8\varepsilon} \Big( \rho^{-\left(\frac{n-1}{cM}\right)} \V{n} \Big) \Big( \rho^{-\left(\frac{m-1}{cM}\right)} \V{m} \Big)  \frac{1}{L^4}  \sum_{k_3=1}^N \sum_{l_3=N+1}^\infty   \vert \hat v_{k_3l_3}\vert^2  .
\end{align}
The stated estimates then follow from \eqref{lemma:v_hat_estimates_fluc} and Corollary \ref{COROLLARY}.
\end{proof}
\noindent \textbf{Bounds for} $\no \Psi_{ \T{F,2} }^{s,nm} (t) \no_{ \T{TD} }$ \textbf{and} $\no \Psi_{ \T{F,2} }^\ell (\tau) \no_{ \T{TD} }$\textbf{.} In $\Psi_{ \T{F,2} }^\ell (t)$, and similarly in $\Psi_{ \T{F,2} }^{s,nm} (t)$, we denote the four lines separately by $\Psi_{\T{F,2i}}^\ell (t)$, $i=1,2,3,4$. We derive the estimates only for the first line, whereas for $i=2,3,4$ everything works in exact analogy to the case $i=1$. Using the Kronecker deltas in ($k_2,k_1,m_2,m_1\le N$ and $l_3,l_1,n_3,n_1\ge N+1$)
\begin{align}
\langle \Omega_0^{[l_3^*k_2l_1^*k_1]} , \Omega_0^{[n_3^*m_2n_1^*m_1]} \rangle_{\mathcal H_{gas}} = & (\delta_{l_3 n_3} \delta_{l_1 n_1} + \delta_{l_3n_1} \delta_{l_1n_3} ) \delta_{l_3l_1}^\perp \delta_{n_3n_1}^\perp \times \nonumber \\
& \times (\delta_{k_2 m_2} \delta_{k_1 m_1} + \delta_{k_2 m_1} \delta_{k_1 m_2} ) \delta_{k_1k_2}^\perp \delta_{m_1m_2}^\perp,
\end{align}
one finds
\begin{align}
& \no\Psi_{\T{F,21}}^\ell (\tau)\no^2 \lesssim \Vol{4}  \sum_{k_1,k_2=1}^N \sum_{l_1,l_2=N+1}^\infty \vert \hat w^{\ell,\varepsilon}_{k_2l_2k_1l_1} \vert \Vol{2}   \sum_{l_3=N+1}^{\infty}  \vert \hat v_{l_2l_3} \vert   \Vol{2}   \sum_{n_2=N+1}^\infty  ( \vert  \hat v_{n_2l_3} \vert + \vert \hat v_{n_2l_1} \vert ).
\end{align}
By means of \eqref{lem:estimates_v_L} and also \eqref{lemma:v_hat_estimates_p}, we obtain $\no\Psi_{\T{F,21}}^\ell (\tau)\no_{\T{TD}} \le C_\varepsilon t^3 \rho^{\frac{1}{4} + 2\varepsilon -\frac{1}{2\varepsilon}} $, and hence,
\begin{align}
\bno \int_0^t d\mu_3(\tau) D(\tau_3) \Psi_{\T{F,21}}^\ell (\tau) \bno_{\T{TD}} \le C_\varepsilon t^3 \rho^{\frac{1}{4}+2\varepsilon -\frac{1}{2\varepsilon}}.
\end{align}
Similarly, we estimate in $\Psi_{\T{F,21}}^{s,00}(t)$ the norm
\begin{align}
& \bno 	\Vol{6} \sum_{ ( k_1 , l_1 ) \in \I_0} \sum_{ ( k_2 , l_2 ) \in \I_0} \sum_{l_3=N+1}^{\infty}  \hat v_{l_2l_3} \hat v^{s,\varepsilon}_{k_2l_2}   \hat v^{s,\varepsilon}_{k_1l_1}   \Big( g_{l_2l_3}(\tau_3) g_{k_2l_2}(\tau_2) g_{k_1l_1}(\tau_1)  \varphi_0  \Big) \otimes \Omega_0^{[l_3^* k_2 l_1^* k_1]}  \bno^2 \nonumber \\ 
& \hspace{2.5cm}\lesssim \Vol{4} \sum_{ ( k_1 , l_1 ) \in \I_0} \Vol{4} \sum_{ ( k_2 , l_2 ) \in \I_0}  \frac{1}{L^2} \sum_{l_3=N+1}^{\infty}  \vert \hat v_{l_2l_3} \vert  \frac{1}{L^2}   \sum_{n_2=N+1}^{\infty} \Big( \vert \hat v^{s,\varepsilon}_{k_1n_2}\vert + \vert \hat v^{s,\varepsilon}_{k_2 n_2}\vert \Big) .
\end{align}
Using \eqref{lemma:v_hat_estimates_p} in combination with \eqref{Estimate:A0}, this leads to
{\begin{align}
\no \Psi_{\T{F,21}}^{s,00}(t)\no_{\T{TD}} & \lesssim t^3 \rho^{-\frac{1}{2}+\varepsilon} \Big(\rho^{2\varepsilon} +  C_\varepsilon \rho^{-1/ \varepsilon }\Big).
\end{align}
\begin{lemma}
\label{lem:Psi_FS2(t)} Let $\Psi_{\T{F,2}}^{s,nm}(t)$ be defined as in \eqref{PSI:FS2}. Then, under the same assumptions as in Theorem \ref{thm:main_thm}, there exists a positive constant $C$ such that
\end{lemma}
\vspace{-0.7cm}
\begin{align}
\no \Psi_{\T{F,2}}^{s,0n}(t) \no_{\T{TD}} + \no \Psi_{\T{F,2}}^{s,n0}(t) \no_{\T{TD}} & \lesssim (1+t)^3 \rho^{-\frac{1}{2}+  3\varepsilon +\frac{1}{2cM} }  \Big(\rho^{2\varepsilon} + C_\varepsilon \rho^{-1/\varepsilon }\Big)	 , \hspace{0.5cm} 1\le n \le M,  \\
\no \Psi_{\T{F,2}}^{s,nm}(t) \no_{\T{TD}}   & \lesssim (1+t)^3 \rho^{-\frac{1}{2} + 5\varepsilon + \frac{1}{cM}}  \Big(\rho^{2\varepsilon} + C_\varepsilon\rho^{-1/\varepsilon }\Big),  \hspace{0.5cm} 1\le n,m \le M,
\end{align}
\textit{holds for all $t\ge 0$.}\\

For $M=\floor{\ln \rho}$, one obtains similar as before
\begin{align}
\sum_{n=1}^{M}  \Big( \no \Psi_{\T{F,2}}^{s,n0}(t) \no_{\T{TD}} + \no \Psi_{\T{F,2}}^{s,0n}(t) \no_{\T{TD}} \Big) & \lesssim (1+t)^3 \rho^{-\frac{1}{2}+4\varepsilon} \Big(\rho^{2\varepsilon} + C_\varepsilon \rho^{-1 / \varepsilon }\Big), \\
\sum_{n,m=1}^{M} \no \Psi_{\T{F,2}}^{s,nm}(t) \no_{\T{TD}}&  \lesssim (1+t)^3  \rho^{-\frac{1}{2}+6\varepsilon} \Big(\rho^{2\varepsilon} + C_\varepsilon \rho^{-1/ \varepsilon }\Big).
\end{align}
\begin{proof}[Proof of Lemma \ref{lem:Psi_FS2(t)}]
We denote the four lines in $\Psi_{\T{F,2}}^{s,nm}(t)$ respectively by $\Psi_{\T{F,2i}}^{s,nm}(t)$, $i=1,...,4$. We prove the Lemma for the first line. The same estimates are readily verified for the other three lines as well. Let us define
\begin{align}
\Psi_{\T{F,211}}^{s,n0}(t,\tau) &=  \frac{1}{L^6} \sum_{ (k_1,l_1) \in \I_n} \sum_{ ( k_2, l_2) \in \I_0} \sum_{l_3=N+1}^{\infty}  \hat v_{l_2 l_3} \hat v^{s,\varepsilon}_{k_2l_2}  \hat v^{s,\varepsilon}_{k_1l_1} \times \\
& \hspace{4.5cm} \times \Big( g_{l_2 l_3}(\tau_2) g_{k_2l_2}(\tau_1) \frac{ g_{k_1l_1}(t) - g_{k_1l_1}(\tau_1) }{i (E_{l_1}-E_{k_1} ) }  \varphi_0  \Big) \otimes  \Omega_0^{[l_3^*k_2 l_1^*k_1  ]},\nonumber\\
\Psi_{\T{F,212}}^{s,n0}(\tau) &=  \frac{1}{L^6} \sum_{ ( k_1,l_1) \in \I_n} \sum_{ ( k_2, l_2) \in \I_0} \sum_{l_3=N+1}^{\infty}  \hat v_{l_2 l_3} \hat v^{s,\varepsilon}_{k_2l_2} \hat v^{s,\varepsilon}_{k_1l_1} \times   \\
& \hspace{4.5cm}\times \Big( g_{l_2l_3}(\tau_3)  g_{k_2l_2}(\tau_2) \frac{ e^{i (E_{l_1}-E_{k_1} ) \tau_1}   \partial_{\tau_1} k_{k_1l_1}(\tau_1) }{i (E_{l_1}-E_{k_1} ) }  \varphi_0 \Big) \otimes  \Omega_0^{[l_3^*k_2 l_1^*k_1  ]},\nonumber\\
\Psi_{\T{F,211}}^{s,0n}(\tau) &=  \frac{1}{L^6} \sum_{ ( k_1, l_1 ) \in \I_0} \sum_{ ( k_2 , l_2 ) \in \I_n} \sum_{l_3=N+1}^{\infty}  \hat v_{l_2 l_3}   \hat v^{s,\varepsilon}_{k_2l_2}  v^{s,\varepsilon}_{k_1l_1}  \times    \\
&\hspace{4.5cm} \times \Big(  g_{l_2l_3}(\tau_2)  \frac{ g_{k_2l_2}(\tau_1) - g_{k_2l_2}(\tau_2) }{i (E_{l_2}-E_{k_2} ) }  g_{k_1l_1}(\tau_1) \varphi_0 \Big) \otimes \Omega_0^{[l_3^*k_2 l_1^*k_1  ]},\nonumber\\
\Psi_{\T{F,212}}^{s,0n}(\tau) &=  \frac{1}{L^6} \sum_{ ( k_1, l_1 ) \in \I_0} \sum_{ (k_2, l_2) \in \I_n} \sum_{l_3=N+1}^{\infty}  \hat v_{l_2 l_3}  \hat v^{s,\varepsilon}_{k_2l_2}  v^{s,\varepsilon}_{k_1l_1} \times   \\
& \hspace{4.5cm} \times \Big(  g_{l_2l_3}(\tau_3)\frac{ e^{i ( E_{l_2}-E_{k_2} ) \tau_2 }  \partial_{\tau_2} k_{k_2l_2}(\tau_2) }{i ( E_{l_2}-E_{k_2} )}  g_{k_1l_1}(\tau_1)  \varphi_0 \Big) \otimes \Omega_0^{[l_3^*k_2 l_1^*k_1  ]}.\nonumber
\end{align}
By partial integration,
\begin{align}
 \Psi_{\T{F,21}}^{s,n0}(t) & = \int_0^t d\mu_2(\tau) D(\tau_2) \Psi_{\T{F,211}}^{s,n0}(t,\tau) - \int_0^t d\mu_3(\tau) D(\tau_3) \Psi_{\T{F,212}}^{s,n0}(\tau), \\
 \Psi_{\T{F,21}}^{s,0n}(t)& = \int_0^t d\mu_2(\tau) D(\tau_2) \Psi_{\T{F,211}}^{s,0n}(\tau) - \int_0^t d\mu_3(\tau) D(\tau_3) \Psi_{\T{F,212}}^{s,0n}(\tau).
\end{align}
We set further, with $G^{\T{(X)}}_{k_2l_2k_1l_1}(t,\tau)$, $X\in \{1,2,3\}$ as in \eqref{DEF:G:1}-\eqref{DEF:G:3},
\begin{align*}
\Psi_{\T{F,21X}}^{s,nm}(t,\tau) & =   \frac{1}{L^6} \sum_{ ( k_1,l_1) \in \I_n} \sum_{ ( k_2,l_2 ) \in \I_m} \sum_{l_3=N+1}^{\infty}  \hat v_{l_2 l_3 } \hat v^{s,\varepsilon}_{k_2l_2}  \hat v^{s,\varepsilon}_{k_1l_1} \Big( g_{l_2l_3}(\tau_{\T{X}})  G^{\T{(X)}}_{k_2l_2k_1l_1}(t,\tau) \varphi_0 \Big)  \otimes \Omega_0^{[l_3^*k_2 l_1^*k_1  ]}. \nonumber
\end{align*}
By partial integration again,
\begin{align*}
\Psi_{\T{F,21}}^{s,nm}(t) = \int_0^t d\tau_1 D(\tau_1) 
\Psi_{\T{F,211}}^{s,nm}(t,\tau) + \int_0^t d\mu_2(\tau) D(\tau_2)
 \Psi_{\T{F,212}}^{s,nm} (t,\tau)  + \int_0^t d\mu_3(\tau)D(\tau_3) \Psi_{\T{F,213}}^{s,nm}(t,\tau).
\end{align*}
Next, we compute using \eqref{lem:Estimates_partial_k(s)_1}, for $Y\in \{1,2\}$,
\begin{align*}
\no \Psi_{\T{F,21Y}}^{s,n0}(\tau)\no^2 & \lesssim \rho^{4\varepsilon} \Vol{4} \sum_{(k_1,l_1)\in \I_n} \rho^{-\left(\frac{n-1}{c M}\right)} \Vol{4} \sum_{(k_2,l_2)\in \I_0}  \frac{1}{L^{2}} \sum_{l_3=N+1}^{\infty}  \vert \hat v_{l_3l_2} \vert  \frac{1}{L^{2}} \sum_{n_2=N+1}^{\infty} \Big(\vert \hat v_{l_3n_2}\vert + \vert \hat v_{l_1n_2} \vert \Big) .
\end{align*}
Similarly, one derives the same estimate for $\no \Psi_{\T{F,21Y}}^{s,0n}(\tau)\no$. Furthermore, using \eqref{BOUND:G:1}-\eqref{BOUND:G:3}, one finds
\begin{align}
\no \Psi_{\T{F,21X} }^{s,nm}(\tau)\no^2 & \lesssim \rho^{8\varepsilon} \Vol{4} \sum_{(k_1,l_1)\in \I_n} \rho^{-\left(\frac{n-1}{c M}\right)}  \Vol{4} \sum_{(k_2,l_2)\in \I_m} \rho^{-\left(\frac{m-1}{c M}\right)} \times \\
& \hspace{4cm} \times \frac{1}{L^{2}}  \sum_{l_3=N+1}^{\infty}  \vert \hat v_{l_3l_2} \vert  \frac{ 1 }{L^{2}}  \sum_{n_2=N+1}^\infty  \Big(\vert \hat v_{l_3n_2}\vert + \vert \hat v_{l_1n_2} \vert \Big).
\end{align}
The proof of the lemma then follows from \eqref{lemma:v_hat_estimates_p}, \eqref{Estimate:A0} and \eqref{Estimate:An}.
\end{proof}

\noindent \textbf{Bounds for} $\no \Psi_{\T{F,3}}^{s,nm}(t) \no_{\T{TD}}$ \textbf{and} $\no \Psi_{\T{F,3}}^\ell (\tau) \no_{\T{TD}}$\textbf{.} We denote the four different lines by $\Psi_{\T{F,3i}}^{s,nm}(t)$, respectively $\Psi_{\T{F,3i}}^\ell (\tau)$ and derive the bounds only for $i=1$ since for $i=2,3,4$, the same estimates are derived analogously.\ Using $\lsp \Omega_0^{[l_1^*k_1]},  \Omega_0^{[n_1^*m_1]} \rsp = \delta_{l_1 n_1} \delta_{k_1 m_1}$ for $k_1,m_1\le N$ and $N+1\le l_1,n_1$, we find
\begin{align}
\no \Psi_{\T{F,31}}^\ell	(\tau) \no^2 & \lesssim  \Vol{8}  \sum_{k_1,k_2=1}^N  \sum_{l_1,l_2=N+1}^\infty  \vert \hat w^{\ell,\varepsilon}_{k_2l_2k_1l_1}  \vert \Vol{4} \sum_{m_2=1}^N \sum_{n_2=N+1}^\infty \vert \hat v_{n_2m_2} \vert.
\end{align}
By \eqref{lemma:v_hat_estimates_fluc} and \eqref{lem:estimates_v_L}, it follows that $\no \Psi_{\T{F,31}}^\ell	(\tau) \no_{\T{TD}} \le C_\varepsilon \rho^{\frac{1}{2}-\frac{1}{2\varepsilon}}$, and thus,
\begin{align}
\bno \int_0^t d\mu_3(\tau) D(\tau_3) \Psi_{\T{F31}}^\ell	(\tau) \bno_{\T{TD}} \le C_\varepsilon t^3 \rho^{\frac{1}{2}-\frac{1}{2\varepsilon}}.
\end{align}
Similarly, in $\Psi_{\T{F,31}}^{s,00}(t)$, we estimate the norm
\begin{align}
& \bno \Vol{6}  \sum_{ (k_1, l_1) \in \I_0} \sum_{ (k_2 , l_2 ) \in \I_0} \hat v_{l_2k_2}  \hat v^{s,\varepsilon}_{k_2l_2}  \hat v^{s,\varepsilon}_{k_1l_1} \Big( g_{l_2k_2}(\tau_3) g_{k_2l_2}(\tau_2) g_{k_1l_1}(\tau_1) \varphi_0 \Big) \otimes \Omega_0^{[ l_1^* k_1]} \bno^2 \lesssim \V{0}^3.
\end{align}
Hence, by means of \eqref{Estimate:A0}, we find
\begin{align}
\no \Psi_{\T{F31}}^{s,00}(t)\no_{\T{TD}} \lesssim Ct^3 \rho^{-\frac{3}{4}+\frac{3}{2}\varepsilon}.
\end{align}
  
\begin{lemma}
\label{lem:Psi_FS3(t)} Let $\Psi_{\T{F,3} }^{s,nm}(t)$ be defined as in \eqref{PSI:F:S3}. Then, under the same assumptions as in Theorem \ref{thm:main_thm}, there exists a positive constant $C$ such that
\end{lemma}
\vspace{-0.7cm}
\begin{align}
\no  \Psi_{\T{F,3}}^{s,0n}(t) \no_{\T{TD}} + \no \Psi_{\T{F,3}}^{s,n0}(t) \no_{\T{TD}} & \lesssim (1+t)^3 \rho^{-\frac{3}{4}+  4\varepsilon +\frac{1}{2cM} } \Big(\rho^{2\varepsilon} + C_\varepsilon \rho^{-1/ \varepsilon }\Big)   , \hspace{0.5cm} 1\le n \le M, \\
\no \Psi_{\T{F,3}}^{s,nm}(t) \no_{\T{TD}}   & \lesssim (1+t)^3 \rho^{-\frac{3}{4} + 6\varepsilon  +\frac{3}{2cM}}  \Big(\rho^{2\varepsilon} + C_\varepsilon\rho^{-1 / \varepsilon }\Big)  ,  \hspace{0.5cm} 1\le n,m \le M,
\end{align}
\textit{holds for all $t\ge 0$.}\\

With $M=\floor{\ln \rho}$,
\begin{align}
 \sum_{n=1}^{M}  \Big( \no \Psi_{\T{F,3}}^{s,n0}(t) \no_{\T{TD}} + \no \Psi_{\T{F,3}}^{s,0n}(t)\no_{\T{TD}} \Big) & \lesssim  (1+t)^3 \rho^{-\frac{3}{4}+ 5 \varepsilon } \Big(\rho^{2\varepsilon } + C_\varepsilon \rho^{-1/ \varepsilon }\Big), \\
 \sum_{n,m=1}^{M}  \no \Psi_{\T{F,3}}^{s,nm}(t) \no_{\T{TD}} & \lesssim (1+t)^3 \rho^{-\frac{3}{4}+7\varepsilon } \Big(\rho^{2\varepsilon} + C_\varepsilon \rho^{-1 / \varepsilon }\Big).
\end{align}
\begin{proof}[Proof of Lemma \ref{lem:Psi_FS3(t)}] Again, we prove the lemma only for the $i=1$ term. Let
\begin{align*}
\Psi_{\T{F,311}}^{s,n0}(t,\tau) & =   \Vol{6}  \sum_{ (k_1, l_1 ) \in \I_n} \sum_{ ( k_2 , l_2 ) \in \I_0} \hat v_{l_2k_2}  \hat v^{s,\varepsilon}_{k_2l_2}  \hat v^{s,\varepsilon}_{k_1l_1}    \Big( g_{l_2k_2}(\tau_2)g_{k_2l_2}(\tau_1) \frac{ g_{k_1l_1}(t) - g_{k_1l_1}(\tau_1) }{i (E_{l_1}-E_{k_1} )}  \varphi_0 \Big) \otimes \Omega_0^{[ l_1^* k_1]} , \\
\Psi_{\T{F,312}}^{s,n0}(\tau) & =   \Vol{6}  \sum_{ ( k_1, l_1) \in \I_n} \sum_{ ( k_2, l_2) \in \I_0} \hat v_{l_2k_2}  \hat v^{s,\varepsilon}_{k_2l_2}  \hat v^{s,\varepsilon}_{k_1l_1}    \Big( g_{l_2k_2}(\tau_3)  g_{k_2l_2}(\tau_2) \frac{ e^{i ( E_{l_1}-E_{k_1}) \tau_1}   \partial_{\tau_1} k_{k_1l_1}(\tau_1) }{i ( E_{l_1}-E_{k_1} )}  \varphi_0 \Big) \otimes \Omega_0^{[ l_1^* k_1]},\\
\Psi_{\T{F,311}}^{s,0n}(\tau)&  =   \Vol{6}  \sum_{ (k_1, l_1 ) \in \I_0} \sum_{ ( k_2, l_2) \in \I_n} \hat v_{l_2k_2}  \hat v^{s,\varepsilon}_{k_2l_2}  \hat v^{s,\varepsilon}_{k_1l_1}    \Big( g_{l_2k_2}(\tau_2) \frac{ g_{k_2l_2}(\tau_1) - g_{k_2l_2}(\tau_2) }{i ( E_{l_2}-E_{k_2} ) }  g_{k_1l_1}(\tau_1)   \varphi_0 \Big) \otimes \Omega_0^{[ l_1^* k_1]},\\
\Psi_{\T{F,312}}^{s,0n}(\tau) & =   \Vol{6}  \sum_{ ( k_1, l_1 ) \in \I_0} \sum_{( k_2, l_2 ) \in \I_n} \hat v_{l_2k_2}  \hat v^{s,\varepsilon}_{k_2l_2}  \hat v^{s,\varepsilon}_{k_1l_1}    \Big( g_{l_2k_2}(\tau_3) \frac{ e^{i ( E_{l_2} -E_{k_2} )\tau_2 }  \partial_{\tau_2} k_{k_2l_2}(\tau_2) }{i ( E_{l_2}-E_{k_2} )}  g_{k_1l_1}(\tau_1) \Big) \varphi_0\otimes \Omega_0^{[ l_1^* k_1]}.
\end{align*}
Via partial integration,
\begin{align}
 \Psi_{\T{F,31}}^{s,n0}(t) & = \int_0^t d\mu_2(\tau) D(\tau_2) \Psi_{\T{F,311}}^{s,n0}(t,\tau) - \int_0^t d\mu_3(\tau) D(\tau_3) \Psi_{\T{F,312}}^{s,n0}(\tau), \\
 \Psi_{\T{F,31}}^{s,0n}(t)&  = \int_0^t d\mu_2(\tau) D(\tau_2) \Psi_{\T{F,311}}^{s,0n}(\tau) - \int_0^t d\mu_3(\tau) D(\tau_3) \Psi_{\T{F,312}}^{s,0n}(\tau).
\end{align}
We set further, for $G^{\T{(X)}}_{k_2l_2k_1l_1}(t,\tau)$, $X\in \{1,2,3\}$ as in \eqref{DEF:G:1}-\eqref{DEF:G:3},
\begin{align}
\Psi_{\T{F,31X}}^{s,nm}(t,\tau) & =    \frac{1}{L^6} \sum_{ ( k_1,l_1) \in \I_n} \sum_{ ( k_2,l_2) \in \I_m}   \hat v_{l_2 k_2 } \hat v^{s,\varepsilon}_{k_2l_2}  \hat v^{s,\varepsilon}_{k_1l_1}  \Big( g_{l_2k_2}(\tau_{\T{X}})  G^{\T{(X)}}_{k_2l_2k_1l_1}(t,\tau) \varphi_0 \Big)  \otimes \Omega_0^{[ l_1^* k_1]}.
\end{align}
Partial integration leads again to 
\begin{align*}
\Psi_{\T{F,31}}^{s,nm}(t) = \int_0^t d\tau_1 D(\tau_1) 
\Psi_{\T{F,311} }^{s,nm}(t,\tau) - \int_0^t d\mu_2(\tau) D(\tau_2)
 \Psi_{\T{F,312}}^{s,nm} (t,\tau)  + \int_0^t d\mu_3(\tau)D(\tau_3) \Psi_{\T{F,313}}^{s,nm}(t,\tau).
\end{align*}
Similar as before, we find for $Y\in \{1,2\}$, using \eqref{lem:Estimates_partial_k(s)_1},
\begin{align}
\no \Psi_{\T{F,31Y} }^{s,n0}(\tau)\no^2 &\lesssim \rho^{4\varepsilon} \Big( \rho^{-\left(\frac{n-1}{ cM}\right)} \V{n}\Big) \V{0}^2 .
\end{align}
Analogously, one derives the same estimate for $\no \Psi_{\T{F,31Y} }^{s,0n}(\tau)\no$. Furthermore, using \eqref{BOUND:G:1}-\eqref{BOUND:G:3}, one finds
\begin{align}
\no \Psi_{\T{F,31X}}^{s,nm}(t,\tau)\no^2 & \lesssim  \rho^{8\varepsilon} \Big( \rho^{-\left(\frac{n-1}{cM}\right)} \V{n}\Big) \Big( \rho^{-\left(\frac{m-1}{cM}\right)} \V{m}\Big)^2   .
\end{align}
The stated estimates follow from \eqref{Estimate:A0} and also \eqref{Estimate:An}.
\end{proof}

This completes the proof of the bound for $\no \Psi_{\T{F}}(t) \no_{\T{TD}}$ in \eqref{lemma:main_lemma_2}.

\subsection{\label{proof:lem2.3_2.4}Proof of Lemmas \ref{lem:v_hat_estimates} and \ref{lem:Estimates_partial_k(s)}}

\begin{proof}[Proof of Lemma~\ref{lem:v_hat_estimates}] Let us first note that the choice $v\in \mathcal C^\infty_0(\mathbb R^2)$ ensures that the constant $D_p$ in \eqref{Paley_Wiener_potential} is smaller than some $C>0$ uniformly in the length $L$ of the torus.\\

We begin with the upper bound in \eqref{lemma:v_hat_estimates_fluc}. Using the Paley-Wiener Theorem, cf.\ \eqref{Paley_Wiener_potential} with $p$ s.t. $pq>3$, 
\begin{align}
\lim_{\T{TD}} \Vol{4} \sum_{k=1}^N \sum_{l=N+1}^{\infty} \Big\vert \mathcal F[v](  p_k  -  p_l ) \Big\vert^q  &  \le \lim_{\T{TD}}  \Vol{4}  \sum_{k=1}^N \sum_{l=N+1}^{\infty} \frac{D^q_p}{(1+ \vert p_{k} - p_{l} \vert)^{qp}} \nonumber\\
& = \frac{D^q_{p}}{(2\pi)^4} \int_{\vert k \vert \le k_F} d^2 k \int_{\vert l \vert \ge k_F} d^2l \frac{1}{(1+ \vert k - l \vert)^{qp}} \nonumber\\
\Big[\vert k \vert + \vert l \vert \ge \vert k + l \vert \Big]	~~~~~~~~ & \le  C_q \int_{\vert k \vert \le k_F} d^2 k \int_{\vert l \vert \ge k_F - \vert k \vert} d^2l \frac{1}{(1+ \vert l \vert)^{qp}} \nonumber\\
&\le C_q  \int_{\vert k \vert \le k_F} d^2 k \bigg[ \frac{-1}{(1+\vert l \vert )^{qp-2}} \bigg]_{k_F-\vert k\vert }^{\infty} \nonumber\\
&\le C_q k_F \bigg[ \frac{1}{(1+k_F-\vert k \vert)^{qp-3}} \bigg]_0^{k_F}, \label{UPPER:BOUND:FLUC}
\end{align}
which proves the upper bound.

To show the lower bound in \eqref{lemma:v_hat_estimates_fluc}, we assume for simplicity that $\mathcal F[v](0) > 0$ (the argument is easily adapted to the general case).\ Let us denote here $\lim_{\T{TD}}\mathcal F[v] = \mathcal F[v_{\T{TD}}]$ with $v_{\T{TD}} \in \mathcal C_0^\infty(\mathbb R^2)$. Due to continuity of $\mathcal F[v_{\T{TD}}] : \mathbb R^2 \to \mathbb R$, there is a nonempty, compact ball of some radius $r>0$, $\overline {B_r(0)} \subset \mathbb R^2$, such that $ \mathcal F[v_{\T{TD}}](k) > 0 $ for all $k\in \overline { B_r(0) }$. In particular, for given $l\in \mathbb R^2$ with $\vert l \vert \in [k_F,k_F+r/10]$, we have $ \mathcal F[v_{\T{TD}}]( k-l ) >0$ for all $k \in \overline { B_r(l) }$ with $\vert k \vert \le k_F$. Since the set
\begin{align}
	A= \Big\{ (k,l) \in \mathbb R^4 \ : \ \vert l \vert \in \Big[ k_F,k_F+r/10 \Big], k\in \overline {B_r(l)}, \vert k \vert\le k_F \Big\}
\end{align}
is nonempty and compact, there exists a nonzero minimum on $A$, $
m \equiv  \min_{(k,l)\in A} \mathcal F[v_{\T{TD}}](k-l) >0$. It is then sufficient to consider the transitions corresponding to $A$ in order to obtain the lower bound:
\begin{align}\label{v_hat_upper_bound}
\lim_{\T{TD}} \Vol{4} \sum_{k=1}^N \sum_{l=N+1}^{\infty}  \Big\vert \mathcal F[v](p_k - p_l) \Big\vert^q & = \frac{1}{(2\pi)^4} \int_{\vert k \vert \le k_F} d^2 k \int_{\vert l \vert \ge k_F} d^2 l ~ \Big\vert \mathcal F[v_{\T{TD}}] (k - l) \Big\vert^{q} \nonumber\\
&  \ge \int_{\vert k \vert \le k_F} d^2 k \int_{\vert l \vert \ge k_F} d^2l ~ \Big\vert \mathcal F[v_{\T{TD}}](k - l) \Big\vert^{q}\  \chi \big( (k,l) \in A \big) \nonumber \\
& \ge m^q \int_{\vert k \vert \le k_F} d^2 k \int_{\vert l \vert \ge k_F} d^2l ~ \chi \big( (k,l) \in A \big) \nonumber \\
& = m^q \int_{k_F}^{k_F+r/10} d \vert l\vert\ \vert l \vert \int_{\vert k \vert \le k_F} d^2 k\  \chi \big(k \in \overline {B_r(l)} \big) \nonumber\\
\Big[\text{for sufficiently large}\ k_F \Big] ~~~~~~~ & = C_q r^3 k_F.
\end{align}
\begin{remark}Along the same lines, one verifies \eqref{flucutations_d_dimension} also for $d=1$ and $d=3$.
\end{remark}

We next come to \eqref{lem:estimates_v_L}. Let $\varepsilon>0$. Applying again Paley-Wiener, this time with $p$ s.t.\ $p/2 -3>0$ and $p>\frac{1}{\varepsilon} + \frac{2}{\varepsilon^2}$, we find
\begin{align}
	\lim_{\T{TD}} \Vol{4} \sum_{k=1}^N \sum_{l=N+1}^{\infty} \Big\vert \mathcal F[v^{l,\varepsilon}] ( p_k  - p_l ) \Big\vert  & \le \lim_{\T{TD}} \Vol{4} \sum_{k=1}^N \sum_{l=N+1}^{\infty} \frac{D_p  \theta\big( \vert  p_k - p_l \vert - \rho^\varepsilon \big) }{ (1+ \vert p_k -  p_l \vert )^p } \nonumber \\
 & \le \frac{D_{p}}{(2\pi)^4}  \int_{\vert k \vert \le k_F} d^2 k \int_{\vert l \vert \ge k_F} d^2l  \frac{\theta(\vert k-l\vert - \rho^\varepsilon ) }{(1+\vert k-l\vert )^{p}} \nonumber\\
& \le \frac{D_{p}}{(2\pi)^4 \rho^{\varepsilon p/2}}  \int_{\vert k \vert \le k_F} d^2 k \int_{\vert l \vert \ge k_F} d^2l  \frac{1}{(1+\vert k-l\vert )^{\frac{p}{2}}} \nonumber \\
& \le \frac{D_{p}}{(2\pi)^4} \rho^{-\varepsilon p/2 } \rho^{\frac{1}{2} }
\end{align}
where we have used in the last step the estimate from \eqref{UPPER:BOUND:FLUC}. \eqref{lem:estimates_v_L} then follows from the choice $p>\frac{1}{\varepsilon} + \frac{2}{\varepsilon^2}$.\\

To show \eqref{lemma:v_hat_estimates_int_An}, one passes to the thermodynamic limit, and computes by direct integration (for sufficiently large $\rho$ and $\varepsilon<1/2$),
\begin{align}
\lim_{\T{TD}} \Vol{4} \sum_{ ( k,l ) \in \I_n}  &  = \lim_{\T{TD}} \Vol{4} \sum_{ k =1}^N \sum_{l=N+1}^\infty \theta\Big(\rho^\varepsilon - \vert  p_k - p_l \vert \Big)  \chi \Big( \rho^{-b_{n}} \le \vert p_l \vert - \vert p_k \vert < \rho^{-b_{n+1}}  \Big)  \nonumber  \\
& = \frac{1}{(2\pi)^4} \int_{\vert k \vert \le k_F} d^2 k \int_{\vert l \vert \ge k_F} d^2l\  \theta\Big(\rho^{\varepsilon} - \vert k -l \vert \Big) \ \chi\Big(\rho^{-b_{n}} \le \vert l \vert - \vert k \vert < \rho^{-b_{n+1}} \Big)   \nonumber\\
& \le C  \rho^{\frac{1}{2}-b_{n+1}} \rho^\varepsilon \Big( \rho^{-b_{n+1}} - \rho^{-b_n} \Big).
\end{align}
For the proof of \eqref{lemma:v_hat_estimates_Ka}, we recall the definition in \eqref{def:F} and insert $v = v^{s,\varepsilon} + v^{\ell, \varepsilon}$:
\begin{align}
\lim_{\T{TD}} \widetilde E_{re}(N,\rho) & = \lim_{\T{TD}} \Vol{4} \sum_{k=1}^N\sum_{l=N+1}^{\infty} \frac{\big\vert \mathcal F[v] ( p_k -  p_l ) \big\vert^2 }{(E_l-E_k)} \theta\Big( \vert  p_l \vert - \vert  p_k \vert  -  \rho^{-\frac{1}{2}} \Big) \nonumber \\
& \le \lim_{\T{TD}} \Vol{4} \sum_{k=1}^N\sum_{l=N+1}^{\infty} \frac{\big\vert \mathcal F[v^{s,\epsilon}] ( p_k -  p_l ) \big\vert^2 }{(E_l-E_k)} \theta\Big( \vert  p_l \vert - \vert  p_k \vert  -  \rho^{-\frac{1}{2}} \Big) \label{E:RE:S} \\
& + \lim_{\T{TD}} \Vol{4} \sum_{k=1}^N\sum_{l=N+1}^{\infty} \frac{\big\vert \mathcal F[v^{\ell,\epsilon}] ( p_k -  p_l ) \big\vert^2 }{(E_l-E_k)} \theta\Big( \vert  p_l \vert - \vert  p_k \vert  -  \rho^{-\frac{1}{2}} \Big) \label{E:RE:L}.
\end{align}
In the first line we proceed with \eqref{ORTHOGONAL} and find for any $M\ge 1$,
\begin{align}
\eqref{E:RE:S} &= \sum_{n=1}^{M} \lim_{\T{TD}}  \Vol{4}  \sum_{(k,l)\in \I_n} \frac{\big\vert \mathcal F[v^{s,\varepsilon}] ( p_k - p_l ) \big\vert^2 }{(E_l-E_k)}  \nonumber\\
&\le   C \sum_{n=1}^{M} \lim_{\T{TD}}	 \rho^{-\left(\frac{n-1}{2cM} \right)} 	\Vol{4}  \sum_{(k,l)\in \I_n}   \nonumber\\
& =   C \sum_{n=1}^{M} \rho^{-\left(\frac{n-1}{2cM} \right)}  \rho^{\frac{1}{2} + \varepsilon} \rho^{-b_{n+1}} \Big( \rho^{-b_{n+1}} - \rho^{-b_{n}} \Big)   \nonumber\\
&=   C  \rho^{-\frac{1}{2} + \varepsilon } \Big( \rho^{\frac{1}{2cM}} - 1 \Big) \sum_{n=1}^{M} \rho^{\frac{n}{2c M}} 
\le C  \rho^{-\frac{1}{2} + \varepsilon}  \rho^{\frac{M+1}{2c M} } \le C \rho^{2\epsilon},
\end{align}
where we have taken the limit $M\to \infty$ and inserted $2c = (\frac{1}{2}+\varepsilon)^{-1}$.
The second line \eqref{E:RE:L} has been estimated in \eqref{RECOLLISIONS:2}. This proves the upper bound in \eqref{lemma:v_hat_estimates_Ka}. For the lower bound, we insert again $v = v^{s,\varepsilon} + v^{\ell, \varepsilon}$,
\begin{align}
\lim_{\T{TD}} \widetilde E_{re}(N,\rho) & \ge   \lim_{\T{TD}} \Vol{4} \sum_{k=1}^N\sum_{l=N+1}^{\infty} \frac{\big\vert \mathcal F[v^{s,\varepsilon}] ( p_k -  p_l ) \big\vert^2 }{(E_l-E_k)} \theta\Big( \vert  p_l \vert - \vert  p_k \vert  -  \rho^{-\frac{1}{2}} \Big).\label{E:RE:LOWER:BOUND}
\end{align}
Since $|p_l-p_k| \leq \rho^{\varepsilon}$, we find that $|p_l|-|p_k| \leq |p_l-p_k| \leq \rho^{\varepsilon}$ and $|p_l| + |p_k| \leq 3 \rho^{\frac{1}{2}}$ (for $\varepsilon\leq 1/2$), i.e., $(E_l-E_k)^{-1} \geq \frac{1}{3} \rho^{-\varepsilon-\frac{1}{2}}$. Furthermore, note that the bound from \eqref{v_hat_upper_bound} holds also if we replace $\mathcal F[v]$ by $\mathcal F[v^{s,\varepsilon}]$ for any $\varepsilon>0$. Thus we find
\begin{align}
\eqref{E:RE:LOWER:BOUND} &\geq \frac{1}{3} \rho^{-\varepsilon-\frac{1}{2}} \lim_{\T{TD}} \Vol{4} \sum_{k=1}^N\sum_{l=N+1}^{\infty} \big\vert \mathcal F[v^{s,\varepsilon}] ( p_k -  p_l ) \big\vert^2 \theta\Big( \vert  p_l \vert - \vert  p_k \vert  -  \rho^{-\frac{1}{2}} \Big) \nonumber \\
&\geq C \rho^{-\varepsilon-\frac{1}{2}} \Big( \rho^{\frac{1}{2}} - C \Big) \nonumber \\
&\geq C \rho^{-\varepsilon}.
\end{align}
Eventually note that here we can pass to the limit $\varepsilon\to 0$ which completes the derivation of the lower bound in \eqref{lemma:v_hat_estimates_Ka}.
\\

The proof of the last bound follows immediately from the decomposition of the potential; cf.\ \eqref{potential_large} and \eqref{potential_small}.
\end{proof}

\begin{proof}[Proof of Lemma \ref{lem:Estimates_partial_k(s)}] 
We prove only \eqref{lem:Estimates_partial_k(s)_1}, since \eqref{lem:Estimates_partial_k(s)_2}-\eqref{lem:Estimates_partial_k(s)_4} are derived in complete analogy.
\begin{align}
\bno  \partial_{\tau_1} k_{k l }(\tau_1) \varphi_0 \bno & = \bno  (\partial_\tau e^{iH_y^f\tau_1 }) e^{i(p_l - p_k)\cdot y} \varphi_{\tau_1}^f + e^{iH_y^f \tau_1} e^{i (p_{k} - p_{l})\cdot y} \partial_{\tau_1} \varphi_{\tau_1}^f \bno	\nonumber\\
& = \bno  \vert p_{k}-p_{l}\vert^2 e^{iH_y^f \tau_1 } e^{i(p_l - p_k)\cdot y} \varphi_{\tau^f_1} - 2(p_l-p_k )\cdot e^{i H_y^f \tau_1} e^{i (p_l - p_k)\cdot y} \nabla_y \varphi_{\tau_1}^f \bno \nonumber \\
&\le \vert p_{k}-p_{l}\vert^2 + C \vert p_{k} - p_{l }\vert 
\end{align}
because $\no \nabla_y \varphi_\tau^f \no = \no \nabla_y \varphi_0\no \le C$ (uniformly in $\rho$). The estimate follows since $\vert p_{k}-p_{l} \vert \le \rho^\varepsilon$ in $\I_n$ for all $0\le n \le M$. 
\end{proof}

\appendix 

\section{The Model in One Dimension\label{app:one_dimension}}

The main difference in the definition of the model in one spatial dimension is that the possible momenta for $L<\infty$ are now given by $p \in (2\pi / L) \mathbb Z$, and that the Fermi momentum $\vert p_N \vert = k_F$ is proportional to $\rho$. Below, we are going to prove the following theorem which is the analogous statement to Theorem \ref{thm:main_thm} (a slightly different statement implying the same result as Theorem \ref{thm:main_thm:one_dimension} was proven in \cite{jeblick:2013}).

\begin{theorem}\label{thm:main_thm:one_dimension} 
Let $d=1$, the masses $m_x=m_y=1/2$ and the coupling constant $g=1$. Let $\varphi_0 \in \mathcal H_y$ with $\no \nabla^4 \varphi_0\no \le C$ uniformly in $\rho$. Then, for any small enough $\varepsilon>0$, there exists a positive constant $C_\varepsilon$ such that
\begin{equation}\label{main_estimate:1d}
\lim_{\substack{N,L\to\infty \\ \rho=N/L=const.}} \left\| e^{-iHt} \Psi_0 - e^{-i H^{\text{mf}} t } \Psi_0 \right\|_{\mathcal H_y\otimes \mathcal H_N} \leq   C_\varepsilon (1+t)^{\frac{3}{2}} \rho^{-\frac{1}{4}+\varepsilon}
\end{equation}
holds for all $t>0$, where
\begin{equation}
\label{Mean-field Hamiltonian_1d} H^{\text{mf}} = - \Delta_y - \sum_{i=1}^N \Delta_{x_i} + \rho \mathcal F[v](0)
\end{equation}
is the free Hamiltonian with constant mean field $\langle \Omega_0, \sum_{i=1}^N v(x_i-y) \Omega_0 \rangle_{\mathcal H_N} = \rho \mathcal F[v](0)$.
\end{theorem}
\begin{remark} Note the two differences compared to Theorem \ref{thm:main_thm}: the absence of an additional next-to-leading order energy correction in $H^{\text{mf}}$ and the better error on the r.h.s.\ \end{remark}
\begin{remark} As explained in Section \ref{sec:dimension}, we expect, and this is in contrast to $d=2$, that the l.h.s.\ of \eqref{main_estimate:1d} is small for large $\rho$ on all time scales. Theorem \ref{thm:main_thm:one_dimension} can prove this only to some extent since the error term on the r.h.s.\ becomes small only as long as $t\ll \rho^{\frac{1}{6}-2\varepsilon/3}$. 
\end{remark}
One possibility to prove \eqref{main_estimate:1d} is to adapt the proof of Theorem \ref{thm:main_thm}. For that, note that the argument depends on the dimension essentially through Lemma \ref{lem:v_hat_estimates} and Corollary \ref{COROLLARY}. The corresponding bounds for $d=1$ are summarized in
\begin{lemma}Let $d=1$, $0<\varepsilon <1/2$ and $M,q\in \mathbb N$. Let $v(x)\in C^{\infty}_0(\mathbb T)  \cap C^{\infty}_0(\mathbb R)$  and $v^{l,\varepsilon}$, $v^{s,\varepsilon}$ defined as in \eqref{potential_large},\eqref{potential_small}. Then there exist positive constants $C$, $C_q$, $C_{q,\varepsilon}$ such that 
\end{lemma}
\vspace{-0.7cm}
\begin{align}
\label{lemma:v_hat_estimates_fluc_1d}\lim_{\T{TD}}     \Vol{2} \sum_{k=1}^N \sum_{l=N+1}^{\infty} \Big\vert \mathcal F[v] ( p_k - p_l ) \Big\vert^q & =   C_q,\\
\label{lem:estimates_v_L_1d} \lim_{\T{TD}}   \Vol{2} \sum_{k=1}^{N} \sum_{l=N+1}^{\infty} \Big\vert \mathcal F[v^{l,\epsilon}] (p_k  - p_l ) \Big\vert^q & \le C_{q,\varepsilon} \rho^{-1/\varepsilon}  ,\\
\lim_{\T{TD}} \V{0} & \le C\rho^{-1},\\
\lim_{\T{TD}} \Big( \rho^{-\left( \frac{n-1}{2cM} \right) }\V{n} \Big) & \le 	C \rho^{-1+ \varepsilon } \Big( \rho^{\frac{1}{2cM}} - \rho^{\frac{1}{cM}} \Big), ~~ 1\le n \le M  \\
\label{lemma:v_hat_estimates_p_1d}\lim_{\T{TD}}   \Vol{} \sum_{k=1}^N \Big\vert \mathcal F[v]  (   p_k  -    p ) \Big\vert  & \le  C \rho^{\varepsilon} + C_\varepsilon \rho^{-1/ \varepsilon } \ \ \ \ \text{for}\ \ \ p \in (2\pi/L)\mathbb Z.
\end{align}
The proof is analogous to the ones for Lemma \ref{lem:v_hat_estimates} and Corollary \ref{COROLLARY}.\\

The only bound that remains to be shown is the one for $\Psi_{\T{B,2}}^{s,n}(t)$, $1\le n\le M$; cf.\ Section \ref{bound:Lemma1.1.b}. In the two-dimensional case, this term was directly canceled by $E^{\varepsilon}_{re}(\rho)$ which is identically zero for $d=1$. However, one easily verifies, using $E_{l_1}-E_{k_1} \ge C k_F \rho^{-b_n} = C \rho^{-\frac{1}{2} + \frac{n-1}{2cM}}$ (since $k_F\propto \rho$), that in one dimension,
\begin{align}
\sum_{n=1}^M \no \Psi_{\T{B,2}}^{s,n}(t)\no_{\T{TD}} \lesssim (1+t) \rho^{-\frac{1}{2}+\varepsilon} M \rho^{\frac{1}{cM}} \lesssim (1+t) \rho^{-\frac{1}{2}+2\varepsilon},
\end{align}
since $M \rho^{\frac{1}{cM}}\lesssim \rho^\varepsilon$ for $M=\floor{\ln \rho}$. This completes the proof of Theorem \ref{thm:main_thm:one_dimension}.

\section{The Model in Three Dimensions\label{app:three_dimensions}}

Let us explain why it is not possible to adapt the argument also to the case $d=3$. Here, the possible momenta are given by $p\in (2\pi/L)\mathbb Z^3$ and $\vert p_N\vert = k_F \propto \rho^{\frac{1}{3}}$. We exemplify this for one particular term, namely
\begin{align} \label{EXAMPLE:3d}
\lno  \sum_{n=1}^M \Psi_{\T{A}}^{s,n}(t) \rno^2 = \sum_{n=1}^M \no \Psi_{\T{A}}^{s,n}(t) \no^2_{\T{TD}} \le \sum_{n=1}^M \bigg[\lim_{\T{TD}}  \Vol{6} \sum_{(k,l)\in \I_n} \frac{1}{(E_{k}-E_{l})^2} \bigg],
\end{align}
which appears at second order in the Duhamel expansion; cf.\ Section \ref{sec:Derivation_A}.\ Here, we have used in addition that the $\Psi_{\T{A}}^{s,n}(t)$ are pairwise orthogonal, and then we applied the first step from \eqref{EXAMPLE:THREE:D}. In order to obtain the optimal bound for the r.h.s., let us be more general as in the case $d=2$ and define the sets $\I_n$ with $b_0=\infty$ and $b_n = b + (n-1)/(2cM)$, for $b>0$ and $2c=\left(b+\varepsilon \right)^{-1}$, $b\in \mathbb R$. Using $(E_l-E_k) \ge \rho^{\frac{1}{3} -b_n}$ in $\I_n$ together with
\begin{align}
\lim_{\T{TD}} \Vol{6} \sum_{(k,l)\in \I_n} \lesssim \rho^{\frac{2}{3}-b_{n+1}} \rho^{2\varepsilon} \Big( \rho^{-b_{n+1}} - \rho^{-b_{n}} \Big),
\end{align} 
one finds  
\begin{align}
\Vol{6} \sum_{(k,l)\in \I_n} \frac{1}{(E_{k}-E_{l})^2}  \lesssim \rho^{2b_n-2b_{n+1}} \Big( 1- \rho^{-b_n + b_{n+1}} \Big) \lesssim \frac{1}{M} \ln \rho \Big( 1+ \mathcal O(M^{-1} \ln \rho) \Big),
\end{align}
$1\le n \le M$. Hence, this way (taking $M\to\infty$) we obtain at best  $\eqref{EXAMPLE:3d} \lesssim \ln \rho$, which would imply a trivial statement like $\no \Psi_{\T{A}}(t) \no_{\T{TD}} \lesssim \ln\rho$ already for $t$ of order one.

\bigskip \noindent \textit{Acknowledgments.}\ We are grateful to Dirk-Andr\'e Deckert, Detlef D\"urr and Herbert Spohn for many helpful discussions.\ M.J. and D.M. gratefully acknowledge financial support from the German Academic Scholarship Foundation.\ S.P.'s research has received funding from Cusanuswerk, the People Programme (Marie Curie Actions) of the European Union's Seventh Framework Programme (FP7/2007-2013) under REA grant agreement n\textdegree~291734, and the German Academic Exchange Service (DAAD). S.P. gratefully acknowledges support from the Institute for Advanced Study. This material is based upon work supported by the National Science Foundation under agreement No.\ DMS-1128155. Any opinions, findings and conclusions or recommendations expressed in this material are those of the authors and do not necessarily reflect the views of the National Science Foundation. The content of this article is part of one of the authors (D.M.) forthcoming PhD thesis.

\bibliographystyle{plain}

\end{document}